\documentclass[aps,pra,twocolumn,superscriptaddress]{revtex4-1}

\usepackage[utf8]{inputenc}  
\usepackage[T1]{fontenc}    
\usepackage{hyperref}       
\usepackage{url}          
\usepackage{booktabs}        
\usepackage{amsfonts}        
\usepackage{nicefrac}        
\usepackage{microtype}       
\usepackage{bm}
\usepackage{braket}
\usepackage{graphicx}
\usepackage{amsthm}
\usepackage{amssymb}
\usepackage{amsmath}
\usepackage{bbm}
\usepackage{lipsum,mwe}
\usepackage{makecell}
\usepackage{array}

\usepackage[ruled,linesnumbered]{algorithm2e}

\usepackage{appendix}

\newtheoremstyle{mytheoremstyle}  
  {\topsep}  
  {\topsep}  
  {\normalfont}  
  {}  
  {\itshape}  
  {.}  
  {5pt plus 1pt minus 1pt}  
  {{\thmname{#1}\thmnumber{ #2}\thmnote{ (#3)}}}

\theoremstyle{mytheoremstyle}

\usepackage{adjustbox}
\usepackage[normalem]{ulem}
\newtheorem{theorem}{Theorem}
\newtheorem{lemma}{Lemma}
\newtheorem{definition}{Definition}
\newtheorem{assu}{Assumption}

\newtheorem{construct}{Construction Rule}
\newtheorem{property}{Property}

\newtheorem*{theorem-non}{Theorem}
\usepackage{mathtools}
\usepackage{bbm}
\usepackage[explicit]{titlesec}
\usepackage{nccmath}
\usepackage{xcolor}

\usepackage[ruled,linesnumbered]{algorithm2e}

\usepackage{tkz-tab}
\usepackage{subcaption}
\usepackage{wrapfig}
\usepackage{enumitem}
\usepackage{caption}
\usepackage[scaled=.92]{helvet}   
\DeclareMathOperator{\Tr}{Tr}

\DeclareMathOperator{\KL}{KL}
\DeclareMathOperator{\supp}{supp}

\DeclareMathOperator{\diag}{diag}

\DeclareMathOperator{\Fnorm}{F}

\DeclareMathOperator{\haar}{Haar}
\DeclareMathOperator{\swap}{SWAP}

\DeclareMathOperator{\Psd}{Psd}

\allowdisplaybreaks[4]

\begin{document}
	
	\title{Transition Role of Entangled Data in Quantum Machine Learning}
	
	\author{Xinbiao Wang}
	\affiliation{
		Institute of Artificial Intelligence, School of Computer Science, Wuhan University, Hubei 430072, China
	}
	\affiliation{
		National Engineering Research Center for Multimedia Software, Wuhan University, Hubei 430072, China
	}
	\affiliation{JD Explore Academy, Beijing 101111, China}
	
	\author{Yuxuan Du}
	\email{duyuxuan123@gmail.com}
	\affiliation{School of Computer Science and Engineering, Nanyang Technological University, Singapore 639798, Singapore}
	\affiliation{JD Explore Academy, Beijing 101111, China}

	\author{Zhuozhuo Tu}
	\affiliation{School of Computer Science, Faculty of Engineering, University of Sydney, NSW 2008, Australia}
	\author{Yong Luo}
	\email{yluo180@gmail.com}
	\affiliation{
		Institute of Artificial Intelligence, School of Computer Science, Wuhan University, Hubei 430072, China
	}
	\affiliation{
		National Engineering Research Center for Multimedia Software, Wuhan University, Hubei 430072, China
	}
	
	\author{Xiao Yuan}
	\affiliation{Center on Frontiers of Computing Studies, Peking University, Beijing 100871, China}
	\affiliation{School of Computer Science, Peking University, Beijing 100871, China}
	\author{Dacheng Tao}
	\email{dacheng.tao@ntu.edu.sg}
	\affiliation{School of Computer Science and Engineering, Nanyang Technological University, Singapore 639798, Singapore}

	\maketitle

	\noindent\textbf{ABSTRACT} 
	
	\noindent\textbf{Entanglement serves as the resource to empower quantum computing. Recent progress has highlighted its positive impact on learning quantum dynamics, wherein the integration of entanglement into quantum operations or measurements of quantum machine learning (QML) models leads to substantial reductions in training data size, surpassing a specified prediction error threshold.  However, an analytical understanding of how the entanglement degree in data affects model performance remains elusive. In this study, we address this knowledge gap by establishing a quantum no-free-lunch (NFL) theorem for learning quantum dynamics using entangled data. Contrary to previous findings, we prove that the impact of entangled data on prediction error exhibits a dual effect, depending on the number of permitted measurements. With a sufficient number of measurements, increasing the entanglement of training data consistently reduces the prediction error or decreases the required size of the training data to achieve the same prediction error. Conversely, when few measurements are allowed, employing highly entangled data could lead to an increased prediction error. The achieved results provide critical guidance for designing advanced QML protocols, especially for those tailored for execution on early-stage quantum computers with limited access to quantum resources.}
	
	\medskip
	\noindent\textbf{INTRODUCTION} 
	
	Quantum entanglement, an extraordinary characteristic of the quantum realm, drives the superiority of quantum computers beyond classical computers \cite{feynman2018simulating}. Over the past decade, diverse quantum algorithms leveraging entanglement have been designed to advance cryptography \cite{shor1999polynomial,lanyon2007experimental} and optimization \cite{deutsch1992rapid,grover1996fast,harrow2009quantum,lloyd2014quantum,du2020quantum_dp}, delivering runtime speedups over classical approaches. Motivated by the exceptional abilities of quantum computers and the astonishing success in machine learning, a nascent  frontier known as quantum machine learning (QML) has emerged \cite{schuld2015introduction, biamonte2017quantum,ciliberto2018quantum,dunjko2018machine,li2022recent,tian2023recent,cerezo2022challenges}, seeking to outperform classical models in specific learning tasks \cite{peruzzo2014variational,  moll2018quantum,havlivcek2019supervised,abbas2021power,huang2021power,liu2021rigorous,wang2021towards,du2021exploring,du2022power,du2023problem}. Substantial progress has been made in this field, exemplified by the introduction of QML protocols that offer provable advantages in terms of query or sample complexity for learning quantum dynamics \cite{huang2021information,buadescu2021improved,aharonov2022quantum,chen2022exponential,huang2022quantum, fanizza2022learning}, as a fundamental problem toward understanding the laws of nature. Most of these protocols share a common strategy to gain advantages: the incorporation of entanglement into quantum operations and measurements, leading to reduced complexity. Nevertheless, an overlooked aspect in prior works is the impact of incorporating entanglement in quantum input states, or entangled data, on the advancement of QML in learning quantum dynamics. Due to the paramount role of data in learning \cite{polyzotis2021can,jakubik2024data,jarrahi2023principles,whang2023data,zha2023data,zha2023data_survey} as well as entanglement in quantum computing, addressing this question will significantly enhance our comprehension of the capabilities and limitations of QML models. 
	
	\begin{figure*}[htbp]
	\centering
	\includegraphics[width=176mm]{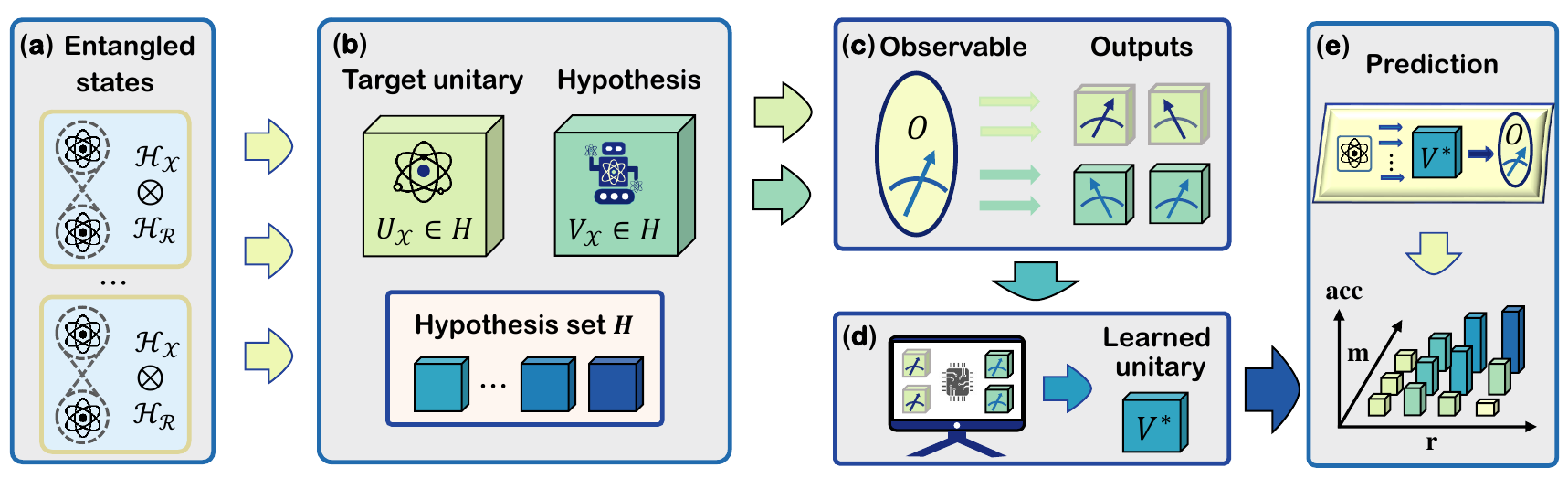}
  \caption{\fontsize{7pt}{8.4pt}{\textbf{Illustration of quantum NFL setting with the entangled data.} The goal of the quantum learner is to learn a unitary  $\bm{V}_{\mathcal{X}}$  that can accurately predict the output of the target unitary  $\bm{U}_{\mathcal{X}}$ under a fixed observable $\bm{O}$, where the subscript $\mathcal{X}$ refers to the quantum system in which the operator $\bm{O}$ act on. The learning process is as follows. (a) A total number of $N$ entangled bipartite quantum states living in Hilbert space $\mathcal{H}_{\mathcal{X}} \otimes \mathcal{H}_{\mathcal{R}}$ ($\mathcal{R}$ denotes the reference system) are taken as inputs, dubbed \textit{entangled data}. (b) Quantum learner proceeds incoherent learning. The entangled data \textit{separately} interacts with the target unitary $\bm{U}_{\mathcal{X}}$ (agnostic) and the candidate hypothesis $\bm{V}_{\mathcal{X}}$ extracted from the same Hypothesis set $H$. (c) The quantum learner is restricted to leverage the finite measured outcomes of the observable $\bm{O}$ on the output states of $\bm{U}_{\mathcal{X}}$ and $\bm{V}_{\mathcal{X}}$ to conduct learning. (d) A classical computer is exploited to infer $\bm{V}^*$ that best estimates $\bm{U}_{\mathcal{X}}$ according to the measurement outcomes. For example, in the case of variational quantum algorithms, the classical computer serves as an optimizer to update the tunable parameters of the ansatz $\bm{V}_{\mathcal{X}}$. (e) The learned unitary $\bm{V}^*$ is used to predict the output of unseen quantum states in Hilbert space $\mathcal{H}_{\mathcal{X}}$ under the evolution of the target unitary $\bm{U}_{\mathcal{X}}$ and the measurement of $\bm{O}$. A large Schmidt rank $r$ can enhance the prediction accuracy when combined with a large number of measurements $m$, but may lead to a decrease in accuracy when $m$ is small.} }
	\label{fig:scheme}
\end{figure*}
 
	A fundamental concept in machine learning that characterizes the capabilities of learning models in relation to datasets is the No-Free-Lunch (NFL) theorem \cite{wolpert1997no, ho2002simple,wolf2018mathematical,adam2019no}. The NFL theorem yields a key insight: regardless of the optimization strategy employed, the ultimate performance of models is contingent upon the size and types of training data. This observation has spurred recent breakthroughs in large language models, as extensive and meticulously curated training data consistently yield superior results \cite{brown2020language,ouyang2022training,bai2022training,touvron2023llama,zhao2023survey}. In this regard, establishing the quantum NFL theorem enables us to elucidate the specific impact of entangled data on the efficacy of QML models in learning quantum dynamics. Concretely, the achieved  theorem can shed light on whether the utilization of entangled data empowers QML models to achieve comparable or even superior performance compared to low-entangled or unentangled data, while simultaneously reducing the sample complexity required. Although initial attempts \cite{poland2020no,sharma2022reformulation,zhao2023learning} have been made to establish quantum NFL theorems, they have relied on infinite query complexity, thus failing to address our concerns adequately (See Supplementary Note~1 and Supplementary Note~2 for details). Building upon prior findings on the role of entanglement and the classical NFL theorem, a reasonable speculation is that high-entangled data contributes to the improved performance of QML models associated with the reduced sample complexity, albeit at the cost of using extensive quantum resources to prepare such data that may be unaffordable in the early stages of quantum computing \cite{preskill2018quantum}.

	In this study, we \textit{negate} the above speculation and exhibit the \textit{transition role} of entangled data when QML models incoherently learn quantum dynamics, as shown in Fig.~1. In the incoherent learning scenario, the quantum learner is restricted to utilizing datasets with varying degrees of entanglement to operate on an unknown unitary and inferring its dynamics using the finite measurement outcomes collected under the projective measurement, differing from Ref.~\cite{sharma2022reformulation} in learning problems and training data. The entangled data refers to quantum states that are entangled with a reference system, with the degree of entanglement quantitatively characterized by the Schmidt rank $r$. We rigorously show that within the context of NFL, the entangled data has a \textit{dual effect} on the prediction error according to the number of measurements $m$ allowed. Particularly, with sufficiently large $m$, increasing $r$ can consistently reduce the required size of training data for achieving the same prediction error. On the other hand, when $m$ is small, the train data with large $r$ not only requires a significant volume of quantum resources for states preparation, but also amplifies the prediction error. As a byproduct, we prove that the lower bound of the query complexity for achieving a sufficiently small prediction error matches the optimal lower bound for quantum state tomography with nonadaptive measurements.  To cover a more generic learning setting, we consider the problem of dynamic learning under arbitrary observable by using $\ell$-outcome positive-operator valued measure (POVM) to collect the measurement output. This setting covers the shadow-based learning models \cite{huang2021information, huang2020predicting,elben2023randomized}. We show that the transition role still holds for arbitrary POVM and increasing the possible outcomes $\ell$ could significantly reduce the query complexity. Numerical simulations are conducted to support our theoretical findings. In contrast to the previous understanding that entanglement mostly confers benefits to QML in terms of sample complexity, the transition role of entanglement identified in this work deepens our comprehension of the relation between quantum information and QML, which facilitates the design of QML models with provable advantages.

	\medskip	
	\noindent\textbf{RESULTS} 
	
	\noindent We first recap the task of learning quantum dynamics. Let $\bm{U}\in \mathbb{S}\mathbb{U}(2^n)$ be the target unitary and $\bm{O}\in \mathbb{C}^{2^n \times 2^n}$ be the observable which is a Hermitian matrix acting on an $n$-qubit quantum system. Here we specify the observable as the projective measurement $\bm{O}=\ket{\bm{o}}\bra{\bm{o}}$ since any observable reads out the classical information from the quantum system via their eigenvectors. The goal of the quantum dynamics learning is to predict the functions of the form
	\begin{equation}\label{eq:learning_model}
		\mathrm{f}_{\bm{U}}(\bm{\psi}) = \Tr(\bm{O}\bm{U}\ket{\bm{\psi}}\bra{\bm{\psi}}\bm{U}^{\dagger}),
	\end{equation}
	where $\ket{\bm{\psi}}$ is an $n$-qubit quantum state living in a $2^n$-dimensional Hilbert space $\mathcal{H}_{\mathcal{X}}$. This task can be done by employing the training data $\mathcal{S}$ to construct a unitary $\bm{V}_{\mathcal{S}}$, i.e., the learned hypothesis has the form of $\mathrm{h}_{{\mathcal{S}}}(\bm{\psi}) = \Tr(\bm{O}\bm{V}_{\mathcal{S}}\ket{\bm{\psi}}\bra{\bm{\psi}}\bm{V}_{\mathcal{S}}^{\dagger})$, which is expected to accurately approximate $\mathrm{f}_{\bm{U}}(\bm{\psi})$ for the unseen data. While the learned unitary acts on an $n$-qubit system $\mathcal{H}_{\mathcal{X}}$, the input state  could be entangled with a reference system $\mathcal{H}_{\mathcal{R}}$, i.e.,  $\ket{\bm{\psi}}\in \mathcal{H}_{\mathcal{X}} \otimes \mathcal{H}_{\mathcal{R}}$. We suppose that all input states have the same Schmidt rank $r\in\{1, \cdots, 2^n\}$. Then the response of the state $\ket{\bm{\psi}_j}$ is given by the measurement output $\bm{o}_j=\sum_{k=1}^m \bm{o}_{jk}/m$, where $m$ is the number of measurements and $\bm{o}_{jk}$ is the output of the $k$-th measurement of the observable $\bm{O}$ on the output quantum state $(\bm{U}\otimes \mathbb{I}_{\mathcal{R}}) \ket{\bm{\psi}_j}$. In this manner, the training data  with $N$ examples takes the form $\mathcal{S}=\{(\ket{\bm{\psi}_j}, {\bm{o}}_j): \ket{\bm{\psi}_j}\in \mathcal{H}_{\mathcal{X}} \otimes \mathcal{H}_{\mathcal{R}}, \mathbb{E}[\bm{o}_{j}]=u_j \}_{j=1}^N$ with $u_j=\Tr((\bm{U}^{\dagger}\bm{O}\bm{U} \otimes \mathbb{I}_{\mathcal{R}})\ket{\bm{\psi}_j}\bra{\bm{\psi}_j})$ being the expectation value of the observable $\bm{O}$ on the state $(\bm{U}\otimes \mathbb{I}_{\mathcal{R}}) \ket{\bm{\psi}_j}$ and $N$ being the size of the training data. Notably, in quantum dynamics learning, sample complexity refers to the size of training data $N$, or equivalently, the number of quantum states in the training data; query complexity refers to the total number of queries of the explored quantum system, i.e., the production of sample complexity and the number of measurements $Nm$.
	
	The risk function is a crucial measure in statistical learning theory to quantify how well the hypothesis function $\mathrm{h}_{{\mathcal{S}}}$ performs in predicting $\mathrm{f}_{\bm{U}}$, defined as 
	\begin{equation}\label{eq:risk_function}
		\mathrm{R}_{\bm{U}}(\bm{V}_{{\mathcal{S}}}) = \int \mathrm{d} {\bm{\psi}} \left( \mathrm{f}_{\bm{U}}(\bm{\psi})-\mathrm{h}_{{\mathcal{S}}}(\bm{\psi})  \right) ^2,
	\end{equation}
	where the integral is over the uniform Haar measure $\mathrm{d} {\bm{\psi}}$ on the state space. Intuitively, $\mathrm{R}_{\bm{U}}(\bm{V}_{{\mathcal{S}}})$ amounts to the average square error distance between the true output $\mathrm{f}(\bm{\psi})$ and the hypothesis output $\mathrm{h}_{\mathcal{S}}(\bm{\psi})$. Moreover, we follow the treatments in Ref.~\cite{sharma2022reformulation} choosing the Haar unitary as the target unitary. Additionally, we construct a sampling rule of the training input states which approximates the uniform distribution of all entangled states with Schmidt rank $r$ (Refer to Supplementary Note~2).

	Under the above setting, we prove the following quantum NFL theorem in learning quantum dynamics, where the formal statement and proof are deferred to  Supplementary Note~3. 
	
	\begin{theorem}
		[Quantum NFL theorem in learning quantum dynamics, informal] Following the settings in Eqn.~(\ref{eq:learning_model}),  suppose that the training error of the learned hypothesis on the training data $\mathcal{S}$ is less than $\varepsilon=\mathcal{O}(1/2^n)$. Then the lower bound of the averaged prediction error in Eqn.~(\ref{eq:risk_function}) yields 
		\begin{equation} 			\mathbb{E}_{\bm{U},\mathcal{S}} \mathrm{R}_{\bm{U}}(\bm{V}_{\mathcal{S}}) \ge \Omega\left(\frac{\tilde{\varepsilon}^2}{4^n}\left(1-\frac{N\cdot\min\{m/(2^nrc_1), rn \}}{2^nc_2}\right)\right), \nonumber 
		\end{equation}
		where $c_1=128/\tilde{\varepsilon}^2$, $c_2= \min\{(1-2\tilde{\varepsilon})^2, (64\tilde{\varepsilon}^2-1)^2\}$,  $\tilde{\varepsilon}=\Theta(2^n\varepsilon)$, and the expectation is taken over all target unitary $\bm{U}$, entangled states $\ket{\bm{\psi}_j}$ and measurement outputs $\bm{o_j}$. 	\label{thm:informal_finite_measurement}
	\end{theorem}

    \smallskip

    	\begin{figure*}[htbp]
	\centering
	\includegraphics[width=176mm]{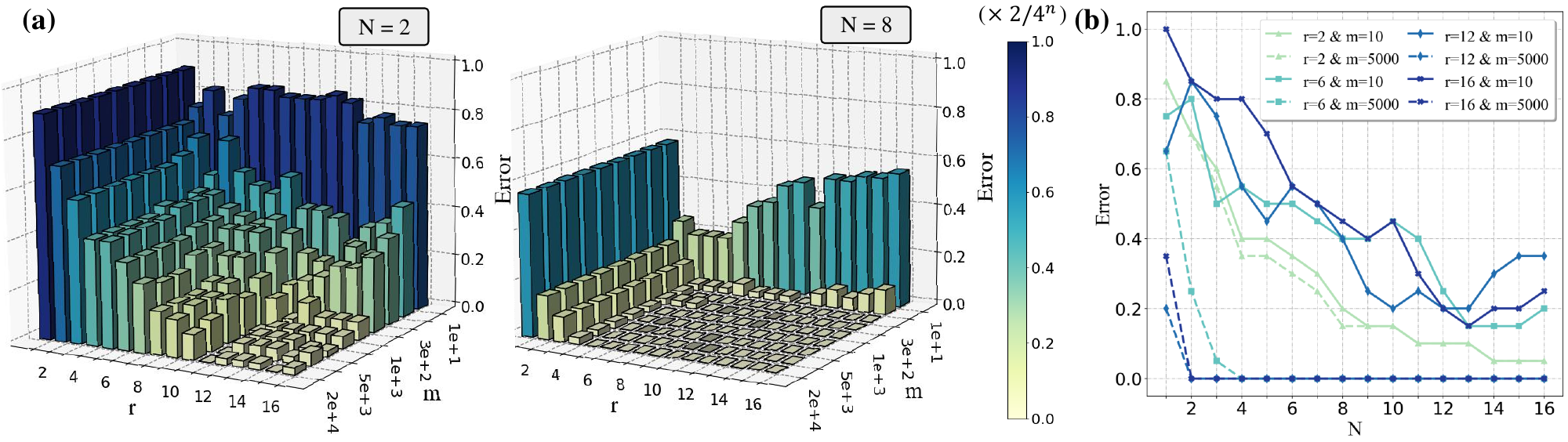}
	\caption{\fontsize{7pt}{8.4pt}{\textbf{Simulation results of quantum NFL theorem when incoherently learning quantum dynamics.}  (a) The averaged prediction error with a varied number of measurements $m$ and Schmidt rank $r$ when $N=2$ and $N=8$. The z-axis refers to the averaged prediction error defined in Eqn.~(\ref{eq:learning_model}).    (b) The averaged prediction error with the varied sizes of training data. The label `$r = a ~\&~ m = b$` refers that the Schmidt rank is $a$ and the number of measurements is $b$.  The label `($\times 2/4^n$)' refers that the plotted prediction error is normalized by a multiplier factor $2/4^n$.}}
\end{figure*}

	\noindent The achieved results indicate the transition role of the entangled data in determining the prediction error. Particularly, when a sufficient number of measurements $m$ is allowed such that the Schmidt rank $r$ obeys $r<\sqrt{m/(c_12^n n)}$, the minimum term in the achieved lower bound refers to $Nrn$ and hence increasing $r$ can constantly decrease the prediction error. Accordingly, in the two extreme cases of $r=1$ and $r=2^n$,  achieving zero averaged risk requires $N=2^nc_2/n$ and $N=1$ training input states, where the latter achieves an exponential reduction in the number of training data compared with the former. This observation implies that the entangled data empower QML with provable quantum advantage, which accords with the achieved results of Ref.~\cite{sharma2022reformulation} in the ideal coherent learning protocol with infinite measurements.
	
	By contrast, in the scenario with $r \ge \sqrt{m/(c_12^nn)}$, increasing $r$ could enlarge the prediction error. This result indicates that the entangled data can be harmful to achieving quantum advantages, which \textit{contrasts with previous results} where the entanglement (e.g., entangled operations or measurements) is believed to contribute to the quantum advantage \cite{jozsa2003role,yoganathan2019one,khatri2019quantum, sharma2022reformulation}. This counterintuitive phenomenon stems from the fact that when incoherently learning quantum dynamics, information obtained from each measurement decreases with the increased $r$ and hence a small $m$ is incapable of extracting all information of the target unitary carried by the entangled state.

	Another implication of Theorem \ref{thm:informal_finite_measurement} is that although the number of measurements $m$ contributes to a small prediction error, it is \textit{not decisive to the ultimate performance} of the prediction error. Specifically, when $m\ge r^2c_1 2^n n$, further increasing $m$ could not help decrease the prediction error which is determined by the entanglement and the size of the training data, i.e., $r$ and $N$.  Meanwhile, at least $r^2c_12^n n$ measurements are required to fully utilize the power of entangled data. These results suggest that the value of $m$ should be adaptive to $r$ to pursue a low prediction error.  
	
	We next comprehend the scenario in which the lower bound of averaged risk in Theorem~\ref{thm:informal_finite_measurement} reaches zero and correlate with the results in quantum state learning and quantum dynamics learning \cite{huang2021information, huang2022quantum, chen2022exponential, buadescu2021improved, yuen2023improved,anshu2024survey}. In particular, the main focus of those studies is proving the minimum query complexity of the target unitary to warrant zero risk. The results in Theorem~\ref{thm:informal_finite_measurement} indicate that the minimum query complexity is $Nm = {\Omega}(4^nrc_1c_2)$, implying the proportional relation between the entanglement degree $r$ and the query complexity. Notably, \textit{this lower bound is tighter than that achieved in} Ref.~\cite{huang2021information} in the same setting. The achieved results in terms of query complexity are also non-trivial, as previous works show that query complexity can benefit from using entanglement in quantum data \cite{piani2009all,bae2019more} and quantum measurements \cite{huang2021information,huang2022quantum}. The advance of our results stems from the fact that Ref.~\cite{huang2021information} simply employs Holevo's theorem to give an upper bound on the extracted information in a single measurement, while our bound integrates more refined analysis such as the consideration of Schmidt rank $r$, the direct use of a connection between the mutual information of the target unitary $\bm{U}$ and the measurement outputs $\bm{o}_j$, and the KL-divergence of related distributions (Refer to Supplementary Note~3 for more details). Moreover, the adopted projective measurement $\bm{O}$ in Eqn.~(\ref{eq:learning_model}) hints that the learning task explored in our study amounts to learning a pure state $\bm{U}^{\dagger}\bm{O}\bm{U}$. From the perspective of state learning, \textit{the derived lower bound in Theorem \ref{thm:informal_finite_measurement} is optimal for the nonadaptive measurement with a constant number of outcomes} \cite{lowe2022lower}. 
    Taken together, while the entangled data hold the promise of gaining advantages in terms of the sample complexity for achieving the same level of prediction error, they may be inferior to the training data without entanglement in terms of query complexity.

	The transition role of entanglement explained above leads to the following construction rule of quantum learning models. First,  when a large number of measurements is allowed, the entangled data is encouraged to be used for improving the prediction performance. To this end, initial research efforts \cite{wu2004preparation,basharov2006decay,lemr2008preparation,lin2016preparation,klco2020minimally,schatzki2021entangled}, which develop effective methods for preparing and storing entangled states, may contribute to QML. Second, when the total number of measurements is limited, it is advised to refrain from using entangled data for learning quantum dynamics.

	\noindent\textit{Remark.}  (i) The training error scaling $\varepsilon=\mathcal{O}(1/2^n)$ in Theorem~\ref{thm:informal_finite_measurement} and the factor of the achieved lower bound $\tilde{\varepsilon}^2/4^n$ comes from the consideration of average performance over Haar unitaries where the expectation value of observable $\bm{O}$ scales as $\Tr(\bm{O})/2^n$ (Refer to Supplementary Note~2). (ii) The results of the transition role for entangled data achieved in Theorem~\ref{thm:informal_finite_measurement} can be generalized to the mixed states because the mixed state can be produced by taking the partial trace of a pure entangled state.

    \smallskip

    In a more generic learning setting, the observable used in the target function defined in Eqn.~(\ref{eq:learning_model}) and the measurement used for collecting the response of training data $\bm{o}$ could be arbitrary and varied. In particular, we consider that the observable $\bm{O}$ defined in Eqn.~(\ref{eq:learning_model}) could be arbitrary Hermitian operator satisfying $\|\bm{O}\|_{1}\le \infty$. The response $\bm{a}_j$ for given input states $\ket{\bm{\psi}_j}$ could be obtained from measuring the output states on system $\mathcal{X}$ with $\ell$-outcome POVM. The training dataset in this case refers to $\mathcal{S}_{\ell}=\{(\ket{\bm{\psi}_j},\bm{a}_j):\ket{\bm{\psi}_j}\in \mathcal{H}_{\mathcal{X}\mathcal{R}}, \bm{a}_{j}=(\bm{a}_{j1},\cdots, \bm{a}_{jm}), \bm{a}_{jk}\in \{z_1, \cdots, z_{\ell}\}\}_{j=1}^{N}$, where $\ket{\bm{\psi}_j}$ refers to the entangled states with Schmidt rank $r$, $\bm{a}_j$ is the $m$-measurement outputs with $\ell$-outcome POVM, and $\{z_i\}_{i=1}^{\ell}$ is the $\ell$ possible outcomes of the employed POVM. In this case, denoting the learned unitary as $\bm{V}_{\mathcal{S}_{\ell}}$, we get the following quantum NFL theorem in learning quantum dynamics for generic measurement, where the formal statement and proof are deferred to Supplementary Note~4.
    \begin{theorem}
		[Quantum NFL theorem in learning quantum dynamics for generic measurement, informal] Following the settings in Eqn.~(\ref{eq:learning_model}) with arbitrary $\bm{O}$ satisfying $\|\bm{O}\|_1 \le \infty$,  suppose the learned hypothesis is learned from training data $\mathcal{S}_{\ell}$. Then the lower bound of the averaged prediction error in Eqn.~(\ref{eq:risk_function}) yields \begin{align}\label{eq:error_bound_POVM}
        \mathbb{E}_{\bm{U},\mathcal{S}_{\ell}} \mathrm{R}_{\bm{U}}(\bm{V}_{\mathcal{S}_{\ell}}) \ge {\varepsilon^2} \left(1- \frac{N\cdot \min\{4m/r, 6m\ell/2^n r, rn\}}{\log(|\mathcal{X}_{2\varepsilon}(\bm{O})|)}  \right) \nonumber
    \end{align}
    where $|\mathcal{X}_{2\varepsilon}(\bm{O})|$ refers to the model complexity and only depends on $\varepsilon$ and the employed observable $\bm{O}$. For projective measurement $\bm{O}=\ket{\bm{o}}\bra{\bm{o}}$, $\log(|\mathcal{X}_{2\varepsilon}(\bm{O})|)=2^nc_2$ is given in the denominator of the achieve lower bound in Theorem~\ref{thm:informal_finite_measurement}.
    \label{thm:informal_finite_measurement_POVM}
	\end{theorem}

    The achieved results in Theorem~\ref{thm:informal_finite_measurement_POVM} deliver three implications. First, the transition role of entangled data still holds for arbitrary observable and POVM. In particular, no matter how large the number of possible outcomes of POVM $\ell$ is, increasing the Schmidt rank will decrease the prediction error as long as the number of measurements $m$ satisfies $\min\{4m/r, 6m\ell/2^n r\} \le rn$, and increase the prediction error otherwise. Second, when the observable is projective measurement and the number of possible outcomes $\ell$ is of constant order, the achieved result in Theorem~\ref{thm:informal_finite_measurement_POVM} reduces to the results achieved in Theorem~\ref{thm:informal_finite_measurement} for the case of employing projective measurement up to a constant factor. Third, increasing the number of possible outcomes of POVM $\ell$ can exponentially reduce the number of measurements required to achieve the same level of prediction error. Particularly, considering two extreme cases of the possible outcomes of POVM $\ell$ being constant scaling $\Theta(1)$ and exponential scaling $\Theta(2^n)$, achieving the same level of prediction error requires the query complexity scaling with the order of $2^nr\log(|\mathcal{X}_{2\varepsilon}(\bm{O})|)$ and $r\log(|\mathcal{X}_{2\varepsilon}(\bm{O})|)$, where the latter case achieves an exponential reduction in terms of the query complexity.

	\smallskip
	\noindent\textbf{Numerical results.}
	We conduct numerical simulations to exhibit the transition role of entangled data, the effect of the number of measurements, and the training data size in determining the prediction error.  The omitted construction details and results are deferred to Supplementary Note~5.
	
	We focus on the task of learning an $n$-qubit unitary under a fixed projective measurement $\bm{O}=(\ket{\bm{0}}\bra{\bm{0}})^{\otimes n}$. The number of qubits is $n=4$. The target unitary $\bm{U}_X$ is chosen uniformly from a discrete set $\{\bm{U}_i\}_{i=1}^M$, where $M={2^n}$ refers to the set size and the operators $\bm{U}_j^{\dagger}\bm{O}\bm{U}_j$ with $\bm{U}_j$ in this set are orthogonal such that the operators $\bm{U}_j^{\dagger}\bm{O}\bm{U}_j$ are well distinguished. The entangled states in $\mathcal{S}$ is uniformly sampled from the set $\{\sum_{j=1}^r \sqrt{c_j} \bm{U}_{j}\ket{\bm{0}}\otimes \ket{\bm{\xi}_j} ~|~ (\sqrt{c_1}, \cdots, \sqrt{c_r})^{\top} \in \mathbb{S}\mathbb{U}(r),~ \ket{\bm{\xi}_j} \in \mathbb{S}\mathbb{U}(2^n) \}$. The size of training data is $N\in\{1,2,\cdots,16\}$ and the Schmidt rank takes $r=\{2^0,\cdots,2^4\}$.  The number of measurements takes $m\in \{10,100,300,\cdots,5000, 20000\}$. We record the averaged prediction error by learning four different $4$-qubit unitaries for $10$ training data.
	
	The simulation results are displayed in Fig.~2. Particularly, Fig.~2(a) shows that for both the cases of $N=2$ and $N=8$, the prediction error constantly decreases with respect to an increased number of measurements $m$ and increased Schmidt rank $r$ when the number of measurements is large enough, namely $m>1000$. On the other hand, for a small number of measurements with $m\le100$ in the case of $N=8$, as the Schmidt rank is continually increased, the averaged prediction error initially decreases and then increases after the Schmidt rank surpasses a critical point which is $r=3$ for $m=10$ and $r=4$ for $m=100$. This phenomenon accords with the theoretical results in Theorem~\ref{thm:informal_finite_measurement} in the sense that the entangled data play a transition role in determining the prediction error for a limited number of measurements. This observation is also verified in Fig.~2(b) for the varied sizes of training data, where for the small measurement times $m=100$, increasing the Schmidt rank could be not helpful for decreasing the prediction error. By contrast, a large training data size consistently contributes to a small prediction error, which echoes with Theorem 1.

	\medskip	
	\noindent\textbf{DISCUSSION}
	
	\noindent In this study, we exploited the effect of the Schmidt rank of entangled data on the performance of learning quantum dynamics with a fixed observable. Within the framework of the quantum NFL theorem, our theoretical findings reveal the transition role of entanglement in determining the ultimate model performance. Specifically, increasing the Schmidt rank below a threshold controlled by the number of measurements can enhance model performance, whereas surpassing this threshold can lead to a deterioration in model performance. Our analysis suggests that a large number of measurements is the precondition to use entangled data to gain potential quantum advantages. In addition, our results demystify the negative role of  entangled data in the measure of query complexity. Last, as with the classical NFL theorem, we prove that increasing the size of the training data always contributes to a better performance in QML.

	Our results motivate several important issues and questions needed to be further investigated. The first research direction is exploring whether the transition role of entangled data exists for other QML tasks such as learning quantum unitaries or learning quantum channels with the response being measurement output \cite{huang2021information, huang2022quantum, caro2023out, jerbi2023power,caro2022generalization,bisio2010optimal,khatri2019quantum,jones2022robust,heya2018variational,cirstoiu2020variational,gibbs2024dynamical,huang2023learning,caro2022learning}. These questions can be considered in both the coherent and incoherent learning protocols, which are determined by whether the target and model system can coherently interact and whether quantum information can be shared between them. Obtaining such results would have important implications for using QML models to solve practical tasks with provable advantages.
	
	A another research direction is inquiring  whether there exists a similar transition role when exploiting entanglement in quantum dynamics and measurements through the use of an ancillary quantum system. The answer for the case of entangled measurement has been given under many specific learning tasks \cite{bubeck2020entanglement, huang2021information, chen2022exponential,huang2022quantum} where the learning protocols with entangled measurements are shown to achieve an exponential advantage over those without in terms the access times to the target unitary. This quantum advantage arises from the powerful information-extraction capabilities of entangled measurements. In this regard, it is intriguing to investigate the effect of quantum entanglement on model performance when entanglement is introduced in both the training states and measurements, as entangled measurements offer a potential solution to the negative impact of entangled data resulting from insufficient information extraction via weak projective measurements. A positive result could further enhance the quantum advantage gained through entanglement exploitation.	
	
	\medskip 
	\noindent\textbf{METHODS}\\
Here we first outline the proof strategy that establishes the lower bound of the averaged prediction error in Theorem~\ref{thm:informal_finite_measurement}. Then we present an intuitive explanation of the transition role of entangled data according to the achieved numerical results.
	
	\smallskip
	
	\noindent \textbf{Proof sketch.} The backbone of the proof refers to Fano's method, which is widely used to derive the lower bound of prediction error in classical learning theory \cite{duchi2016lecture}. This method involves the following three parts.\\ Part (I): The space of the target dynamics $\mathcal{U}=\{\bm{U}\in \mathbb{S}\mathbb{U}(d)\}$ is discretized into a $2\varepsilon$-packing $\mathcal{M}_{2\varepsilon}=\{\bm{U}_{x'}\}_{x'=1}^{|\mathcal{M}_{2\varepsilon}|}$ such that the dynamics within $\mathcal{M}_{2\varepsilon}$ are sufficiently distinguishable under a distance metric related to the target function in Eqn.~(\ref{eq:learning_model}).\\
	Part (II): The dynamics learning problem in Eqn.~(\ref{eq:learning_model}) is translated to the hypothesis testing problem related to the $2\varepsilon$-packing $\mathcal{M}_{2\varepsilon}$. Such a hypothesis testing problem amounts to a communication protocol between two parties, namely Alice and Bob. Particularly, Alice chooses an element $X$ of $\{1, \cdots, |\mathcal{M}_{2\varepsilon}|\}$ uniformly at random and employs the corresponding unitary $\bm{U}_X$ to construct the training data $\mathcal{S}=\{\ket{\bm{\psi}_j}, \bm{o}_j^{(X)}\}_{j=1}^N$ with $\ket{\bm{\psi}_j}$ being the randomly sampled entangled input state and $\bm{o}_j$ being the associated measurement output of the state $(\bm{U}_X\otimes \mathbb{I})\ket{\bm{\psi}_j}$ under the projective measurement $\bm{O}$. Bob's goal is to retrieve the information of $X$ from the discrete set $\{1, \cdots, |\mathcal{M}_{2\varepsilon}|\}$ based on $\mathcal{S}$. The inferred index by Bob is denoted by $\hat{X}$. This leads to the hypothesis testing problem with the null hypothesis $\hat{X}\ne X$. In this regard, we demonstrate that the averaged prediction error is greater than a quantity related to the error probability $\mathbb{P}(\hat{X}\ne X)$ of the hypothesis testing problem, namely $\mathbb{E}_{\bm{U},\mathcal{S}}\mathrm{R}_{\bm{U}}(\bm{V}_{\mathcal{S}}) \ge \mathbb{E}_{X,\mathcal{S}} \varepsilon^2 \mathbb{P}(\hat{X}\ne X)$. This provides the theoretical guarantee of reducing the learning problem to the hypothesis testing problem. \\
Part (III): Fano's inequality is utilized to establish an upper bound on the error probability $\mathbb{P}(\hat{X}\ne X)$ of the hypothesis testing problem, i.e.,
	\begin{equation}
		\mathbb{P}(\hat{X}\ne X) \ge 1-\frac{I(X;\hat{X})+\log 2}{\log(|\mathcal{M}_{2\varepsilon}|)},
	\end{equation}
	which is dependent on two factors: the cardinality of the $2\varepsilon$-packing $\mathcal{M}_{2\varepsilon}$, and the mutual information  $I(X,\hat{X})$ between the target index $X$ and the estimated index $\hat{X}$.

	To summarize, Fano's method reduces the challenging problem of lower bounding the prediction error to separately lower bounding the packing cardinality and upper bounding the mutual information, which we could develop techniques for tackling. In particular, we obtain the lower bound of the $2\varepsilon$-packing cardinality $|\mathcal{M}_{2\varepsilon}|$ by employing the probabilistic argument to show the existence of a large but well-separated collection of quantum dynamics (i.e., the $2\varepsilon$-packing $\mathcal{M}_{2\varepsilon}$) under a metric dependent on the observable $\bm{O}$. 
	
	We establish an upper bound for the mutual information $I(X;\hat{X})$ by considering two cases: one with a small number of measurements and another with a sufficiently large number of measurements. For the former case, the mutual information $I(X;\hat{X})$ is upper bounded by a quantity involving the KL-divergence between the probability distributions of the measurement output $\bm{o}_j^{(x)}$ related to various index $x$, which has the order of $\mathcal{O}(Nmd\varepsilon^2/r)$ in the average case. On the other hand, while the mutual information $I(X;\hat{X})$ cannot grow infinitely with the number of measurements, we derive another upper bound of $I(X;\hat{X})$ with the mutual information $I(X;\{(\bm{U}_X\otimes \mathbb{I})\ket{\bm{\psi}_j}\}_{j=1}^N)$ for the case of a sufficiently large number of measurements which could extract the maximal amount of information from each output state. In this regard, the mutual information is upper bounded by an $m$-independent quantity with the order of $\mathcal{O}(rn)$. This leads to the final upper bound of the mutual information  $\min\{\mathcal{O}(Nmd\varepsilon^2/r), \mathcal{O}(rn) \}$ where the Schmidt rank $r$ plays the opposite role in the two scenarios with the various number of measurements allowed, resulting in the transition role of entangled data. 
 
    Taken all together, we can obtain the lower bound of the averaged prediction error in Theorem~\ref{thm:informal_finite_measurement}.

\smallskip	
	\noindent \textit{Remark.}
	While many similar pipelines involving the utilization of Fano's inequality have been used for obtaining the lower bound of sample complexity in quantum state learning tasks, we give a detailed explanation about how our results differ from previous studies in Supplementary Note~1E.
		
	\smallskip
	
	\noindent \textbf{Intuitive explanations based on numerical results.}
	We now give an intuitive explanation about the transition role of the entangled data based on the numerical simulations. 
	
	Before elucidating, we first detail the construction rule of the target unitaries set and the entangled input states. Particularly, to construct the set consisting of well-distinguished target unitaries under the distance metric related to the observable $\bm{O}=(\ket{\bm{0}}\bra{\bm{0}})^{\otimes n}$, one way is to choose the unitaries $\bm{U}_j$ in $\mathbb{S}\mathbb{U}(2^n)$ at uniformly random such that the target operators $\bm{U}_j^{\dagger}\bm{O}\bm{U}_j=\ket{\bm{e}_j}\bra{\bm{e}_j}$ are mutually orthogonal. In this regard, learning the target unitary under the observable $\bm{O}$ is equivalent to identifying the unknown index $k^*$ corresponding to the target operator $\bm{U}_{k^*}^{\dagger}\bm{O}\bm{U}_{k^*}=\ket{\bm{e}_{k^*}}\bra{\bm{e}_{k^*}}$.  The entangled input states have the form of $\ket{\bm{\psi}_j}=\sum_{k=1}^r \sqrt{c_{jk}} \ket{\bm{\xi}_{jk}}_{\mathcal{X}}\ket{\bm{\zeta}_{jk}}_{\mathcal{R}}$ where the Schmidt coefficients $\{c_{jk}\}$ satisfy $\sum_{k=1}^r c_{jk}=1$. As the target unitary $\bm{U}_{k^*}$ acts on the quantum system $\mathcal{X}$, the identification of the corresponding index $k^*$ solely depends on the partial trace of the entangled states, i.e., $\sigma_j := \Tr_{\mathcal{R}}(\ket{\bm{\psi}_j}\bra{\bm{\psi}_j})=\sum_{k=1}^r c_{jk} \ket{\bm{\xi}_{jk}}\bra{\bm{\xi}_{jk}}_{\mathcal{X}}$. To this end, we consider that the states $\ket{\bm{\xi}_{jk}}$ in the reduced states are sampled from the computational basis $\{\ket{\bm{e}_i}\}_{i=1}^{2^n}$ and the coefficient vector $\bm{c}_j=(\sqrt{c_{j1}},\cdots, \sqrt{c_{jr}})$ is sampled from the Haar distribution in the $r$-dimensional Hilbert space $\mathcal{H}_r$. In this manner, we construct $\mathcal{S}=\{\ket{\bm{\psi}_j}, \bm{o}_j\}_{j=1}^N$ and use it to learn the unknown index $k^*$ by solving the following minimization problem
	\begin{equation}\label{eq:method_minimization}
		\hat{k}=\arg\min_{k\in [2^n]} \sum_{j=1}^N \left(\bm{o}_j^{(k)} - \bm{o}_j \right)^2,
	\end{equation}
	where $[2^n]$ refers to the set $\{1, \cdots, 2^n\}$ and $\bm{o}_j^{(k)}$ is the collected measurement outputs by applying the observable $\bm{U}_k^{\dagger}\bm{O}\bm{U}_k$ with $\bm{U}_k \in \{\bm{U}_{k'}\}_{k'\in [2^n]}$ to input states $\ket{\bm{\psi}_j}$.
	
	The successful identification of the target index $k^*$ relies on the satisfaction of two key conditions: 
	\begin{enumerate}
		\item The states set  $\{\ket{\bm{\xi}_{ji}}\bra{\bm{\xi}_{ji}}\}_{j,i=1}^{N,r}$ contains the target operator $\bm{U}_{k^*}^{\dagger}\bm{O}\bm{U}_{k^*} = \ket{\bm{e}_{k^*}}\bra{\bm{e}_{k^*}}$.
		\item The measurement outputs $\{\bm{o}_j\}_{j=1}^N$ closely approximate the corresponding Schmidt coefficient $c_{k^*}$ of the operator $\bm{U}_{k^*}^{\dagger}\bm{O}\bm{U}_{k^*} =\ket{\bm{e}_{k^*}}\bra{\bm{e}_{k^*}} \in \{{\ket{\bm{\xi}_{ji}}\bra{\bm{\xi}_{ji}}}\}_{j,i=1}^{N,r}$. 
	\end{enumerate}
	The first condition ensures that the measurement outputs $\{\bm{o}_j\}_{j=1}^N$ are non-zero, enabling the identification of the target index $k^*$. Moreover, if the target operator $\bm{U}_{k^*}^{\dagger}\bm{O}\bm{U}_{k^*}$ is not included in $\{\ket{\bm{\xi}_{ji}}\bra{\bm{\xi}_{ji}}\}_{j,i=1}^{N,r}$, measuring the output states with observable $\bm{O}$ will always yield zero for all $k\in[2^n]$ such that any $k\in [2^n]$ is the solution of Eqn.~(\ref{eq:method_minimization}).  
	
The second condition ensures that the non-zero measurement outcomes can closely approximate the ground truth to facilitate the identification of the target index $k^*$. In this regard, the states set $\{{\ket{\bm{\xi}_{ji}}\bra{\bm{\xi}_{ji}}}\}_{j,i=1}^{N,r}$, which associates with highly entangled input states (a large Schmidt rank $r$), is more likely to contain the unknown target operator $\bm{U}_{k^*}^{\dagger}\bm{O}\bm{U}_{k^*}$. Consequently, they exhibit a higher identifiability with a greater probability of producing non-zero measurement outcomes. On the other hand, the average magnitude of the Schmidt coefficients $\{c_{ji}\}_{i=1}^r$ of a highly entangled state $\ket{\bm{\psi}_j}$ decreases with the Schmidt rank $r$. This leads to a small probability of the target operator $\bm{U}_{k^*}^{\dagger}\bm{O}\bm{U}_{k^*}$ being measured according to the Born rule. As a result, achieving a close approximation necessitates a larger number of measurements. For instance, the entangled states with $r=2^n$ could always yield non-zero measurement outputs with a large number of measurements. Conversely, if the unentangled state $\ket{\bm{\psi}_j}\bra{\bm{\psi}_j}$ with $r=1$ is exactly identical to the target operator $\bm{U}_{k^*}^{\dagger}\bm{O}\bm{U}_{k^*}$, one measurement is sufficient to identify the target index $k^*$. These observations provide an intuitive indication about the transition role of entangled data from the lens of quantum information that the entangled data could contain more information than unentangled data, but at the same time increase the difficulty of information extraction using projective measurement.

	\medskip
	
	\noindent \textbf{DATA AVAILABILITY}\\
	The entangled data generated in this study are available at the Github repository \url{https://github.com/wangxinbiao08/Transition-role-of-entangled-data-in-QML}.
	
	\medskip
	
	\noindent \textbf{CODE AVAILABILITY}\\
	The code used in this study are available at the Github repository \url{https://github.com/wangxinbiao08/Transition-role-of-entangled-data-in-QML}.

\medskip

\noindent \textbf{ACKNOWLEDGEMENTS}\\
Y.L. acknowledges support from National Natural Science Foundation of China (Grant No. U23A20318 and 62276195). X.Y. acknowledges support from the National Natural Science Foundation of China (Grant No.~12175003 and No.~12361161602),  NSAF (Grant No.~U2330201).

\medskip

\noindent \textbf{AUTHOR CONTRIBUTIONS}\\
The project was conceived by Y.D., and D.T. Theoretical results were proved by X.W., Y.D., and Z.T. Numerical simulations and analysis were performed by X.W., Y.D., Y.L., and X.Y. All authors contributed to the write-up.

\medskip

\noindent \textbf{COMPETING INTERESTS}\\
The authors declare no competing interests.

 \clearpage
 \newpage
 \appendix
 \onecolumngrid 

\begin{center}
\large{Supplementary Note for \\ ``Transition Role of Entangled Data in Quantum Machine Learning''}
 \end{center}

\tableofcontents

\medskip

\noindent  \textbf{Roadmap:} Supplementary Note~\ref{app_sec:prelimilaries} provides preliminaries for the necessary mathematical background and introduces the related work. The problem setup of quantum dynamics learning in the framework of the quantum no-free-lunch (NFL) theorem is presented in Supplementary Note~\ref{app_sec:NFL_Preliminaries}. We elucidate the results and the proof of Theorem~\ref{thm:formal_finite_measurement} in Supplementary Note~\ref{app_sec:NFL_SM}. Finally, Supplementary Note~\ref{app_sec:Numerical_NFL} exhibits the numerical details omitted in the main text and more numerical results. 

\section{Preliminaries}\label{app_sec:prelimilaries}
In this section, we will first present some essential mathematical foundations for deriving the main results of this work. It encompasses several key aspects, such as the introduction of pertinent notations, random variables, information theory, and Haar integration, which are separately elaborated upon in Supplementary Note~\ref{app_subsec:notation} to Supplementary Note~\ref{app_subsec:Haar_integration}. Moreover, to contextualize our work within the existing literature, we conduct a comprehensive review of relevant studies in Supplementary Note~\ref{app_subsec:related_work}.

\subsection{Notation} \label{app_subsec:notation}
We unify the notations throughout the whole work. The number of qubits and training data size are denoted by $n$ and $N$, respectively. Let $\mathcal{H}_d$ denote a $d$-dimensional Hilbert space. The Hilbert space of an $n$-qubit  system is denoted by $\mathcal{H}_{2^n}$. For simplicity, $d$ refers to the dimension of the $n$-qubit quantum system, i.e., $d=2^n$. The $d$-dimensional unitary group and special unitary group are denoted as $\mathbb{U}(d)$ and $\mathbb{S}\mathbb{U}(d)$, respectively.  The notations $[N]$ refers to the set $\{1, 2, \cdots, N\}$. We denote $\|\cdot\|_1$ as the trace norm and $\|\cdot\|_{\Fnorm}$ as the Frobenius norm. The cardinality of a set is denoted as $|\cdot|$. We use the standard bra-ket notation for pure quantum states. The identity operator on the $d$-dimensional Hilbert space is denoted by $\mathbb{I}_d$. We denote $\ket{\bm{e}_k}$ as the computational basis with the $k$-th entry being one and the other entries being zero. For clearness, we define sample complexity as the size of training data, or equivalently, the number of quantum states in the training data, and define query complexity as the total number of queries of the explored quantum system.

\subsection{Random variables} \label{app_subsec:Random_variables}
We denote random variables using the capital letter, i.e., $X$, including matrix-valued random variables. We use the lowercase letters in Roman (e.g., $\rm{p},q$) and the capital letters (e.g., $\mathbb{P}$) with appropriate subscripts to denote the probability density function (PDF) and the corresponding cumulative distribution function (CDF) which obey $\mathrm{d}\mathbb{P}_{X}(x)= \mathrm{p}_{X}(x) \mathrm{d} x$. For instance, suppose $X$ is a random variable taking values in $\mathcal{X}$ according to some distribution $\mathrm{p}_{X}: \mathcal{X} \to [0,1]$, where $\mathcal{A}$ is the set of Borel-measurable subsets \cite{holevo2011probabilistic} of $\mathcal{X}$. Let $\mathrm{g}: \mathcal{X} \to \mathbb{R}$ be any function of $x$. We denote $\mathbb{E}_{X}\mathrm{g}(X)$ and  $\mathbb{E}_{X\sim \mathrm{p}_{X}}\mathrm{g}(X)$ interchangeably as the expectation of $\rm{g}(\cdot)$ with respect to the distribution $\mathrm{p}_{X}$, i.e., 
\[ \int_{\mathcal{X}}\mathrm{g}(X) \mathrm{d} \mathbb{P}_{X}(x) ~ \text{or}~  \int_{\mathcal{X}}\mathrm{g}(X) \mathrm{p}_{X}(x) \mathrm{d} x,\] where the latter notation is used when there may be some ambiguity about the distribution is. When no confusion occurs, we drop  all subscripts and write $\mathbb{E}\mathrm{g}(X)$. Next suppose we have random variables $(X,Y)$ jointly distributed on $\mathcal{X} \times \mathcal{Y}$. We use $\mathrm{p}_{Y|x}(y)$ and $\mathrm{p}(y|X=x)$ interchangeably to denote the conditional probability that $Y=y$ given $X=x$.

\subsection{Information theory} \label{app_subsec:Information_theory}
In this subsection, we review the basic definitions in information theory, including (Shannon) entropy, KL-divergence, mutual information, and their conditional versions. Refer to Refs.~\cite{duchi2016lecture,wilde2013quantum,watrous2018theory} for a more deep understanding. Throughout the whole paper, $\log$ denotes the logarithm with base $2$. We first consider the scenario of discrete random variables taking values in the same space. 

\medskip

\noindent \textbf{Entropy}. We begin with a central concept in information theory---the (Shannon) entropy. Let $\mathbb{P}$ be a distribution on a finite set $\mathcal{X}$. Denote the PDF as $\mathrm{p}$ associated with $\mathbb{P}$. Let $X$ be a random variable distributed according to $\mathbb{P}$. The \textit{entropy} of $X$ (or of $\mathrm{p}$) is defined as
\begin{equation}\label{eq:entropy}
	\mathrm{H}(X):=-\sum_{x\in \mathcal{X}} \mathrm{p}(x)\log \mathrm{p}(x).
\end{equation}
This quantity is always positive when $\mathrm{p}(x) < 1$ for all $x$, and vanishes if and only if $X$ is a deterministic variable, i.e., $\mathrm{p}(x)=1$ with taking $0\log0 = 0$. The Shannon entropy measures the uncertainty about the random variable $X$. 

Let us consider another discrete random variable, denoted by $Y$, which takes values in the set $\mathcal{Y}$. The joint distribution of $X$ and $Y$ can be given by $\mathrm{p}_{X,Y}: \mathcal{X} \times \mathcal{Y} \to [0,1]$. The joint entropy of these random variables is	
\begin{equation}
	\mathrm{H}(X,Y) = -\sum_{x\in \mathcal{X}}\sum_{y\in \mathcal{Y}} \mathrm{p}_{X,Y}(x,y) \log(\mathrm{p}_{X,Y}(x,y))
\end{equation}
and the conditional entropy of $X$ given $Y$ refers to
\begin{equation}\label{eq:conditional_entropy}
	\mathrm{H}(X|Y) = \mathrm{H}(X,Y)-\mathrm{H}(Y),
\end{equation}
which can be intuitively interpreted as the total information of $X$ and $Y$, and the amount of information left in the random variable $X$ after observing the random variable $Y$, respectively. 

We next review some important properties of Shannon entropy.

\begin{property}[Subadditivity of entropy, \cite{duchi2016lecture}]
	\label{property:subadditivity_entropy}
	Let $X_1, \cdots, X_t$ be any sequence of random variables, we have
	\begin{equation}
		\mathrm{H}(X_1, \cdots, X_t) \le \mathrm{H}(X_1) + \cdots + \mathrm{H}(X_t).
	\end{equation}
\end{property} 

\begin{property}[Maximum Value, Property 10.1.5 in \cite{wilde2013quantum}]
	\label{property:max_value_entropy}
	The maximum value of the entropy $\mathrm{H}(X)$ for a random variable $X$ taking values in an alphabet $\mathcal{X}$ is $\log|\mathcal{X}|$:
	\begin{equation}\label{eq:max_value_entropy}
		\mathrm{H}(X) \le \log(|\mathcal{X}|).
	\end{equation}
\end{property}

\begin{property}
	\label{property:conditional_entropy}
	Let $X,Y$ be any two random variables, we have
	\begin{equation}
		\mathrm{H}(Y) \le \mathrm{H}(X, Y).
	\end{equation}
\end{property}
This property can be obtained by observing the definition of conditional entropy that $\mathrm{H}(Y)=\mathrm{H}(X,Y)-\mathrm{H}(X|Y)$ and the non-negativity of entropy.

\medskip

\noindent \textbf{KL-divergence and $\chi^2$-divergence}. KL-divergence is used to measure the distance between probability distributions. Consider two discrete distributions (PDFs) $\mathrm{p},\mathrm{q}: \mathcal{X} \to [0,1]$ with $\mathcal{X}$ being the value space of the associated random variables, the KL-divergence between $\mathrm{p}$ and $\mathrm{q}$ is given by
\[
\mathrm{D}_{\KL}(\mathrm{p}~\| \mathrm{q}) =\left\{
\begin{aligned}
	& \sum_{x \in \mathcal{X}} \mathrm{p}(x) \log \frac{\mathrm{p}(x)}{\mathrm{q}(x)},  & ~ \supp(\mathrm{p}) \subseteq \supp(\mathrm{q}), \\
	& \infty,  & ~ \mbox{otherwise.}
\end{aligned}
\right.
\]
The KL-divergence satisfies $\mathrm{D}_{\KL}(\mathrm{p}\|\mathrm{q}) \ge 0$, where the equality holds if and only if $\mathrm{p}(x)=\mathrm{q}(x)$ for all $x\in \mathcal{X}$. Particularly, employing the convexity of $\log$ and Jensen's inequality yields
\begin{eqnarray}
	\mathrm{D}_{\KL}(\mathrm{p}\|\mathrm{q}) = -\mathbb{E} \left[ \log\frac{\mathrm{q}(X)}{\mathrm{p}(X)}  \right] \ge -\log \mathbb{E} \left[ \frac{\mathrm{q}(X)}{\mathrm{p}(X)}  \right] = -\log \left(\sum_{x\in \mathcal{X}} \mathrm{p}(x) \frac{\mathrm{q}(x)}{\mathrm{p}(x)}  \right) = -\log(1) = 0.
\end{eqnarray}

Notably, the simplification of KL-divergence could be difficult for some specific distributions due to the logarithm operation of PDFs. In this regard, an alternative way to compare distributions defined in the same space is the $\chi^2$-divergence.
\begin{definition}[$\chi^2$-divergence]\label{def:chi_div}
	The $\chi^2$-divergence between two discrete distributions $\mathrm{p}, \mathrm{q}: \mathcal{X} \to [0,1]$ defined on the same sample space $\mathcal{X}$ refers to
	\begin{equation}
		\mathrm{D}_{\chi^2}(\mathrm{p}\|\mathrm{q})=\sum_{x \in \mathcal{X}} \mathrm{q}(x) \left( \frac{\mathrm{p}(x)}{\mathrm{q}(x)} \right)^2-1.
	\end{equation}
\end{definition}
Remarkably, the KL-divergence can be upper bounded by $\chi^2$-divergence up to a multiplicative factor, which is encapsulated in the following lemma.
\begin{lemma}[Lemma 2.9 in Ref.~\cite{lowe2022lower}]
	Let $\mathrm{p}, \mathrm{q}: \mathcal{X} \to [0,1]$ be two discrete distributions with $\mathcal{X}$ being the value space of the associated random variables. It holds that
	\begin{equation}
		\mathrm{D}_{KL}(\mathrm{p}\|\mathrm{q}) \le  \frac{1}{ln(2)} \cdot \mathrm{D}_{\chi^2}(\mathrm{p}\|\mathrm{q}).
	\end{equation}
\end{lemma}

\medskip

\noindent \textbf{Mutual information}. The \textit{mutual information} $\mathrm{I}(X;Y)$ between $X$ and $Y$ is defined as the KL-divergence between their joint probability distribution $\mathrm{p}_{X,Y}$ and their product of (marginal) probability distributions $\mathrm{p}_X \mathrm{p}_Y$. Mathematically,
\begin{equation}\label{eq:MI_KL}
	\mathrm{I}(X;Y):= \mathrm{D}_{\KL}(\mathrm{p}_{X,Y} \| \mathrm{p}_X \mathrm{p}_Y) = \sum_{x,y} \mathrm{p}_{X,Y}(x,y) \log \frac{\mathrm{p}_{X,Y}(x,y)}{\mathrm{p}_{X}(x) \mathrm{p}_{Y}(y)}.
\end{equation} 
According to the non-negativity property of KL-divergence, it is easy to see that $\mathrm{I}(X;Y) \ge 0$ and the equality holds if and only if random variables $X$ and $Y$ are independent, i.e., $\mathrm{p}_{X,Y}(x,y)=\mathrm{p}_{X}(x) \mathrm{p}_{Y}(y)$. The mutual information between two random variables can be defined in several mathematically equivalent ways. Another definition used in this work expresses the mutual information in terms of the entropy of random variables. It may be shown that the above equation is equal to
\begin{equation}\label{eq:MI_entropy}
	\mathrm{I}(X; Y) = \mathrm{H}(X) - \mathrm{H}(X|Y) = \mathrm{H}(Y) - \mathrm{H}(Y|X) = \mathrm{H}(X) + \mathrm{H}(Y) - \mathrm{H}(X,Y).
\end{equation}
We would give a brief derivation of the equivalence between the definition of mutual information given in the first equality of Eqn.~(\ref{eq:MI_entropy}) and that in Eqn.~(\ref{eq:MI_KL}), while the second equality and the third equality in Eqn.~(\ref{eq:MI_entropy}) employs the definition of conditional entropy in Eqn.~(\ref{eq:conditional_entropy}). Particularly, using Bayes' rule, we have $\mathrm{p}_{X,Y}(x,y)=\mathrm{p}_Y(y)\mathrm{p}_{X|Y}(x|y)$, so 
\begin{eqnarray}
	\mathrm{I}(X;Y) = && \sum_{x,y} \mathrm{p}_Y(y) \mathrm{p}_{X|Y}(x|y) \log\frac{\mathrm{p}_{X|Y}(x|y)}{\mathrm{p}_X(x)} 
	\nonumber \\
	= && -\sum_{x,y} \mathrm{p}_Y(y) \mathrm{p}_{X|Y}(x|y) \log{\mathrm{p}_X(x)} + \sum_{y} \mathrm{p}_Y(y)  \sum_{x} \mathrm{p}_{X|Y}(x|y) \log \mathrm{p}_{X|Y}(x|y)
	\nonumber \\
	= && \mathrm{H}(X) - \mathrm{H}(X|Y).
\end{eqnarray}
To this end, the mutual information can be thought of as the amount of entropy removed (on average) in $X$ by observing $Y$.

Analogous to the definition of conditional entropy, the \textit{conditional mutual information} between $X$ and $Y$ given a third random variable $Z$ is defined as
\begin{equation}
	\mathrm{I}(X;Y|Z) := \sum_{z} \mathrm{I}(X;Y|Z=z) \mathrm{p}(z) =\mathrm{H}(X|Z) - \mathrm{H}(X|Y,Z) = \mathrm{H}(Y|Z) - \mathrm{H}(Y|X,Z),
\end{equation}
which refers to the mutual information between $X$ and $Y$ when $Z$ is observed (on average). 

We now present three fundamental properties about the mutual information. These facts will be leveraged to derive stronger lower bounds on the prediction error than those obtained by using Holevo's theorem when the measurements are subject to certain restrictions.

\begin{property}
	[Chain rule for mutual information]\label{property:chain_rule_MI}
	Let $X, Y_1, Y_2, \cdots, Y_t$ be any $t+1$ random variables, the mutual information between $X$ and $Y_1, \cdots, Y_t$ obeys
	\begin{equation}\label{eq:chain_rule_MI}
		\mathrm{I}(X;Y_1, \cdots, Y_t) = \sum_{j=1}^t \mathrm{I}(X; Y_j|Y_{j-1}, \cdots, Y_1).
	\end{equation}
\end{property}

\begin{property}
	[Subadditivity of mutual information]\label{property:subadd_MI}
	If $Y_1, \cdots, Y_t$ are independent given $X$, the mutual information between $X$ and $Y_1, \cdots, Y_t$ obeys
	\begin{equation}\label{eq:subadd_MI}
		\mathrm{I}(X;Y_1, \cdots, Y_t)  \le \sum_{j=1}^t \mathrm{I}(X;Y_j).
	\end{equation}
\end{property}

\begin{property}
	[Data processing inequality]\label{property:data_pro_ineqn}
	Let the random variables $X,Y,Z$ form a Markov chain $X \to Y \to Z$ such that if given $Y$, the random variables $X$ and $Z$ are independent. Then we have 
	\begin{equation}\label{eq:data_pro_ineqn}
		\mathrm{I}(X;Z) \le \mathrm{I}(X;Y).
	\end{equation}
\end{property}

\subsection{Haar integration} \label{app_subsec:Haar_integration}
In this subsection, we present a set of lemmas that enable the analytical computation of integrals of polynomial functions over the unitary group and the quantum state with respect to the unique normalized Haar measure. For a more comprehensive discussion on this subject matter, we refer the readers to Refs.~\cite{collins2006integration,puchala2011symbolic,cerezo2021cost}.
\begin{lemma}\label{lem:Tr(WW)}
	Let $\{ W_y\}_{y\in Y} \subset \mathbb{U}(d)$ form a unitary $t$-design with $t\ge1$, and let $A, B: \mathcal{H}_d\to\mathcal{H}_d$ be arbitrary linear operators. Then
	\begin{equation}
		\frac{1}{|Y|}\sum_{y\in Y}\Tr\left(W_{y}AW_{y}^{\dagger}B\right)=\int_{W \sim \haar}\mathrm{d}W\Tr\left(WAW^{\dagger}B\right)=\frac{\Tr(A)\Tr(B)}{d}.
	\end{equation}
\end{lemma}

\begin{lemma}\label{lem:Tr(WW)Tr(WW)}
	Let $\{ W_y\}_{y\in Y} \subset \mathbb{U}(d)$ form a unitary $t$-design with $t\ge 2$, and let $A, B, C, D: \mathcal{H}_d\to\mathcal{H}_d$ be arbitrary linear operators. Then
	\begin{align}
		\frac{1}{|Y|}\sum_{y\in Y}\Tr\left(W_{y}AW_{y}^{\dagger}B\right)\Tr\left(W_{y}CW_{y}^{\dagger}D\right)=&\int_{W \sim \haar}\mathrm{d}W \Tr(WAW^{\dagger}B)\Tr(WCW^{\dagger}D) \nonumber \\
		=&\frac{1}{d^2-1}\Big(\Tr(A)\Tr(B)\Tr(C)\Tr(D) + \Tr(AC)\Tr(BD) \Big) \nonumber \\
		& - \frac{1}{d(d^2-1)}\Big(\Tr(AC)\Tr(B)\Tr(D) + \Tr(A)\Tr(C)\Tr(BD) \Big).
	\end{align}
\end{lemma}

\begin{lemma}\label{lem:Tr(psipsi)}
	Let the quantum state $\ket{\bm{\phi}} \in \mathcal{H}_d$ follows the Haar distribution, and let $A: \mathcal{H}_d\to\mathcal{H}_d$ be arbitrary linear operators. Then
	\begin{equation}
		\int_{\bm{\phi} \sim \haar}\ket{\bm{\phi}}\bra{\bm{\phi}} \mathrm{d}\bm{\phi} = \frac{\mathbb{I}_d}{d}, ~  \mbox{and hence}  ~ \int_{\bm{\phi} \sim\haar}\mathrm{d}{\bm{\phi}} \Tr\left(\ket{\bm{\phi}}\bra{\bm{\phi}}A \right)
		= \frac{\Tr(A)}{d}.
	\end{equation}
\end{lemma}

\begin{lemma}\label{lem:Tr(psipsi)Tr(psipsi)}
	Let the quantum state $\ket{\bm{\phi}} \in \mathcal{H}_d$ follows the Haar distribution, and let $A, B: \mathcal{H}_d\to\mathcal{H}_d$ be arbitrary linear operators. Then
	\begin{equation}
		\int_{\bm{\phi} \sim \haar}\ket{\bm{\phi}}\bra{\bm{\phi}}^{\otimes 2} \mathrm{d} \phi = \frac{\mathbb{I}_d^{\otimes 2}+\swap}{d(d+1)}, ~  \mbox{and}  ~ 
		\int_{\bm{\phi} \sim \haar} \mathrm{d}{\bm{\phi}} \Tr\left(\ket{\bm{\phi}}\bra{\bm{\phi}}A \right)\Tr\left(\ket{\bm{\phi}}\bra{\bm{\phi}}B \right)
		= \frac{\Tr(A)\Tr(B)+\Tr(AB)}{d(d+1)},
	\end{equation}
	where the notation $\swap$ refers to the swap operator on the space $\mathcal{H}_d^{\otimes 2}$.	
\end{lemma}

\begin{lemma}\label{lem:Tr(WA)Tr(WB)_Orth}
	Let $A: \mathcal{H}_d \to \mathcal{H}_d$ be arbitrary traceless linear operators, i.e., $\Tr(A)=0$. Let $\ket{\bm{\mu}}$ and $\ket{\bm{\nu}}$ be two orthogonal Haar random states. Then
	\begin{equation}\label{eq:Tr(WA)Tr(WB)_Orth}
		\int_{\bm{\nu} \sim \haar}\int_{\bm{\mu} \sim \haar}\Tr(A\ket{\bm{\mu}}\bra{\bm{\mu}})\Tr(A\ket{\bm{\nu}}\bra{\bm{\nu}}) \mathrm{d} \bm{\mu} \mathrm{d} \bm{\nu} = -\frac{1}{d(d^2-1)} \Tr\big( A^2 \big).
	\end{equation}
\end{lemma}

While Lemma~\ref{lem:Tr(WW)} to Lemma~\ref{lem:Tr(psipsi)Tr(psipsi)} have been established in many previous literature \cite{collins2006integration,puchala2011symbolic,cerezo2021cost}, Lemma~\ref{lem:Tr(WA)Tr(WB)_Orth} is particularly tailored to derive the main results of this work in a specific setting. To this end, here we provide a brief derivation of Lemma 5 by employing the results of Lemma 2.

\begin{proof}[Proof of Lemma~\ref{lem:Tr(WA)Tr(WB)_Orth}]
	We first observe that the orthogonal Haar random states $\ket{\bm{\mu}}$ and $\ket{\bm{\nu}}$ can be treated as the arbitrary two columns of a Haar random unitary $W$ in $\mathbb{S}\mathbb{U}(d)$. With this regard, $\ket{\bm{\mu}}$ and $\ket{\bm{\nu}}$ can be expressed in the form of $\ket{\bm{\mu}}=W\ket{0}^{\otimes n}$ and $\ket{\bm{\nu}}=W\ket{1}^{\otimes n}$, respectively. Therefore, we have
	\begin{align}
		& \int_{\bm{\nu} \sim \haar}\int_{\bm{\mu} \sim \haar}\Tr(A\ket{\bm{\mu}}\bra{\bm{\mu}})\Tr(A\ket{\bm{\nu}}\bra{\bm{\nu}}) \mathrm{d} \bm{\mu} \mathrm{d} \bm{\nu} 
		\nonumber \\
		= & \int_{W \sim \haar} \Tr\left(AW \ket{\bm{0}}\bra{\bm{0}}^{\otimes n}W^{\dagger})\Tr(AW \ket{\bm{1}}\bra{\bm{1}}^{\otimes n}W^{\dagger}\right) \mathrm{d} W 
		\nonumber \\
		= &  \frac{1}{d^2-1} \left( \Tr\big( A \big)^2 + \Tr\big( A^2 \big) \Tr\big(\ket{0}\bra{0}^{\otimes n} \ket{1}\bra{1}^{\otimes n} \big) \right) - \frac{1}{d(d^2-1)} \left( \Tr\big( A^2 \big) +  \Tr\big( A \big)^2 \Tr\big(\ket{\bm{0}}\bra{\bm{0}}^{\otimes n} \ket{\bm{1}}\bra{\bm{1}}^{\otimes n} \big) \right)
		\nonumber \\
		= & -\frac{1}{d(d^2-1)} \Tr\big( A^2 \big),
	\end{align}
	where the second equality employs Lemma~\ref{lem:Tr(WW)Tr(WW)}, the second equality follows $\Tr(A)=0$ and $\Tr(\ket{\bm{0}}\bra{\bm{0}}^{\otimes n} \ket{\bm{1}}\bra{\bm{1}}^{\otimes n} )=0$. This completes the proof.
\end{proof}

\subsection{Related work}\label{app_subsec:related_work}
In this subsection, we review prior literature  related to the quantum NFL theorem, quantum dynamics learning, and the main methodologies employed in the proof of this work.

\medskip

\noindent	\textbf{Quantum NFL theorem.} 
The no-free-lunch (NFL) theorem is a celebrated result in learning theory that limits one’s ability to learn a function with a training dataset. It is widely studied in classical learning theory \cite{wolpert1996existence, wolpert1996lack, wolpert1997no, ho2002simple,wolf2018mathematical,adam2019no}. Ref.~\cite{poland2020no} made the initial attempt to establish NFL theorem in the context of quantum machine learning where the inputs and outputs are quantum states. Ref.~\cite{sharma2022reformulation} reformulated the quantum NFL theorem in the quantum-assisted learning protocols where the bipartite entangled states are prepared as the input states by introducing a reference quantum system. They demonstrated that the utilization of entangled data can remove the exponential cost of the training data size for learning unitaries with unentangled data. However, their results are established in the ideal setting with infinite measurements, which can not apply to quantifying the practical power of entangled data.

\medskip

\noindent   \textbf{Complexity in QML.} 
The complexity of quantum learning is a multi-faceted concept \cite{banchi2023statistical}, including query complexity (a.k.a, copy complexity) and sample complexity (a.k.a, data complexity). Each of these metrics has its practical meaning and has been separately studied in a lot of work to explore the potential quantum advantages in the field of QML. We categorize and summarize the related work according to the employed evaluation metrics in Supplementary Table~\ref{tab:complexity}.

\renewcommand{\arraystretch}{1.5}
\begin{table}[ht]
	\centering
	\begin{tabular}{|p{7.5cm}|p{8.2cm}|}
		\hline  
		\multicolumn{2}{|c|}{\textbf{Sample complexity}}\\
		\hline
		\centering \textbf{Description} &   \textbf{~~~~~~~~~~~~~~~~Related work} \\
		\hline
		As in classical machine learning, limitations on the amount of available data which may be classical or quantum play a key role in determining the achievable accuracy levels for data-driven methods. Sample complexity, denoted as $N$, refers to the requirements in terms of the number of distinct input states on the performance of a learning algorithm. & \makecell[lt]{\textbf{\cite{huang2021power}}: studying the 
			generalization of quantum kernels; \\ \textbf{\cite{sharma2022reformulation}}:
			studying quantum NFL theorem. \\ \textbf{\cite{caro2023out}}: studying the 
			out-of-distribution generalization \\ in learning quantum unitary;\\
			\textbf{\cite{caro2022generalization}}: studying the generalization of VQAs with finite \\ parameterized 
			gats;}
		\\
		\hline 
		\multicolumn{2}{|c|}{\textbf{Query complexity}}\\
		\hline 
		\centering \textbf{Description} & \textbf{~~~~~~~~~~~~~~~~Related work} \\
		\hline
		With quantum data, the number of copies available of a given quantum state, i.e., the number of measurements $m$, determines the amount of information that can be extracted from the state. Query complexity, denoted as $Q$, reflects the requirements in terms of copies of all quantum states in the training dataset that are needed to ensure given accuracy levels. In this regard, we have $Q=N\cdot m$. & \makecell[lt]{\textbf{\cite{huang2021information}, \cite{chen2022exponential}, \cite{huang2022quantum}}: studying the quantum advantage of \\ learning 
			quantum processes with quantum memory.}
		\\  
		\hline  
	\end{tabular}
	\caption{Definitions of sample complexity and query complexity,  and the related work in the field of QML.}
	\label{tab:complexity}
\end{table}

In particular, Ref.~\cite{huang2021power} studied the generalization of quantum kernels in terms of sample complexity and achieved provable quantum advantage for some specific data. Refs.~\cite{caro2022generalization,caro2023out,caro2021encoding} studied the generalization of quantum learning models in terms of sample complexity from various aspects, like data encoding, data distribution, and circuit complexity. Moreover, Refs.~\cite{huang2021information,huang2022quantum, chen2022exponential} studied the query complexity of learning quantum processes with and without quantum memory. Besides, instead of considering the query complexity, the sample complexity and the number of measurements for each input state are also considered in many existing works. Ref.\cite{liu2021rigorous} and Ref.~\cite{wang2021towards} studied the generalization error bound of quantum kernels in terms of sample complexity and the number of measurements under the setting of finite measurement limitation. It is meaningful in the sense that some states are easy to prepare many times while preparing a large set of different training states is hard in many practical scenarios.

\medskip

\noindent	\textbf{Quantum dynamics learning.}
Quantum computers have the potential to enhance classical machine-learning models. One such application is the utilization of quantum dynamics learning, which involves the conversion of an analogue quantum unitary into a digital form that can be subsequently examined on either a quantum or classical computer. Numerous studies have delved into understanding the provable quantum advantage in terms of sample complexity, aiming to achieve approximate learning of quantum unitaries with small prediction errors. For instance, Ref.~\cite{chung2018sample,huang2021power,liu2021rigorous,caro2022generalization,wang2021towards,abbas2021power,banchi2021generalization,caro2021encoding,cai2022sample,du2022efficient,du2023problem,jerbi2023power} studied in variational quantum machine learning how the sample complexity bounds are determined by the type of datasets and the complexity of the quantum model which acts as an unknown quantum unitary to be learned. These findings are obtained upon assumptions regarding the complexity of the unknown unitaries or the type of datasets, which are subsequently utilized to constrain the information-theoretic complexity of the learning task. As a result, these approaches are typically not applicable to general datasets and highly complex quantum unitaries without specific conditions. Recently, a growing literatures prove the exponential separation in sample complexity between learning unitaries with and without quantum memory \cite{chen2022exponential, huang2023learning, huang2021information, buadescu2021improved}. However, the proposed algorithms heavily depend on the entanglement measurement for optimal tomography, which may pose significant challenges from a practical standpoint. 

Besides learning the quantum dynamics for predicting the output states, another important task is learning the quantum dynamics under a given observable to predict the expectation value of the observable \cite{huang2020predicting, huang2021information}. Among various choices of observable, taking the projective measurement as the observable related to taking the projective measurement as observable in the learning task $\mathrm{f}_{\bm{U}}(\bm{\rho})=\Tr(\bm{O}\bm{U}\bm{\rho} \bm{U}^{\dagger})$ is related to many practical tasks, such as the classification task \cite{li2022recent,caro2022generalization,du2023problem} and quantum feature extraction and dimensionality reduction \cite{lloyd2014quantum,liang2020variational,xin2021experimental}. In quantum classification tasks, projective measurements can be used to distinguish between different classes of data. After processing the quantum states representing data points through a variational circuit, projective measurements are performed to read out the result, which corresponds to the classification of the data.  A typical task for dimensionality reduction is a quantum principle component analysis. For given training data, learning the target unitary under projective measurement could be used to construct the most significant features of the original data.

\medskip

\noindent	\textbf{Techniques employed in the proof.}
The main used tool for proving the lower bound of the prediction error is Fano's inequality. This inequality is widely used for obtaining the lower bounds in both classical and quantum learning theory with various formulations \cite{Tsybakov_2009,wilde2013quantum,duchi2016lecture,quek2022exponentially}. In the filed of quantum machine learning, one of its applications involves utilizing Fano's inequality to derive the lower bound of sample complexity for quantum state tomography \cite{flammia2012quantum,haah2016sample, lowe2022lower,huang2021information}. Particularly, a fundamental framework for establishing such lower bounds builds upon insights from Fano's inequality \cite{Tsybakov_2009,wilde2013quantum}, suggesting that performing tomography with sufficient precision accurately reduces to distinguishing well-separated quantum states.  However, the aim in our work is to proving the lower bound of prediction errors   for dynamics learning. In this regard, we adopt another formulation of Fano's inequality given in Ref.~\cite{duchi2016lecture}, which is widely used in classical learning theory for proving the lower bound of prediction errors. Such formulation suggests that learning unitary with limited training data accurately reduces to distinguishing well-separated PDF of the related measurement outputs. Intuitively, the expression of Fano's inequality used in this work and that used in the context of quantum state tomography could be regarded as the \textit{dual formulation} in the sense that we employ the amount of training quantum states to lower bound the prediction error while previous works exploited a pre-fixed prediction error to lower bound the sample complexity.

Despite the various formulations of Fano's inequality, a common and standard technique in the field of density estimation involves ``discretizing'' the learning problem. This technique is employed to achieve worst-case lower bounds, which can be seen as the classical counterpart of quantum tomography (see for example Chapter 2 of Ref.~\cite{Tsybakov_2009}). One way to rigorously establish this argument is by utilizing Fano's inequality in conjunction with Holevo's theorem, which offers an interpretation in terms of a communication protocol between two parties, namely Alice and Bob. Holevo's theorem is conventionally employed to provide an upper bound on the mutual information between the target message $X$, transmitted by Alice, and the decoded message $\hat{X}$ received by Bob. However, there are two caveats to using Holevo's theorem in the context of the quantum NFL theorem. First, the utilization of Holevo's theorem does not take into account restrictions on the measurements we are allowed to perform on the $N$ copies of the state. Second,  it fails to consider the limitations imposed by the availability of a limited number of training states on the accessible information regarding $X$.

In contrast to previous studies, our approach distinguishes itself by directly exploiting the connection between the mutual information of two random variables and the Kullback-Leibler (KL) divergence of related distributions. Additionally, we utilize techniques for Haar integration with respect to the random training data and the random target unitary. Two crucial technical steps in our approach involve analyzing the KL divergence instead of individual measurement outcomes' probabilities, as well as investigating the maximal mutual information under the constraint of limited training states. These steps enable us to establish tight lower bounds on prediction error using a predetermined single-copy non-adaptive measurement strategy. 

\section{Problem setup of learning quantum dynamics in the view of quantum NFL theorem}\label{app_sec:NFL_Preliminaries}
For self-consistency, in this section we first recap the main results of the quantum NFL theorem which is achieved in the ideal setting by the study \cite{sharma2022reformulation}. Then we introduce the learning problems of the entanglement-assisted quantum NFL theorem in a realistic scenario for further elucidation.

\subsection{Quantum NFL theorem for learning quantum dynamics in the ideal setting}

In this subsection, we briefly recap the results achieved in Ref.~\cite{sharma2022reformulation} for self-consistency. The learning problem studied in Ref.~\cite{sharma2022reformulation} aims to learn the full representation of target unitary operator $\bm{U}$ through training a hypothesis unitary $\bm{V}_{\mathcal{S}_Q}$ on the training data 
\begin{equation}
	\mathcal{S}_Q = \left\{ (\ket{\bm{\psi}_j}, \ket{\bm{\phi}_j}): \ket{\bm{\psi}_j},\ket{\bm{\phi}_j}=(\bm{U}\otimes \mathbb{I}_d)\ket{\bm{\psi}_j} \in \mathcal{H}_{\mathcal{X},\mathcal{R}}  \right\}_{j=1}^N.
\end{equation}
In this regard, they adopt the averaged error in the trace norm of the output quantum states rather than the distance of classical measurement outputs $\mathrm{R}_{\bm{U}}(\bm{V}_{\mathcal{S}_Q})$ between the learned hypothesis and the target hypothesis as the risk function $\widetilde{\mathrm{R}}_{\bm{U}}(\bm{V}_{\mathcal{S}_Q})$ to quantify the accuracy of the hypothesis $\bm{V}_{\mathcal{S}_Q}$. Formally, the risk function is defined as 
\begin{eqnarray}\label{eq:risk_trace_norm}
	\widetilde{\mathrm{R}}_{\bm{U}}(\bm{V}_{\mathcal{S}_Q}) := && \int_{\bm{\phi} \sim \haar} \mathrm{d} \phi \frac{1}{4} \left\| \bm{U}\ket{\bm{\phi}}\bra{\bm{\phi}}\bm{U}^{\dagger} - \bm{V}_{\mathcal{S}_Q}\ket{\bm{\phi}}\bra{\bm{\phi}}\bm{V}_{\mathcal{S}_Q}^{\dagger} \right\|_1^2
	\nonumber \\
	= && 1- \frac{d+|\Tr(\bm{U}^{\dagger}\bm{V}_{\mathcal{S}_Q})|^2}{d(d+1)},
\end{eqnarray}
where $\ket{\bm{\phi}}\in \mathcal{H}_{\mathcal{X}}$ and the second equality follows directly Haar integration calculation. They showed that under the assumption of perfect training on the output states with the infinite measurements, i.e., $(\bm{U}\otimes \mathbb{I}_d)\ket{\bm{\psi}_j}=(\bm{V}_{\mathcal{S}_Q}\otimes \mathbb{I}_d)\ket{\bm{\psi}_j}$ for $j\in[N]$, then the unitary operators $\bm{U}^{\dagger}\bm{V}_{\mathcal{S}_Q}$ can be reduced to the form of 
\begin{equation}
	\bm{U}^{\dagger}\bm{V}_{\mathcal{S}_Q} = \left(\begin{array}{cc}
		\mathbb{I}_{N\cdot r} &  0  \\
		0 &  \mathrm{Y}
	\end{array}\right)
\end{equation}
where $Y \in \mathbb{S}\mathbb{U}(d-Nr)$ follows Haar distribution when the target unitary is Haar distributed. This leads to their main results as shown in the following theorem.

\begin{theorem-non}[Formulated theorem according to the results of \cite{sharma2022reformulation}]
	Let $\bm{U}$ and $\bm{V}_{\mathcal{S}_Q}$ be the target unitary sampled from Haar distribution and the learned hypothesis unitary on the entangled data $\mathcal{S}_Q$, such that the assumption of perfect training on the output states holds. the averaged risk function defined in Eqn.~(\ref{eq:risk_trace_norm}) over all target unitaries and training data yields
	\begin{equation}
		\mathbb{E}_{\bm{U}} \mathbb{E}_{\mathcal{S}_Q}\widetilde{\mathrm{R}}_{\bm{U}}(\bm{V}_{\mathcal{S}_Q}) \ge 1 - \frac{d+r^2N^2}{d(d+1)}.
	\end{equation}
\end{theorem-non}
It implies that the entangled data can remove the exponential cost of training data size in learning unitaries. Specifically, in the extreme cases of $r=d$ and $r=1$, the required training data size for reaching the zero risk of the former achieves an exponential advantage over that of the latter.

\subsection{Quantum NFL theorem for learning quantum dynamics in the realistic scenario} \label{app_subsec:learning_problems_QNFL}

Let us first recall the required notations for the learning problems of quantum NFL theorem. Let $\bm{U} \in \mathbb{S}\mathbb{U}(d)$ denote the target unitary and $\bm{O}=\ket{\bm{o}}\bra{\bm{o}}$ be the fixed projective measurement which acts on the quantum system $\mathcal{X}$. Denote $\mathcal{S}$ as the training data with the size $|\mathcal{S}|=N$. The entanglement-assisted protocol introduces a reference system $\mathcal{R}$ to prepare the entangled training states in the Hilbert space $\mathcal{H}_{\mathcal{X}\mathcal{R}}:=\mathcal{H}_{\mathcal{X}} \otimes \mathcal{H}_{\mathcal{R}}$ with $\dim(\mathcal{H}_{\mathcal{R}})=d$. The $m$-measurement outcomes $(\bm{o}_{j1}, \cdots, \bm{o}_{jm})$ of the observable $\bm{O}$ on the output states $(\bm{U} \otimes \mathbb{I}_d)\ket{\bm{\psi}_j}$ are collected to construct the response of the input $\ket{\bm{\psi}}$ by taking the average value $\bm{o}_j=\sum_{k=1}^m \bm{o}_{jk}/m$. This leads to the training data
\begin{equation}
	\mathcal{S} = \left\{ (\ket{\bm{\psi}_j}, \bm{o}_j) : \ket{\bm{\psi}_j}\in \mathcal{H}_{\mathcal{X}\mathcal{R}}, \bm{o}_j=\frac{1}{m}\sum_{k=1}^m \bm{o}_{jk} \right\}_{j=1}^N,
\end{equation}
where $\bm{o}_{jk}$ is the $k$-th measurement outcome on the state $(\bm{U} \otimes \mathbb{I}_d)\ket{\bm{\psi}_j}$ and $\mathbb{E}\bm{o}_j = \Tr[(\bm{U}^{\dagger}\bm{O}\bm{U} \otimes \mathbb{I}_d)\ket{\bm{\psi}_j}\bra{\bm{\psi}_j}]$.
From the Schmidt decomposition of a pure state, each $\ket{\bm{\psi}_j}$ can be represented as
\begin{equation}\label{eq:entangled_state}
	\ket{\bm{\psi}_j}=\sum_{k=1}^r \sqrt{c_{j, k}}\ket{\bm{\xi}_{j, k}}_{\mathcal{X}}\ket{\bm{\zeta}_{j, k}}_{\mathcal{R}},
\end{equation}
where $\sum_{k=1}^r c_{j, k}=1$ and hence $\ket{c_j}=(\sqrt{c_{j1}}, \cdots, \sqrt{c_{jr}})^{\top}$ forms a state vector in the Hilbert space $\mathcal{H}_r$. Suppose that all training states $\ket{\bm{\psi}_j}\in \mathcal{S}$ have the same Schmidt rank $r\in \{1,2,\cdots,d\}$ across the cut $\mathcal{H}_{\mathcal{X}} \otimes \mathcal{H}_{\mathcal{R}}$.

The problem of incoherent quantum dynamics learning aims to train a hypothesis unitary $\bm{V}_{\mathcal{S}}$ on the training data $\mathcal{S}$ such that it can approximate the output of the target unitary $\bm{U}$ under the observable $\bm{O}$ for given any unseen input state $\ket{\bm{\psi}} \in \mathcal{H}_d$. Namely, this leads to the target function with respect to the input state $\ket{\bm{\psi}}$ (or $\bm{\rho}=\ket{\bm{\psi}}\bra{\bm{\psi}}$ in the density matrix representation)
\begin{equation}
	\label{eq:learning_model}
	\mathrm{f}_{\bm{U}}(\bm{\psi})=\Tr(\bm{U}^{\dagger}\bm{O}\bm{U}\ket{\bm{\psi}}\bra{\bm{\psi}}).
\end{equation}
Similarly, we denote the learned hypothesis function as $\mathrm{h}_{\mathcal{S}}(\bm{\psi})=\Tr(\bm{V}_{\mathcal{S}}^{\dagger}\bm{O}\bm{V}_{\mathcal{S}}\ket{\bm{\psi}}\bra{\bm{\psi}})$. Then to evaluate the quality of the learned unitary $\bm{V}_{\mathcal{S}}$, we define the risk function $\mathrm{R}_{\bm{U}}(\bm{V}_{\mathcal{S}})$ as 
\begin{equation}\label{eq:app_risk_fun}
	\mathrm{R}_{\bm{U}}(\bm{V}_{\mathcal{S}}) = \mathbb{E}_{\ket{\bm{\psi}} \sim \haar} \Tr\left(\bm{O}\left(\bm{V}_{\mathcal{S}}\ket{\bm{\psi}}\bra{\bm{\psi}}\bm{V}_{\mathcal{S}}^{\dagger}-\bm{U}\ket{\bm{\psi}}\bra{\bm{\psi}}\bm{U}^{\dagger}\right)\right) ^2.
\end{equation} 

The standard quantum NFL theorem considers the averaged risk function over all possible training data, target unitaries, and optimization algorithms from which the learned hypothesis unitary has the same output as the target unitary on the training input states, which has the form of 
$\mathbb{E}_{\mathcal{S}} \mathbb{E}_{\bm{U}} \mathrm{R}_{\bm{U}}(\bm{V}_{\mathcal{S}})$.
Particularly, the uniform sampling of the target unitary in $\mathbb{S}\mathbb{U}(d)$ corresponds to the Haar distribution. While random pure states uniformly sampled from the Hilbert space $\mathcal{H}_{\mathcal{X}\mathcal{R}}$ follow Haar distribution, it can not be assumed that the states uniformly sampled from the set of entangled states defined in Eqn.~(\ref{eq:entangled_state}) follow Haar distribution in $\mathbb{S}\mathbb{U}(d^2)$ due to the restriction of fixed Schmidt ranks. On the other hand, we observe that the uniform distribution of the entangled training state $\ket{\bm{\psi}_j}$ should satisfy the following conditions: (i) the state vectors $\ket{\bm{\xi}_j}$ and $\ket{\bm{\zeta}_j}$ in the subsystem $\mathcal{X}$ and $\mathcal{R}$ should be uniformly distributed in the corresponding Hilbert space and are independent; (ii) the state vectors $\ket{\bm{\xi}_j}$ and $\ket{\bm{\xi}_k}$ ($\ket{\bm{\zeta}_j}$ and $\ket{\bm{\zeta}_k}$) in the subsystem $\mathcal{X}$ ($\mathcal{R}$) should be orthogonal for $j\ne k$; (iii) the vector consisting of the coefficients $\ket{\bm{c}_{j}}$ is also uniformly distributed.
We summarize these observations as the following construction rules of the random entangled training states. 

\begin{construct}\label{construct:1}
	Each training entangled state $\ket{\bm{\psi}_j}$ in the form of Eqn.~(\ref{eq:entangled_state}) is independently sampled by separately sampling the $r$ orthogonal Haar states $\ket{\bm{\xi}_{j1}}, \cdots, \ket{\bm{\xi}_{jr}} \in \mathbb{S}\mathbb{U}(d)$, $r$ orthogonal Haar states $\ket{\bm{\zeta}_{j1}}, \cdots, \ket{\bm{\zeta}_{jr}}  \in \mathbb{S}\mathbb{U}(d) $ and the Haar state $\ket{\bm{c}_j} \in \mathbb{S}\mathbb{U}(r)$. The random variables $\ket{\bm{\xi}_{jk}}$, $\ket{\bm{\zeta}_{jk'}}$ and $\ket{\bm{c}_j}$ are mutually independent. 	
\end{construct}	

\subsection{The differences between the quantum NFL theorem in the ideal and realistic scenario} \label{app_subsec:difference_QML}
In this section, we discuss the difference between the quantum NFL theorem in Ref.~\cite{sharma2022reformulation} and this study. In particular, we would like to clarify that the studied problem in our work differs from that studied in Ref.~\cite{sharma2022reformulation} in both the learning problem and the training data, namely, 
\begin{itemize}
	\vspace{-0.1cm}
	\item[1.] While Ref.~\cite{sharma2022reformulation} aims to predict the output states evolved by the unitary $\bm{U}$, we aim to learn the operator $\bm{U}^{\dagger}\bm{O}\bm{U}$ to accurately predict the expectation value of the observable $\bm{O}$ on the output states, which covers a wide class of important learning tasks in QML.
	\vspace{-0.1cm}
	\item[2.] While the training data in Ref.~\cite{sharma2022reformulation} refers to the pairs of input-output quantum states $(\ket{\bm{\psi}}, (\bm{U}\times \mathbb{I})\ket{\bm{\psi}})$, we employ the quantum-classical pairs $((\ket{\bm{\psi}}, \bm{o})$ as the training data, where $\bm{o}$ refers to the vectors of the measurement outcomes on the output state $(\bm{U}\times \mathbb{I})\ket{\bm{\psi}}$.
	\vspace{-0.1cm}
\end{itemize}
In this study, we would like to highlight the impact of entangled data on the learning performance for a class of learning task, i.e., learning $\Tr(\bm{O}\bm{U}\ket{\bm{\psi}}\bra{\bm{\psi}}\bm{U}^{\dagger})$, from a more practical setting of finite measurement. The consideration of finite measurement leads to fundamental differences in practical implications and theoretical proof. In particular, A critical challenge in leveraging quantum systems is the efficient extraction of information, which is notably impeded by finite measurements.
Such limitations could lead to markedly divergent conclusions in realizing potential quantum advantages, heavily dependent on whether these finite measurements are incorporated into the analysis or not \cite{dalzell2023quantum}. Theoretically, while the derivation of the quantum NFL in Ref.~\cite{sharma2022reformulation} primarily employs algebraic manipulation of input-output states and unitary operators, our study establishes a quantum NFL from an information-theoretic standpoint, which involves the utilization of many complicated techniques such as mutual information and Fano's inequality, as presented in the subsequent sections.

\section{Quantum NFL theorem for projective measurements (Theorem~1)}\label{app_sec:NFL_SM}
In the realistic setting where the aim is to learn the target unitary under a fixed observable $\mathrm{f}_{\bm{U}}(\bm{\rho})=\Tr(\bm{U}^{\dagger}\bm{O}\bm{U}\bm{\rho})$ with a limited number of measurements, the learned hypothesis $\mathrm{h}_{\mathcal{S}}(\bm{\rho})=\Tr(\bm{V}_{\mathcal{S}}^{\dagger}\bm{O}\bm{V}_{\mathcal{S}}\bm{\rho})$ can not achieve zero training error on the training data $\mathcal{S}$ due to the statistical error in measuring the quantum system. Consequently, the assumption of perfect training, which is the pivotal condition in the proof of the NFL theorem in previous literature \cite{poland2020no,sharma2022reformulation}, is not applicable. With this regard, we adopt the information-theoretic based technique, which does not rely on the perfect training assumption, to prove the quantum NFL theorem when the number of measurements is finite. The core idea is to ``reduce'' the learning problem to the multi-way hypothesis testing problem. This is equivalent to showing that the risk error of the quantum dynamics learning problem can be lower bounded by the probability of error in testing problems, which can develop tools for. We then employ Fano's lemma to derive the information-theoretic bound for this hypothesis testing problem. 

The organization of this section is as follows. We first discretize the class of target function through $2\varepsilon$-packing in Supplementary Note~\ref{app_subsec:discretizing_packing}. 
We demystify in  Supplementary Note~\ref{app_subsec:reduce_learn_hypo} how to reduce the learning problem to the hypothesis testing problem on the $2\varepsilon$-packing with Fano's inequality. In this regard, learning the target unitary can be reduced to identifying the corresponding index $X$ in the $2\varepsilon$-packing. Then, we elucidate how to obtain the results of Theorem~\ref{thm:formal_finite_measurement} by separately bounding the related terms in the Fano's inequality, namely, the mutual information $\mathrm{I}(X;\hat{X})$ between the target index $X$ and the estimated index $\hat{X}$, and the cardinality of the $2\varepsilon$-packing. The theoretical proof of Theorem~\ref{thm:formal_finite_measurement} is given in Supplementary Note~\ref{app_subsec:proof_theorem}. We provide the theoretical guarantee of the reduction from the learning problem to the hypothesis testing problem in Supplementary Note~\ref{subsec:proof_est_test}. Supplementary Note~\ref{subsec:proof_packing_observable_lem} and Supplementary Note~\ref{subsec:proof_sum_MI_upper_bound} present the theoretical proofs for bounding the cardinality of the $2\varepsilon$-packing and bounding the mutual information $\mathrm{I}(X,\hat{X})$, respectively. Finally, in Supplementary Note~\ref{app_subsec:tightness} we elucidate the tightness of the derived lower bound through establishing the connection to the results of quantum state tomography.

\subsection{Discretizing the class of target functions through \texorpdfstring{$2\varepsilon$}{Lg}-packing}\label{app_subsec:discretizing_packing}
In the quantum dynamics learning problem under a fixed observable $\bm{O}$ described in Supplementary Note~\ref{app_subsec:learning_problems_QNFL}, the goal is to identify the target function $\mathrm{f}_{\bm{U}}(\bm{\rho})=\Tr(\bm{U}^{\dagger}\bm{O}\bm{U}\bm{\rho})$ from the hypothesis set $\mathcal{F}=\{\mathrm{f}_{\bm{U}}(\bm{\rho})=\Tr(\bm{U}^{\dagger}\bm{O}\bm{U}\bm{\rho})| \bm{U} \in \mathbb{S}\mathbb{U}(d)\}$ according to the measurement output $\bm{o}$. This task is hard when $\mathcal{F}$ contains a large amount of very different operators. With this regard, we discretize the set of target functions $\mathcal{F}$ by equipping it with a local $\varepsilon$-packing, as shown in Supplementary Figure~\ref{fig:PackingNet}.  
\begin{definition}
	[$\varepsilon$-packing and local $\varepsilon$-packing]\label{def:packing_net} For a given set of functionals $\mathcal{F}$ and a distance metric $\varrho$ on this set, the \textit{$\varepsilon$-packing} $\mathcal{M}_{\varepsilon}(\mathcal{F},\varrho)$ is a discrete subset of $\mathcal{F}$ whose elements are guaranteed to be distant from each other by a distance greater than or equal $2\varepsilon$. Namely, for any element $\mathrm{f}_1, \mathrm{f}_2 \in \mathcal{M}_{\varepsilon}(\mathcal{F},\varrho)$, the distance between $\mathrm{f}_1$ and $\mathrm{f}_2$ satisfies
	\begin{equation}
		\varrho(\mathrm{f}_1, \mathrm{f}_2) \ge 2\varepsilon.
	\end{equation}  
	Similarly, the \textit{local $\varepsilon$-packing}  $\mathcal{M}_{\varepsilon}^{(\gamma \varepsilon)}(\mathcal{F},\varrho)$ is a discrete subset of $\mathcal{F}$ such that for any element $\mathrm{f}_1, \mathrm{f}_2 \in \mathcal{M}_{\varepsilon}^{(\gamma\varepsilon)}(\mathcal{F},\varrho)$, the distance between $\mathrm{f}_1$ and $\mathrm{f}_2$ satisfies
	\begin{equation}
		\gamma\varepsilon \ge	\varrho(\mathrm{f}_1, \mathrm{f}_2) \ge 2\varepsilon,
	\end{equation}  	
	where $\gamma\geq 2$.
\end{definition}

\begin{figure*}[htbp]
	\centering
	\includegraphics[width=0.88\textwidth]{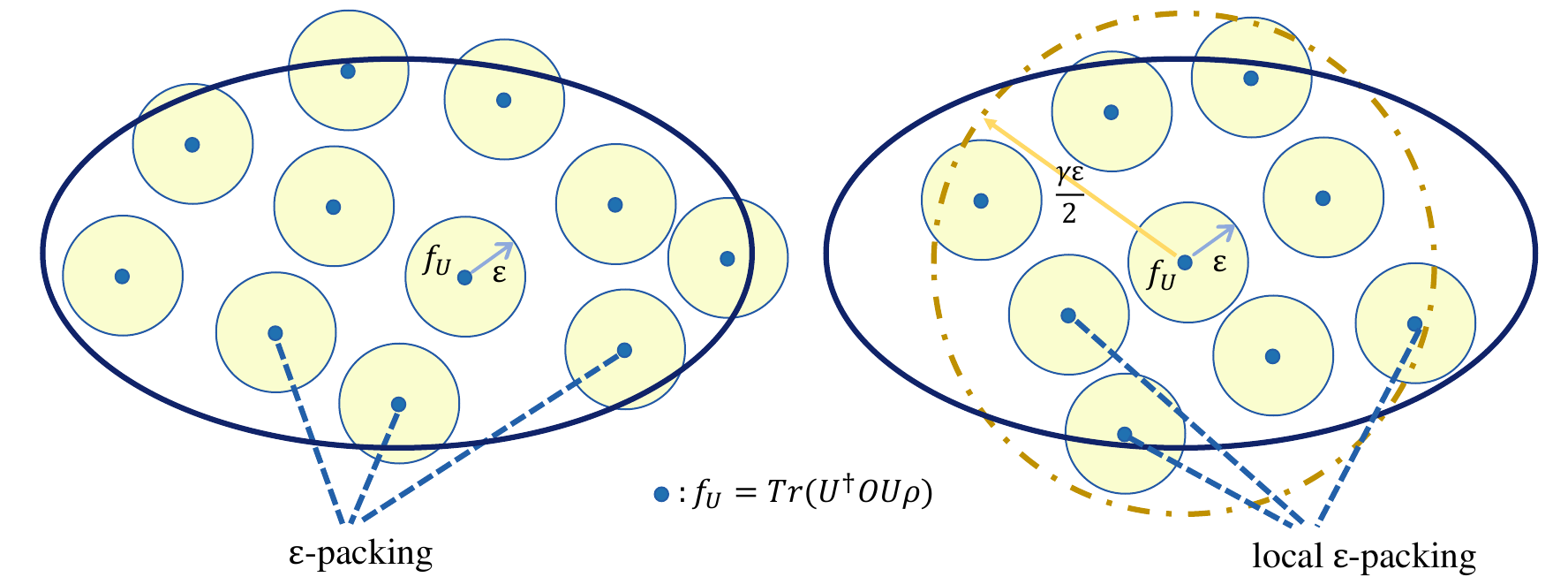}
	\caption{\small{\textbf{The $\varepsilon$-packing of the hypothesis space.}  The left panel and the right panel refer to the $\varepsilon$-packing and the local $\varepsilon$-packing with maximal distance $\gamma\varepsilon$ of $\mathcal{F}$, respectively.
	}}
	\label{fig:PackingNet}
\end{figure*}
For a sufficient large $\gamma$ obeying $\gamma\varepsilon \ge \max_{\mathrm{f}_1, \mathrm{f}_2\in\mathcal{F}}\varrho(\mathrm{f}_1,\mathrm{f}_2)$, the local $\varepsilon$-packing reduces to   $\varepsilon$-packing. In this work, we adopt the local $\varepsilon$-packing instead of the  $\varepsilon$-packing to obtain a tighter bound.  This is because the learned hypothesis $\mathrm{h}_{\mathcal{S}}$ only needs to approximate $\mathrm{f}_{\bm{U}}(\bm{\rho})=\Tr(\bm{U}^{\dagger}\bm{O}\bm{U}\bm{\rho})$ on the training input states $\{\bm{\rho}_j\}_{j=1}^N$ up to a small training error, which can be viewed as the relaxation of the perfect training assumption in the ideal setting. In particular, a small bounded training error can restrict the possible space in which the target unitary resides to the vicinity of the learned hypothesis. Consequently, the local $\varepsilon$-packing suffices to discretize the hypothesis set $\mathcal{F}$ in this context. 

The packing cardinality $|\mathcal{M}_{\varepsilon}^{(\gamma\varepsilon)}(\mathcal{F}, \varrho)|$ heavily depends on the choice of the distance metric $\varrho$, which is generally determined by the employed risk function. The following lemma gives the analytical expression of the metric $\varrho$ for the risk function defined in Eqn.~(\ref{eq:app_risk_fun}).
\begin{lemma}
	[Reformulation of the risk function]\label{lem:rewrite_risk}
	For any given projective measurement $\bm{O}=\ket{\bm{o}}\bra{\bm{o}}$ and any fixed random Haar state $\bm{\rho}$, the risk function defined in Eqn.~(\ref{eq:app_risk_fun}) has the following equivalent expression
	\begin{equation}\label{eq:rewrite_risk}
		\mathrm{R}_{\bm{U}}(\bm{V}_{\mathcal{S}}) 
		= \Phi \left(  \varrho(\mathrm{h}_{\mathcal{S}},\mathrm{f}_{\bm{U}}) \right),	
	\end{equation}
	where $\mathrm{f}_{\bm{U}}(\bm{\rho})=\Tr(\bm{U}^{\dagger}\bm{O}\bm{U}\bm{\rho})$, $\mathrm{h}_{\mathcal{S}}(\bm{\rho})=\Tr(\bm{V}_{\mathcal{S}}^{\dagger}\bm{O}\bm{V}_{\mathcal{S}}\bm{\rho})$, $\Phi(t)=t^2$, and
	\begin{equation}\label{eq:distance_metric}
		\varrho(\mathrm{h}_{\mathcal{S}},\mathrm{f}_{\bm{U}})= \frac{1}{\sqrt{2d(d+1)}}\|\bm{U}^{\dagger}\bm{O}\bm{U} - \bm{V}_{\mathcal{S}}^{\dagger}\bm{O}\bm{V}_{\mathcal{S}} \|_1.
	\end{equation}
\end{lemma}

\begin{proof}
	[Proof of Lemma~\ref{lem:rewrite_risk}]
	We recall that the risk function given in Eqn.~(\ref{eq:learning_model}) refers to
	\begin{equation}
		\mathrm{R}_{\bm{U}}(\bm{V}_{\mathcal{S}}) = \mathbb{E}_{\ket{\bm{\phi}} \sim \haar} \Tr\left(\bm{O}\left(\bm{V}_{\mathcal{S}}\ket{\bm{\phi}}\bra{\bm{\phi}}\bm{V}_{\mathcal{S}}^{\dagger}-\bm{U}\ket{\bm{\phi}}\bra{\bm{\phi}}\bm{U}^{\dagger}\right)\right) ^2.
	\end{equation}
	Denote $\bm{G}:=\bm{U}^{\dagger}\bm{O}\bm{U}-\bm{V}_{\mathcal{S}}^{\dagger}\bm{O}\bm{V}_{\mathcal{S}}$, we have  
	\begin{align}
		\mathrm{R}_{\bm{U}}(\bm{V}_{\mathcal{S}}) 
		& =  \mathbb{E}_{\ket{\bm{\phi}} \sim \haar} \Tr\left( \bm{G}  \ket{\bm{\phi}}\bra{\bm{\phi}} \right)^2 \nonumber \\
		& = \Tr\left( \bm{G}^{\otimes 2} \mathbb{E}_{\ket{\bm{\phi}} \sim \haar} \ket{\bm{\phi}}\bra{\bm{\phi}}^{\otimes 2} \right) \nonumber \\
		& = \frac{1}{d(d+1)}\Tr\left( \bm{G}^{\otimes 2} (\mathbb{I}^{\otimes 2} + \swap)\right) \nonumber \\
		& = \frac{1}{d(d+1)}\left(\Tr(\bm{G})^2+\Tr(\bm{G}^2)\right)\nonumber \\
		& = \frac{1}{d(d+1)}\left(2\Tr(\bm{O}^2)-2\Tr(\bm{U}^{\dagger}\bm{O}\bm{U}\bm{V}_{\mathcal{S}}^{\dagger}\bm{O}\bm{V}_{\mathcal{S}}) \right), \label{eq:app_risk_simplify_1}
	\end{align}
	where the second equality employs $\Tr(A)^2=\Tr(A\otimes A)$, the third equality exploits Lemma~\ref{lem:Tr(psipsi)Tr(psipsi)}, and the $\swap$ is the swap operator on the space $\mathcal{H}_d^{\otimes 2}$.
	When the observable $\bm{O}$ takes the projective measurement $\ket{\bm{o}}\bra{\bm{o}}$, the operators $\bm{U}^{\dagger}\bm{O}\bm{U}$ and $\bm{V}_{\mathcal{S}}^{\dagger}\bm{O}\bm{V}_{\mathcal{S}}$ can be treated as density matrix of the quantum state that we denote as $\ket{\bm{\mu}}\bra{\bm{\mu}}$ and $\ket{\bm{\nu}}\bra{\bm{\nu}}$ with $\ket{\bm{\mu}}=\bm{U}\ket{\bm{o}}$ and $\ket{\bm{\nu}}=\bm{V}\ket{\bm{o}}$ respectively. With this observation, we can rewrite the risk function as 
	\begin{eqnarray}
		\mathrm{R}_{\bm{U}}(\bm{V}_{\mathcal{S}}) 
		&& = \frac{1}{d(d+1)}\left(2\Tr(\ket{\bm{o}}\bra{\bm{o}}^2)-2\braket{\bm{\mu} | \bm{\nu}}^2 \right)
		\nonumber \\
		&& = \frac{2}{d(d+1)} \left(1-\braket{\bm{\mu} | \bm{\nu}}^2 \right)
		\nonumber \\
		&& = \frac{1}{2d(d+1)} \left\|\bm{U}^{\dagger}\bm{O}\bm{U} - \bm{V}_{\mathcal{S}}^{\dagger}\bm{O}\bm{V}_{\mathcal{S}} \right\|_1^2. \label{eq:app_risk_simplify_2}
	\end{eqnarray}
	where the first equality comes from substituting the density matrix representation of operators $\bm{O}, \bm{U}^{\dagger}\bm{O}\bm{U},  \bm{V}_{\mathcal{S}}^{\dagger}\bm{O}\bm{V}_{\mathcal{S}}$ into Eqn.~(\ref{eq:app_risk_simplify_1}), the second equality is based on the fact $\Tr(\bm{\rho}^2)=1$ for any pure state $\bm{\rho}$, the third equality employs the relation between trace distance and fidelity that $1-\braket{\bm{u}|\bm{v}}^2=\|\ket{\bm{u}}\bra{\bm{u}}- \ket{\bm{v}}\bra{\bm{v}} \|_1^2/4$.
	Eqn.~(\ref{eq:rewrite_risk}) can be directly obtained by denoting $\Phi(t)=t^2$, and $\varrho(\mathrm{h}_{\mathcal{S}}, \mathrm{f}_{\bm{U}})=\|\bm{U}^{\dagger}\bm{O}\bm{U} - \bm{V}_{\mathcal{S}}^{\dagger}\bm{O}\bm{V}_{\mathcal{S}} \|_1/\sqrt{2d(d+1)}$. 
\end{proof}

The definition of $\varrho$ in Eqn.~(\ref{eq:distance_metric}) satisfies the properties of the distance metric, i.e., (1) non-negativity: $\varrho(\mathrm{f}_{\bm{V}}, \mathrm{f}_{\bm{U}}) \ge 0$, (2) identity of indiscernible: $\varrho(\mathrm{f}_{\bm{V}}, \mathrm{f}_{\bm{U}}) \ge 0$ iff $\mathrm{f}_{\bm{V}}=\mathrm{f}_{\bm{U}}$, (3) symmetry:
$\varrho(\mathrm{f}_{\bm{V}}, \mathrm{f}_{\bm{U}}) = \varrho(\mathrm{f}_{\bm{U}}, \mathrm{f}_{\bm{V}})$, (4) triangle inequality: $\varrho(\mathrm{f}_{\bm{V}}, \mathrm{f}_{\bm{U}}) \le  \varrho(\mathrm{f}_{\bm{U}}, \mathrm{f}_{\bm{W}}) + \varrho(\mathrm{f}_{\bm{V}}, \mathrm{f}_{\bm{W}})$. With this well-defined distance metric $\varrho$ on the space $\mathcal{F}$, the local $2\varepsilon$-packing is denoted as $\mathcal{M}_{2\varepsilon}^{(\gamma\varepsilon)}$ with dropping the dependence on $\mathcal{F}$ and $\varrho$ for simplification. Without loss of generality, we employ the positive integer set  $\mathcal{X}_{2\varepsilon}^{(2\gamma\varepsilon)}=[|\mathcal{M}_{2\varepsilon}^{(2\gamma\varepsilon)}|]$ to index the elements in the local $2\varepsilon$-packing, i.e., $\mathcal{M}_{2\varepsilon}^{(2\gamma\varepsilon)}=\{\mathrm{f}_{\bm{U}_x}\}_{x\in \mathcal{X}_{2\varepsilon}^{(2\gamma\varepsilon)}}$, where each index $x$ in $\mathcal{X}_{2\varepsilon}^{(2\gamma\varepsilon)}$  uniquely corresponds to an element $\mathrm{f}_{\bm{U}_x}$ in the local packing $\mathcal{M}_{2\varepsilon}^{(2\gamma\varepsilon)}$. In the following, we refer the index set $\mathcal{X}_{2\varepsilon}^{(2\gamma\varepsilon)}$ to the local $2\varepsilon$-packing.

\subsection{Reducing the learning problem to hypothesis testing}\label{app_subsec:reduce_learn_hypo}
We now elucidate how to reduce the quantum dynamics learning problem introduced in Supplementary Note~\ref{app_subsec:learning_problems_QNFL} to the hypothesis testing problem by exploiting Fano's method, which is extensively used in deriving the lower bound of the generalization error in statistical learning theory \cite{duchi2016lecture}. 
Assume that the target function is in the discrete local $2\varepsilon$-packing $\{\mathrm{f}_{\bm{U}_x}\}_{x \in \mathcal{X}_{2\varepsilon}^{(2\gamma\varepsilon)}}$, then the learning problem aims to identify the underlying index $x$ from $\mathcal{X}_{2\varepsilon}^{(2\gamma\varepsilon)}$ according to the training data $\mathcal{S}$, which is exactly described by the hypothesis testing. Particularly, given the set of training input states $\{\bm{\rho}_j\}_{j=1}^N$, we refer the hypothesis testing to any measurable mapping $\Psi_{\bm{\rho}_1, \cdots, \bm{\rho}_N }: \bm{o} \to \mathcal{X}_{2\varepsilon}^{(2\gamma\varepsilon)}$. The learning problem can be reduced to the following hypothesis testing problem as shown in Supplementary Figure~\ref{fig:hypothesistesting}:

\begin{itemize}
	\item first, nature chooses $X=x$ according to the uniform distribution over $\mathcal{X}_{2\varepsilon}^{(2\gamma\varepsilon)}$ and the target function is denoted as $\mathrm{f}_{\bm{U}_x}$ with $\bm{U}_x$ being the unitary operator in $\{\bm{U}_{x'} \}_{x'\in \mathcal{X}_{2\varepsilon}^{(2\gamma\varepsilon)}}$;
	\item second, conditioned on the choice $X=x$, and given the Haar random input states set $\{\bm{\rho}_j=\ket{\bm{\psi}_j}\bra{\bm{\psi}_j}\}_{j=1}^N$, the measurement outcome is drawn from  distribution $\mathbb{P}_{\bm{o}|x}=\prod_{j=1}^N \mathbb{P}_{\bm{o}|x}^{(j)}$, where $\mathbb{P}_{\bm{o}|x}^{(j)}$ characterized the distribution of $\bm{o}_{j}=\sum_{k=1}^m \bm{o}_{jk}/m$ conditional on $X=x$. The set of input states  $\{\bm{\rho}_j\}_{j=1}^N$ and the measurement outcomes $\{\bm{o}_j\}_{j=1}^N$ together form the training data $\mathcal{S}=\{(\bm{\rho}_j, \bm{o}_j)\}_{j=1}^N$;
	\item  third, the learning model performs the following optimization to minimize the empirical training error  
	\begin{equation}
		\mathrm{h}_{\mathcal{S}}= \mathop{\arg\min}\limits_{\mathrm{h}\in \{\mathrm{f}_{\bm{U}_{x'}}|x' \in \mathcal{X}_{2\varepsilon}^{( 2\gamma\varepsilon)}\}} \frac{1}{N}\sum_{j=1}^N \left(\mathrm{h}(\bm{\rho}_j)-\bm{o}_j\right)^2;
	\end{equation}
	\item fourth,  the hypothesis testing $\Psi_{\bm{\rho}_1, \cdots, \bm{\rho}_N}(\bm{o}):= \arg \min_{\tilde{x} \in\mathcal{X}_{2\varepsilon}^{(2\gamma\varepsilon)}}\varrho(\mathrm{h}_{\mathcal{S}}, \mathrm{f}_{\bm{U}_{\tilde{x}}})$ is conducted to determine the randomly chosen index $X$. The associated error probability of the hypothesis testing problem is denoted by  $\mathbb{P}_{\bm{o},X}(\Psi_{\bm{\rho}_1, \cdots, \bm{\rho}_N}(\bm{o})\ne X)$.  
\end{itemize}

\begin{figure*}[htbp]
	\centering         		\includegraphics[width=176mm]{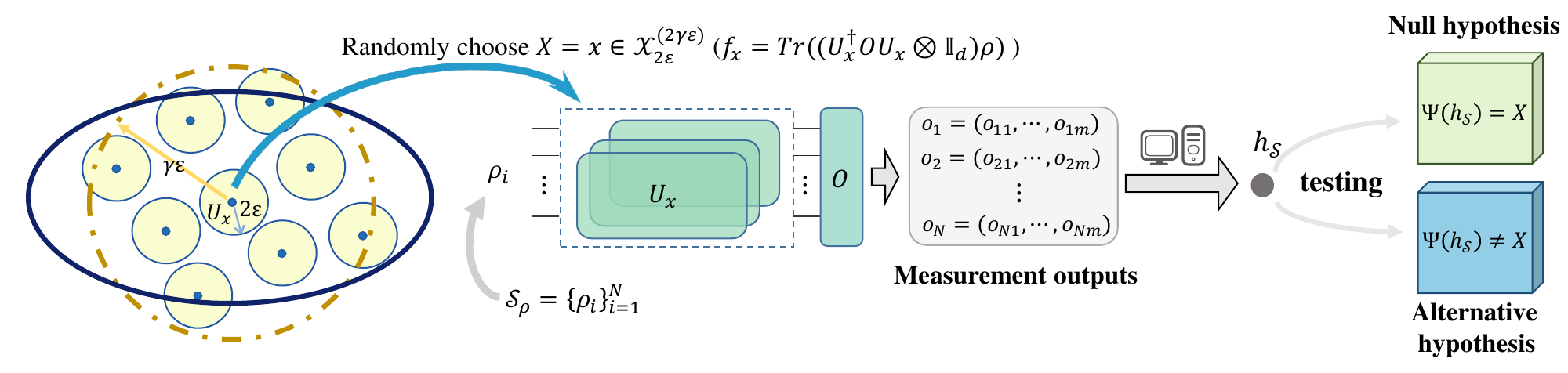}
	\caption{\small{\textbf{The paradigm of the reduction from quantum dynamics learning to hypothesis testing.  
	}}}
	\label{fig:hypothesistesting}
\end{figure*}

The choice of the uniform distribution over $\mathcal{X}_{2\varepsilon}^{(2\gamma\varepsilon)}$ in the first step is because in the context of the quantum NFL theorem the target unitary is assumed to be sampled from the Haar distribution. The distribution $\mathbb{P}_{\bm{o}|x}^{(j)}$ of the measurement outcome $\bm{o}_{jk}$ in the second step is assumed to be a Gaussian distribution whose mean $u_x=\Tr((\bm{U}_x^{\dagger}\bm{O}\bm{U}_{x}\otimes \mathbb{I}_{\mathcal{R}})\bm{\rho}_j)$ varies with the index $x$ and the input state $\bm{\rho}_j$, and the variance $\sigma^2$ is assumed to be an $n$-dependent constant which is identified below (see Assumption \ref{assu:normal_distribution}). This learning algorithm in the third step is arbitrary, as long as the learned function can achieve the minimal training error (i.e., can be greater than $0$), which can be regarded as the relaxation of the perfect training assumption with zero training error. In the fourth step, the notation $\mathbb{P}_{\bm{o},X}$ denotes the joint distribution over the random index $X$ and $\bm{o}$. In particular, define $\bar{\mathbb{P}}_{\bm{o}}=\sum_{x\in\mathcal{X}_{2\varepsilon}^{(2\gamma\varepsilon)}} \mathbb{P}_{\bm{o}|x}/|\mathcal{X}_{2\varepsilon}^{(2\gamma\varepsilon)}|$ as the mixture distribution, the measurement output is drawn (marginally) from $\bar{\mathbb{P}}_{\bm{o}}$, and our hypothesis testing problem is to determine the randomly chosen index $X$  given the training data $\mathcal{S} = \{(\bm{\rho}_j, \bm{o}_j)\}_{j=1}^N$ with $\bm{o}=(\bm{o}_1, \cdots, \bm{o}_N)$ sampled from this mixture $\bar{\mathbb{P}}_{\bm{o}}$. Denote the average risk function over the training sets $\mathcal{S}$ and target unitary $\bm{U}$ as 
\begin{equation}\label{eq:minimax_risk}
	\mathbb{E}_{\bm{U}}\mathbb{E}_{\mathcal{S}} \mathrm{R}_{\bm{U}}(\bm{V}_{\mathcal{S}}) = \mathrm{R}_N(\Phi \odot \varrho) := \mathbb{E}_{\bm{U}}\mathbb{E}_{\mathcal{S}} \Phi \left(  \varrho(\mathrm{h}_{\mathcal{S}},\mathrm{f}_{\bm{U}}) \right),
\end{equation}
where the second equality follows Lemma~\ref{lem:rewrite_risk}. Notably, the randomness of $\mathcal{S}$ comes from the randomness of the input states set $\{\bm{\rho}_{j}\}_{j=1}^N$ and the randomness of measurement outcomes $\bm{o}_j$. For simplifying the calculation of Haar integration, we make the following mild assumption on the distribution of the measurement output $\bm{o}_j=\sum_{k=1}^m \bm{o}_{jk}/m$.
\begin{assu}\label{assu:normal_distribution}
	For any index $X=x$, the outcome $\bm{o}_j=\sum_{k=1}^m \bm{o}_{jk}/m$ of the projective measurements $\bm{O}=\ket{\bm{o}}\bra{\bm{o}}$ on the given output state $(\bm{U}_X \otimes \mathbb{I}_{\mathcal{R}})\ket{\bm{\psi}_j}$ follows the binomial distribution $\mathbb{B}(m, u_x^{(j)})$ with $u_x=\mathbb{E}[\bm{o}_j]=\Tr((\bm{U}_x^{\dagger}\bm{O}\bm{U}_x\otimes \mathbb{I}_{\mathcal{R}}) \ket{\bm{\psi}_j}\bra{\bm{\psi}_j})$. According to the central limit theorem,    the binomial distribution is approximated by the normal distribution $\mathbb{N}(u_x, (\sigma_x^{(j)})^2)$ with the mean $u_x = \Tr((\bm{U}_x^{\dagger}\bm{O}\bm{U}_x\otimes \mathbb{I}_{\mathcal{R}}) \ket{\bm{\psi}_j}\bra{\bm{\psi}_j} )$ and the variance $(\sigma_x^{(j)})^2=u_x(1-u_x)/m$. Furthermore, the variance is assumed to take the expectation of $(\sigma_x^{(j)})^2$ over the random input state $\ket{\bm{\psi}_j}$, i.e., $\sigma^2=\mathbb{E}_{\ket{\bm{\psi}_j} \sim \haar}(\sigma_x^{(j)})^2$. 
\end{assu}

\noindent\textbf{Remark}. We note that Assumption \ref{assu:normal_distribution}  is mild, as the convergence of binomial distribution to the normal distribution is guaranteed by the central limit theorem for a large number of measurements \cite{loeve2017probability}. Additionally, the approximate proxy of the variance $\sigma_x^2$ with its expectation $\sigma^2$ is based on the observation that the variance $\sigma_x^2$ is smaller than the expectation of $u_x$ with at least a multiplier factor of $1/m$ and is exponentially concentrated in the number of qubits $n$ which is widely studied in barren plateaus of variation quantum algorithms \cite{mcclean2018barren, cerezo2021cost, zhang2021toward}. These observations enable the effective  discrimination of the normal distribution of the measurement outputs corresponding to the different hypothesis unitaries in the local $2\varepsilon$-packing.

\medskip	
With this setup, the following lemma whose proof is deferred to Supplementary Note~\ref{subsec:proof_est_test} gives a theoretical guarantee of the reduction from the quantum dynamics learning problem to the hypothesis testing problem.
\begin{lemma}\label{lem:est_test}
	Given the local $2\varepsilon$-packing $\mathcal{X}_{2\varepsilon}^{(2\gamma\varepsilon)}$ of the functional class $\mathcal{F}$, the average risk error in Eqn.~(\ref{eq:minimax_risk}) is lower bounded by
	\begin{equation}\label{eq:est_test}
		\mathrm{R}_N(\Phi \odot \varrho) \ge \mathbb{E}_{\bm{\rho}_1, \cdots, \bm{\rho}_N} \Phi(\varepsilon)  \mathbb{P}_{\bm{o},X}(\Psi_{\bm{\rho}_1, \cdots, \bm{\rho}_N} (\bm{o})\ne X),
	\end{equation}
	where the probability measure $\mathbb{P}$ refers to the joint distribution of the random index $X$ and the measurement output $\bm{o}$.
\end{lemma}

Lemma~\ref{lem:est_test} reduces the problem of lower bounding the risk function $\mathrm{R}_N(\Phi \odot \varrho)$ to the problem of lower bounding the error of hypothesis testing. The Fano's inequality gives an information-theoretical bound to the latter.

\begin{lemma}[Fano's inequality]\label{lem:Fano}
	Assume that $X$ is uniform in $\mathcal{X}_{2\varepsilon}^{(2\gamma\varepsilon)}$. The learning procedure can be depicted by the Markov chain $X \to \mathcal{S} \to \hat{X}$, where $\hat{X}$ is returned by the hypothesis testing, i.e., $\Phi_{\bm{\rho}_1, \cdots, \bm{\rho}_N} (\bm{o})=\hat{X}$. Then we have
	\begin{equation}\label{eq:Fano}
		\mathbb{P}_{\bm{o},X}(\Psi_{\bm{\rho}_1, \cdots, \bm{\rho}_N}(\bm{o})\ne X) \ge 1 - \frac{\mathrm{I}(X;\hat{X})+\log 2}{\log(|\mathcal{X}_{2\varepsilon}^{(2\gamma\varepsilon)}|)},
	\end{equation}
	where $\mathrm{I}(X; \hat{X})$ refers to the mutual information between the estimated index $\hat{X}$ and $X$.
\end{lemma}

The information-theoretic lower bound for the risk function defined in Eqn.~(\ref{eq:minimax_risk}) can be achieved by using Fano's method for multiple hypothesis testing established in Lemma~\ref{lem:est_test} and Lemma~\ref{lem:Fano}. Moreover, it is reduced to separately bound the mutual information $\mathrm{I}(X;\hat{X})$ and the local $2\varepsilon$-packing cardinality $|\mathcal{X}_{2\varepsilon}^{(2\gamma\varepsilon)}|$, which are given by the following two lemmas. 

\begin{lemma}
	[Lower bound of local $2\varepsilon$-packing cardinality for the output under the projective measurement]\label{lem:packing_observable_lem}
	Let $\bm{O}=\ket{\bm{o}}\bra{\bm{o}}$ be the projective measurement, $\mathcal{F}= \{\mathrm{f}_{\bm{U}}: \bm{\rho} \to \Tr(\bm{U}^{\dagger}\bm{O}\bm{U}\bm{\rho})|U\in \mathbb{S}\mathbb{U}(d)\}$ be the function class of the output of quantum system given an arbitrary fixed Haar state $\bm{\rho}$, and $\varrho(\mathrm{f}_{\bm{U}_1}, \mathrm{f}_{\bm{U}_2}) =  \sqrt{\mathbb{E}_{\bm{\rho} \sim \haar}\Tr\left(\bm{O}\left(\bm{U}_1\bm{\rho} \bm{U}_1^{\dagger}-\bm{U}_2\bm{\rho} \bm{U}_2^{\dagger}\right)\right) ^2 }$ be the distance measure. Then there exist a local $2\varepsilon$-packing  $\mathcal{X}_{2\varepsilon}^{(2\gamma\varepsilon)}$ with $\varepsilon=\mathcal{O}(1/d)$ in the $\varrho$-metric such that the packing cardinality yields
	\begin{equation}\label{eq:packing_observable}
		\left|\mathcal{X}_{2\varepsilon}^{(2\gamma\varepsilon)} \right|\ge \exp\left(\frac{d\min\{(1-2\tilde{\varepsilon})^2, (4\gamma^2\tilde{\varepsilon}^2-1)^2\}}{16}\right),
	\end{equation}
	where $\tilde{\varepsilon}=2\sqrt{2d(d+1)}\varepsilon$ and $\gamma$ refers to an arbitrary constant obeying $0<4\gamma^2\tilde{\varepsilon}^2-1<1$ and $\gamma>2$. 
\end{lemma}

\begin{lemma}
	[Upper bound of the mutual information $\mathrm{I}(X;\hat{X})$]\label{lem:sum_MI_upper_bound}
	Following the notations in Lemma~\ref{lem:packing_observable_lem} and Assumption~\ref{assu:normal_distribution}, the average of mutual information over the training states $\{\bm{\rho}_j\}_{j=1}^N$ sampled from the distribution described in Construction Rule \ref{construct:1} yields
	\begin{equation}
		\mathbb{E}_{\bm{\rho}_1,\cdots, \bm{\rho}_N}\mathrm{I}(X;\hat{X}) \le N\cdot \min \left\{\frac{4md\gamma^2\varepsilon^2}{r}, r\log\left(d\right)  \right\},
	\end{equation}
	where $\gamma$ refers to an arbitrary constant obeying $0<4\gamma^2\tilde{\varepsilon}^2-1<1$ with $\tilde{\varepsilon}=2\sqrt{2d(d+1)}{\varepsilon}$ and $\gamma>2$.
\end{lemma}

Hence, learning quantum dynamics in the framework of the quantum NFL theorem taking consideration of the entangled quantum state and the finite measurements is encapsulated in the following theorem.

\begin{theorem}[Formal statement of Theorem~1 in the maintext]
	\label{thm:formal_finite_measurement}
	Let $\{ \mathrm{f}_{\bm{U}_{x}}\}_{x \in \mathcal{X}_{2\varepsilon}^{(2\gamma\varepsilon)}}$ be a local $2\varepsilon$-packing with the maximal distance being $\gamma\varepsilon$ of the function class $\mathcal{F}$ in the $\varrho$-metric. Assume that the index $X$ corresponding to the target function $\mathrm{f}_{\bm{U}_{X}}$ is uniformly sampled from the set $\mathcal{X}_{2\varepsilon}^{(2\gamma\varepsilon)}$. Conditional on $X=x$, we obtain the training set $\mathcal{S}=\{\bm{\rho}_j, \bm{o}_j \}_{j=1}^N$ where $\bm{\rho}_j$ is the random entangled state of Schmidt rank $r$ sampled from the distribution described in Construction Rule \ref{construct:1}, and $\bm{o}_j=\sum_{k=1}^m\bm{o}_{jk}/m$ is the measurement output of the observable $\bm{O}$ following Assumption~\ref{assu:normal_distribution}. Then the averaged risk function defined in Eqn.~(\ref{eq:minimax_risk}) is lower bounded by
	\begin{equation}
		\mathbb{E}_{\bm{U}}\mathbb{E}_{\mathcal{S}} \mathrm{R}_{\bm{U}}(\bm{V}_{\mathcal{S}}) \ge \frac{\tilde{\varepsilon}^2}{8d(d+1)} \left(1- \frac{N\cdot \min\{128m\tilde{\varepsilon}^2/r(d+1), r\log(d))\}+16\log 2 }{d\min\{(1-2\tilde{\varepsilon})^2, (64\tilde{\varepsilon}^2-1)^2\}}  \right),
	\end{equation}
	where $\tilde{\varepsilon}=2\sqrt{2d(d+1)}{\varepsilon}$ and the expectation is taken over all target unitaries $\bm{U}$.
\end{theorem}

\subsection{Proof of Theorem~\ref{thm:formal_finite_measurement}}\label{app_subsec:proof_theorem}

We are now ready to prove Theorem~\ref{thm:formal_finite_measurement}.
\begin{proof}
	[Proof of Theorem~\ref{thm:formal_finite_measurement}]
	Combining Lemma~\ref{lem:Fano} with the reduction from learning to testing in Lemma~\ref{lem:est_test}, we have
	\begin{align}
		\mathbb{E}_{\bm{U}}\mathbb{E}_{\mathcal{S}} \mathrm{R}_{\bm{U}}(\bm{V}_{\mathcal{S}}) = & \mathrm{R}_N(\Phi \odot \varrho )
		\nonumber \\
		\ge & \mathbb{E}_{\bm{\rho}_1,\cdots, \bm{\rho}_N} \Phi(\varepsilon)  \mathbb{P}(\Psi_{\bm{\rho}_1,\cdots, \bm{\rho}_N}(\bm{o})\ne X) 
		\nonumber \\
		\ge & \varepsilon^2 \cdot \left(1-\frac{\mathbb{E}_{\bm{\rho}_1,\cdots, \bm{\rho}_N}\mathrm{I}(X;\hat{X})+\log2}{\log(|\mathcal{X}_{2\varepsilon}^{(2\gamma\varepsilon)}|)} \right).
	\end{align}
	Let $\gamma=4$ and $\varepsilon=\tilde{\varepsilon}/(2\sqrt{2d(d+1)})$ with $\tilde{\varepsilon}$ taking any number satisfying $\gamma \tilde{\varepsilon} \in (1/2,\sqrt{2}/2)$ such that $0<4\gamma^2\tilde{\varepsilon}^2-1<1$. 
	Employing the results of Lemma~\ref{lem:packing_observable_lem} and Lemma~\ref{lem:sum_MI_upper_bound}, we have
	\begin{align}
		\mathrm{R}_N(\Phi \odot \varrho ) 
		& \ge \varepsilon^2 \left(1- \frac{N\cdot \min\{4md\gamma^2\varepsilon^2/r, r\log(d)\}+\log 2}{d\min\{(1-2\tilde{\varepsilon})^2, (4\gamma^2\tilde{\varepsilon}^2-1)^2\}/16}  \right)
		\nonumber \\
		& \ge \frac{\tilde{\varepsilon}^2}{8d(d+1)} \left(1- \frac{N\cdot \min\{8md\tilde{\varepsilon}^2/rd(d+1), r \log(d)\}+\log 2}{d\min\{(1-2\tilde{\varepsilon})^2, (64\tilde{\varepsilon}^2-1)^2\}/16}  \right)
		\nonumber \\
		& = \frac{\tilde{\varepsilon}^2}{8d(d+1)} \left(1- \frac{N\cdot \min\{128m\tilde{\varepsilon}^2/r(d+1), 16r\log(d)\}+16\log 2}{d\min\{(1-2\tilde{\varepsilon})^2, (64\tilde{\varepsilon}^2-1)^2\}}  \right).
	\end{align}
	where the first inequality employs Lemma~\ref{lem:packing_observable_lem} that $\log(|\mathcal{X}_{2\varepsilon}^{(2\gamma\varepsilon)}|) \ge  d\min\{(1-2\tilde{\varepsilon})^2, (4\gamma^2\tilde{\varepsilon}^2-1)^2\}/16$.
	This completes the proof.
\end{proof}

\subsection{Proof of Lemma~\ref{lem:est_test}-reducing the learning problem to hypothesis testing}\label{subsec:proof_est_test}
\begin{proof}[Proof of Lemma~\ref{lem:est_test}]
	To see this result, we first recall that the averaged risk given in Eqn.~(\ref{eq:minimax_risk}) has the form 
	\begin{equation}
		\mathrm{R}_N(\Phi \odot \varrho)
		= \mathbb{E}_{\bm{U}}\mathbb{E}_{\mathcal{S}} [\Phi(\varrho(\mathrm{h}_{\mathcal{S}},\mathrm{f}_{\bm{U}}))], \nonumber
	\end{equation}
	where $\Phi(t)=t^2$, $\varrho$ is the distance metric on the functionals set $\mathcal{F}$ with the definition given in Eqn.~(\ref{eq:distance_metric}), and $\mathrm{h}_{\mathcal{S}}$ refers to an arbitrary learned hypothesis on the training data $\mathcal{S}$. 
	Denote $\Theta$ as the value space of $\bm{o}$ and $\mathbb{P}_{\bm{o}}$ as the cumulative distribution function (CDF) of $\bm{o}$. Decomposing the integration with respect to $\mathcal{S}$ into that with respect to $\{\bm{\rho}_j\}_{j=1}^N$ and $\bm{o}$  yields
	\begin{align}
		\mathrm{R}_N(\Phi \odot \varrho)
		=&  \mathbb{E}_{\bm{U}}\mathbb{E}_{\bm{\rho}_1, \cdots, \bm{\rho}_N} \mathbb{E}_{\bm{o}} [\Phi(\varrho(\mathrm{h}_{\mathcal{S}},\mathrm{f}_{\bm{U}}))]
		\nonumber \\
		=&  \mathbb{E}_{\bm{U}}\mathbb{E}_{\bm{\rho}_1, \cdots, \bm{\rho}_N} \int_{\Theta} \Phi(\varrho(\mathrm{h}_{\mathcal{S}},\mathrm{f}_{\bm{U}})) \mathrm{d} \mathbb{P}_{\bm{o}}
		\nonumber \\
		=&  \mathbb{E}_{\bm{U}}\mathbb{E}_{\bm{\rho}_1, \cdots, \bm{\rho}_N} \int_{ \{\bm{o}:\varrho(\mathrm{h}_{\mathcal{S}},\mathrm{f}_{\bm{U}})\ge 0\} } \Phi(\varrho(\mathrm{h}_{\mathcal{S}},\mathrm{f}_{\bm{U}})) \mathrm{d} \mathbb{P}_{\bm{o}}
		\nonumber \\ 
		\ge&  \mathbb{E}_{\bm{U}}\mathbb{E}_{\bm{\rho}_1, \cdots, \bm{\rho}_N} \int_{\{\bm{o}:\varrho(\mathrm{h}_{\mathcal{S}},\mathrm{f}_{\bm{U}})\ge \varepsilon\} } \Phi(\varepsilon) \mathrm{d} \mathbb{P}_{\bm{o}}
		\nonumber \\ 
		=  &  \mathbb{E}_{\bm{U}}\mathbb{E}_{\bm{\rho}_1, \cdots, \bm{\rho}_N} \Phi(\varepsilon)\mathbb{P}_{\bm{o}}(\varrho(\mathrm{h}_{\mathcal{S}},\mathrm{f}_{\bm{U}}) \ge \varepsilon),
		\nonumber \\ 
		=  &  \mathbb{E}_{\bm{\rho}_1, \cdots, \bm{\rho}_N}\frac{1}{|\mathcal{X}_{2\varepsilon}^{(2\gamma\varepsilon)}|}\sum_{x\in \mathcal{X}_{2\varepsilon}^{(2\gamma\varepsilon)}} \Phi(\varepsilon)\mathbb{P}_{\bm{o}}(\varrho(\mathrm{h}_{\mathcal{S}},\mathrm{f}_{\bm{U}_X}) \ge \varepsilon | X=x), \label{eq:est_test_1}
	\end{align}
	where the first equality follows that the randomness of $\mathcal{S}$ comes from $\{\bm{\rho}_j\}_{j=1}^N$ and $\bm{o}$, the second equality employs the integral representation with respect to $\bm{o}$ over the integration area $\Theta$, the third equality follows that the inequality $\varrho(\mathrm{h}_{\mathcal{S}}, \mathrm{f}_{\bm{U}})\ge 0 $ holds for any $\bm{o}\in \Theta$ and hence $\Theta=\{\bm{o}:\varrho(\mathrm{h}_{\mathcal{S}},\mathrm{f}_{\bm{U}})\ge 0\}$ with $\mathrm{h}_{\mathcal{S}}$ being $\bm{o}$-dependent, the first inequality follows that $\Phi$ is non-negative and non-decreasing in the integral region, the fourth equality uses the independence of $\Phi(\varepsilon)$ on $\bm{o}$ and the probability representation $\int_A \mathrm{d} \mathbb{P}_{\bm{o}}=\mathbb{P}_{\bm{o}}(A)$ with $A=\{\bm{o}:\varrho(\mathrm{h}_{\mathcal{S}},\mathrm{f}_{\bm{U}_x}) \ge \varepsilon\}$, the final equality follows the equivalence of the expectation under the Haar distribution on $\mathbb{S}\mathbb{U}(d)$ and the uniform distribution on the packing. 
	
	To obtain Eqn.~(\ref{eq:est_test}) with Eqn.~(\ref{eq:est_test_1}), we now only need to prove the inequality 
	\[\mathbb{P}_{\bm{o}}(\varrho(\mathrm{h}_{\mathcal{S}},\mathrm{f}_{\bm{U}_X}) \ge \varepsilon | X=x ) \ge \mathbb{P}_{\bm{o}}(\Psi_{\bm{\rho}_1, \cdots, \bm{\rho}_N}(\bm{o})\ne X | X=x)\] for any index $x \in \mathcal{X}_{2\varepsilon}^{(2\gamma\varepsilon)}$.  		Denote two events $A=\{\bm{o}:\varrho(\mathrm{h}_{\mathcal{S}},\mathrm{f}_{\bm{U}_x}) \ge \varepsilon\}$ and $B=\{\bm{o}: \Psi_{\bm{\rho}_1, \cdots, \bm{\rho}_N}(\bm{o}) \ne x\}$, the inequality can be achieved by proving $B \subset A$ or equivalently $A^C \subset B^C$, where $A^C, B^C$ denote the complement of $A$ and $B$. Particularly, the inequality $\mathbb{P}_{\bm{o}}(A)\ge \mathbb{P}_{\bm{o}}(B)$ holds if and only if $B \subset A$ or equivalently $A^C \subset B^C$. In the following,
	we complete this proof by proving the inclusion relationship of $A^C \subset B^C$, i.e., $\varrho(\mathrm{h}_{\mathcal{S}},\mathrm{f}_{\bm{U}_x}) < \varepsilon$ implies that $\Psi_{\bm{\rho}_1, \cdots, \bm{\rho}_N}(\bm{o})=x$.
	
	Assume that there exists an index $x$ such that $\varrho(\mathrm{h}_{\mathcal{S}},\mathrm{f}_{\bm{U}_x}) < \varepsilon$, then for any $x'\ne x$ with $x' \in \mathcal{X}_{2\varepsilon}^{(2\gamma\varepsilon)}$, we have
	\begin{equation}
		\varrho(\mathrm{h}_{\mathcal{S}},\mathrm{f}_{\bm{U}_{x'}}) \ge \varrho(\mathrm{f}_{\bm{U}_{x'}},\mathrm{f}_{\bm{U}_x}) - \varrho(\mathrm{h}_{\mathcal{S}},\mathrm{f}_{\bm{U}_{x}}) \ge \varepsilon \ge \varrho(\mathrm{h}_{\mathcal{S}},\mathrm{f}_{\bm{U}_x})
	\end{equation}
	where the first inequality employs the triangle inequality, the second inequality follows that $\varrho(\mathrm{f}_{\bm{U}_{x'}},\mathrm{f}_{\bm{U}_x}) \ge 2\varepsilon$ for any $x\ne x' \in \mathcal{X}_{2\varepsilon}^{(2\gamma\varepsilon)}$, the last inequality is based on the presupposition. This implies that $x=\arg \min_{\tilde{x} \in \mathcal{X}_{2\varepsilon}^{(2\gamma\varepsilon)}}\varrho(\mathrm{h}_{\mathcal{S}},  \mathrm{f}_{\bm{U}_{\tilde{x}}})$ and hence $\Psi_{\bm{\rho}_1, \cdots, \bm{\rho}_N}(\bm{o})=x$ according to the definition of $\Psi_{\bm{\rho}_1, \cdots, \bm{\rho}_N}(\bm{o})$.
\end{proof}

\subsection{Proof of Lemma~\ref{lem:packing_observable_lem}---bounding the local \texorpdfstring{$2\varepsilon$}{Lg}-packing cardinality \texorpdfstring{$|\mathcal{X}_{2\varepsilon}^{(2\gamma\varepsilon)}|$}{Lg}}\label{subsec:proof_packing_observable_lem}

The idea of lower bounding the local packing cardinality of the metric space $(\mathcal{F}, \varrho)$ in Lemma~\ref{lem:packing_observable_lem} can be decomposed into two steps. First, we derive a lower bound of the local packing cardinality of the space consisting of $n$-qubit pure states $\mathcal{H}_{d}$. This proof is given by Lemma~\ref{lem:proj_concentration} and Lemma~\ref{lem:packing_states_exist} in a probabilistic existence argument. Second, the local packing cardinality of metric spaces $(\mathcal{F}, \varrho)$ can be obtained by building the connection between metric spaces $(\mathcal{H}_{d}, \|\cdot\|_1)$ and $(\mathcal{F}, \varrho)$.

In the probabilistic method for bounding the local $\varepsilon$-packing cardinality, a sequence of independent and identically distributed (i.i.d) unitary operators $W_1, W_2, \cdots$ are sampled from the Haar distribution on $\mathbb{S}\mathbb{U}(d)$  to construct the $\varepsilon$-packing with  the states $\bm{\rho}_{\varepsilon,W_i}=W_i\ket{\bm{\phi}} \bra{\bm{\phi}} W_i^{\dagger}$ where $\ket{\bm{\phi}}\in\mathcal{H}_d$ is any fixed quantum state. We then apply standard concentration of measure results to argue that the probability of selecting an undesirable state set (i.e., there exist two states whose trace distance is less than $2\varepsilon$ or larger than $\gamma\varepsilon$) is exponentially small. This in turn implies that such states set is a local $\varepsilon$-packing with the maximal distance being $\gamma\varepsilon$. To this end, we first exploit the concentration of projector overlaps, which have been used in deriving the lower bounds for quantum state tomography \cite{haah2016sample, aaronson2018shadow, huang2021information, lowe2022lower}.

\begin{lemma}[Lemma~3.2, \cite{lowe2022lower}]\label{lem:proj_concentration}
	Let $W\in \mathbb{S}\mathbb{U}(d)$ be a Haar-random unitary operator and let $\Pi_1, \Pi_2: \mathcal{H}_d \to \mathcal{H}_d$ be orthogonal projection operators with the rank $r_1,r_2$ respectively. For all $t \in (0,1)$ it holds that
	\begin{align}
		\mathbb{P}_{W \sim \haar}\left[ \Tr(\Pi_1 W \Pi_2 W^{\dagger}) \le (1-t) \frac{r_1r_2}{d} \right] \le \exp \left(-\frac{r_1r_2t^2}{2}\right),
		\nonumber \\
		\mbox{and}~~\mathbb{P}_{W \sim \haar}\left[ \Tr(\Pi_1 W \Pi_2 W^{\dagger}) \ge (1+t) \frac{r_1r_2}{d} \right] \le \exp \left(-\frac{r_1r_2t^2}{4}\right).
	\end{align}
\end{lemma}

With this lemma, we can derive the lower bound of the local $\varepsilon$-packing cardinality of $n$-qubit pure states, which is encapsulated in the following lemma.

\begin{lemma}
	[Lower bound of packing cardinality for $n$-qubit quantum states]
	\label{lem:packing_states_exist}
	Under the distance metric of trace norm $\|\cdot\|_1$, there exists a local $\varepsilon$-packing $\mathcal{P}_{\varepsilon}^{(\gamma\varepsilon)}$ of the set of all $n$-qubit pure quantum states $\ket{\bm{\psi}}$ where the distance between arbitrary two elements in $\mathcal{P}_{\varepsilon}^{(\gamma\varepsilon)}$ is less than $\gamma\varepsilon$ satisfying 
	\begin{equation}    
		\left|\mathcal{P}_{\varepsilon}^{(\gamma\varepsilon)} \right| \ge \exp\left(\frac{d\min\{(1-\varepsilon)^2, (4\gamma^2\varepsilon^2-1)^2\}}{16}\right)  \label{eq:packing_states_exist},
	\end{equation}
	where $\gamma$ refers to an arbitrary constant obeying $0<4\gamma^2\varepsilon^2-1<1$ and $\gamma>2$. 
\end{lemma}

\begin{proof}[Proof of Lemma~\ref{lem:packing_states_exist}] 
	We give a probabilistic existence argument of the lower bound of local $\varepsilon$-packing cardinality. Particularly, we first construct a local $\varepsilon$-packing in which the distance between arbitrary two elements is less than $\gamma\varepsilon$ by applying a probabilistic method. Let $\bm{\rho}_0=\ket{\bm{\phi}}\bra{\bm{\phi}}$ be any fixed quantum state and $\bm{\rho}_i=W_i\bm{\rho}_0 W_i^{\dagger}$ for each $i\in [L]$ with $W_1, W_2, \cdots, W_L \in \mathbb{S}\mathbb{U}(d) $ being arbitrary unitary operators sampled from Haar distribution. 
	In the following, with the aim of showing the existence of such local $\varepsilon$-packing with the desired lower bound, we will show the event that there exists $i \ne j \in [L]$ such that the trace distance between  $\bm{\rho}_i$ and $\bm{\rho}_j$ is less than $2\varepsilon$ or larger than $\gamma\varepsilon$ occurs with a small probability.
	Mathematically, the probability of this event has the form of 
	\begin{align}
		& \mathbb{P}_{W_1,\cdots,W_L\sim \haar} \left[ \bigcup_{i=1}^L \bigcup_{j=i+1}^L \left\{\| \bm{\rho}_{i} - \bm{\rho}_{j} \big\|_1 \le 2\varepsilon\right\} \bigcup \left\{\| \bm{\rho}_{i} - \bm{\rho}_{j} \big\|_1 \ge \gamma\varepsilon \right\} \right] 
		\nonumber \\
		\le & \sum_{i=1}^L \sum_{j=i+1}^L
		\bigg(\mathbb{P}_{W_i,W_j \sim \haar} \left[ \| \bm{\rho}_{i} - \bm{\rho}_{j} \big\|_1 \le 2\varepsilon  \right] + \mathbb{P}_{W_i,W_j \sim \haar} \left[ \| \bm{\rho}_{i} - \bm{\rho}_{j} \big\|_1 \ge \gamma\varepsilon  \right]  \bigg) 
		\nonumber \\
		= & \sum_{i=1}^L \sum_{j=i+1}^L
		\bigg(\mathbb{P}_{W_i,W_j \sim \haar} \left[ \big\| \bm{\rho}_{W} - \bm{\rho}_{\mathbb{I}} \big\|_1 \le 2\varepsilon  \right] + \mathbb{P}_{W_i,W_j \sim \haar} \left[ \big\| \bm{\rho}_{W} - \bm{\rho}_{\mathbb{I}} \big\|_1 \ge \gamma\varepsilon  \right]  \bigg)
		\nonumber \\
		= & \frac{L(L+1)}{2} \bigg(\mathbb{P}_{W \sim \haar} \left[ \big\| \bm{\rho}_{W} - \bm{\rho}_{\mathbb{I}} \big\|_1 \le 2\varepsilon  \right] + \mathbb{P}_{W \sim \haar} \left[ \big\| \bm{\rho}_{W} - \bm{\rho}_{\mathbb{I}} \big\|_1 \ge \gamma\varepsilon  \right]  \bigg)
		\label{eq:union_event_prob}
	\end{align}
	where the first inequality employs the subadditivity of probability measure, the first equality exploits the definition of $\bm{\rho}_W=W\ket{\bm{\phi}}\bra{\bm{\phi}} W^{\dagger}$ and the invariance of trace distance under arbitrary unitary transformations, i.e.,
	\begin{equation}\label{eq:trace_dist_equivalence}
		\left\| \bm{\rho}_{i} -\bm{\rho}_{j} \right\|_1 = \left\| W_i \bm{\rho}_0 W_i^{\dagger} - W_j \bm{\rho}_0 W_j^{\dagger}  \right\|_1 = \left\| W_j^{\dagger}(W_i \bm{\rho}_0 W_i^{\dagger} - W_j \bm{\rho}_0 W_j^{\dagger})W_j \right\|_1 = \| \bm{\rho}_W - \bm{\rho}_{\mathbb{I}} \|_1,
	\end{equation}
	where the last equality follows that the unitary operator $W=W_j^{\dagger}W_i$ follows Haar distribution as $W_i,W_j \in \mathbb{S}\mathbb{U}(d)$ are sampled from Haar distribution. Next, we separately consider the probability of the events in Eqn.~(\ref{eq:union_event_prob}) of $\| \bm{\rho}_{W} - \bm{\rho}_{\mathbb{I}} \big\|_1 \le 2\varepsilon$ and  $\| \bm{\rho}_{W} - \bm{\rho}_{\mathbb{I}} \big\|_1 \ge \gamma\varepsilon$ under the Haar distribution.  Particularly, we first note that with the definition $\bm{\rho}_W=W\ket{\bm{\phi}}\bra{\bm{\phi}} W^{\dagger}$ for any $W \in \mathbb{S}\mathbb{U}(d)$, we have
	\begin{align}
		\big\| \bm{\rho}_W - \bm{\rho}_{\mathbb{I}} \big\|_1 = & \big\| W\ket{\bm{\phi}}\bra{\bm{\phi}} W^{\dagger} -  \ket{\bm{\phi}}\bra{\bm{\phi}} \big\|_1 
		\nonumber \\
		\ge &  \frac{1}{2}\big\| W\ket{\bm{\phi}}\bra{\bm{\phi}} W^{\dagger} -  \ket{\bm{\phi}}\bra{\bm{\phi}} \big\|_1 
		\nonumber \\
		\ge &  \Tr \left( \left(\frac{W\ket{\bm{\phi}}\bra{\bm{\phi}} W^{\dagger} -  \ket{\bm{\phi}}\bra{\bm{\phi}}}{2}\right) \left(\mathbb{I}_d - 2\ket{\bm{\phi}}\bra{\bm{\phi}} \right) \right)
		\nonumber \\
		= & \Tr \left( \left(\mathbb{I}_d - \ket{\bm{\phi}}\bra{\bm{\phi}}  \right) W \ket{\bm{\phi}}\bra{\bm{\phi}} W^{\dagger}   \right)
		\nonumber \\
		= & \Tr \left( \Pi_1 W \Pi_2 W^{\dagger}   \right),
	\end{align}
	where the second inequality employs the property of trace norm that $\|A\|_1 = \max\{|\Tr(AW)|: W\in \mathbb{U}(d) \}$ for any square operator $A$, and $\mathbb{I}_d - 2\ket{\bm{\phi}}\bra{\bm{\phi}} \in \mathbb{U}(d)$, and the last equality is obtained by denoting $\Pi_2=\ket{\bm{\phi}}\bra{\bm{\phi}}$ and $\Pi_1=\mathbb{I}_{d}-\Pi_2$.
	This leads to the conclusion that 
	\begin{equation}\label{eq:single_event_prob}
		\mathbb{P}_{W \sim \haar} \left[ \| \bm{\rho}_{W} - \bm{\rho}_{\mathbb{I}} \big\|_1 \le 2\varepsilon \right] 
		\le  \mathbb{P}_{W\sim \haar}   \left[ \Tr(\Pi_1 W \Pi_2 W^{\dagger}) \le 2\varepsilon \right] 
		\le  \exp \left( -\frac{(d-1)(1-2\varepsilon)^2}{2}\right) 
	\end{equation}
	where the last inequality employs Lemma~\ref{lem:proj_concentration} with taking $t=1-2\varepsilon$. 
	
	On the other hand, employing the relation between trace distance and fidelity of pure states yields
	\begin{align}
		\big\| \bm{\rho}_W - \bm{\rho}_{\mathbb{I}} \big\|_1 = & \big\| W\ket{\bm{\phi}}\bra{\bm{\phi}} W^{\dagger} -  \ket{\bm{\phi}}\bra{\bm{\phi}} \big\|_1 
		\nonumber \\
		= &  \frac{1}{2}\sqrt{ 1- \Tr(W\ket{\bm{\phi}}\bra{\bm{\phi}} W^{\dagger}\ket{\bm{\phi}}\bra{\bm{\phi}}) }
		\nonumber \\
		= &  \frac{1}{2}\sqrt{ \Tr(W\ket{\bm{\phi}}\bra{\bm{\phi}} W^{\dagger}(\mathbb{I}_d-\ket{\bm{\phi}}\bra{\bm{\phi}})) }.
	\end{align}
	This leads to the probability upper bound of the event 	$\big\| \bm{\rho}_W - \bm{\rho}_{\mathbb{I}} \big\|_1 \ge \gamma \varepsilon$ with Lemma~\ref{lem:proj_concentration}, i.e.,
	\begin{align}
		\mathbb{P}_{W \sim \haar} \left[ \| \bm{\rho}_{W} - \bm{\rho}_{\mathbb{I}} \big\|_1 \ge  \gamma \varepsilon \right] 
		= & \mathbb{P}_{W\sim \haar}   \left[\Tr(W\ket{\bm{\phi}}\bra{\bm{\phi}} W^{\dagger}(\mathbb{I}_d-\ket{\bm{\phi}}\bra{\bm{\phi}})) \ge 4\gamma^2\varepsilon^2 \right] 
		\nonumber \\
		\le & \exp \left( -\frac{(d-1)(4\gamma^2\varepsilon^2-1)^2}{2}\right), 
		\label{eq:single_event_prob_lower}
	\end{align}
	where the first inequality employs the condition $0<4\gamma^2\varepsilon^2-1<1$, and the second equation in Lemma~\ref{lem:proj_concentration} with taking $t=4\gamma^2\varepsilon^2-1$.
	In conjunction with Eqn.~(\ref{eq:union_event_prob}), Eqn.~(\ref{eq:single_event_prob}), and Eqn.~(\ref{eq:single_event_prob_lower}), we have
	\begin{align}
		& \mathbb{P}_{W_1,\cdots,W_L\sim \haar} \left[ \bigcup_{i=1}^L \bigcup_{j=i+1}^L \left\{\| \bm{\rho}_{i} - \bm{\rho}_{j} \big\|_1 \le 2\varepsilon\right\} \bigcup \left\{\| \bm{\rho}_{i} - \bm{\rho}_{j} \big\|_1 \ge \gamma\varepsilon \right\} \right]  
		\nonumber \\
		\le & \frac{L(L+1)}{2} \left(\exp \left( -\frac{(d-1)(1-2\varepsilon)^2}{2}\right) + \exp \left( -\frac{(d-1)(4\gamma^2\varepsilon^2-1)^2}{2}\right) \right) 
		\nonumber \\
		\le &  L^2 \exp \left( -\frac{d\min\{(1-2\varepsilon)^2, (4\gamma^2\varepsilon^2-1)^2\}}{4}\right). 
	\end{align}
	This inequality implies that when taking $L=\exp(d(1-\varepsilon)^2/16)$, the probability of the event that $\{\bm{\rho}_i\}_{i=1}^L$ is not an $\varepsilon$-packing or there exists $\bm{\rho}_i, \bm{\rho}_j$ obeying $\|\bm{\rho}_i-\bm{\rho}_j\|_1 \le \gamma\varepsilon$  is strictly less than one. Therefore, we can conclude that there exists a local $\varepsilon$-packing whose cardinality is larger than $L=\exp(d\min\{(1-2\varepsilon)^2, (4\gamma^2\varepsilon^2-1)^2\}/16)$. This completes the proof.
	
\end{proof}

We are now ready to present the proof of Lemma~\ref{lem:packing_observable_lem}.
\begin{proof}[Proof of Lemma~\ref{lem:packing_observable_lem}]
	To measure the local $2\varepsilon$-packing cardinality of $\mathcal{F}$, we first consider the local packing cardinality of the operator group $\mathcal{U}_{O}=\{\bm{U}^{\dagger}\bm{O}\bm{U}| U\in \mathbb{S}\mathbb{U}(d)\}$ and then employ the relation between $\varrho(\Tr(\bm{U}_1^{\dagger}\bm{O}\bm{U}_1\bm{\rho}), \Tr(\bm{U}_2^{\dagger}\bm{O}\bm{U}_2\bm{\rho}))$  and $\varrho_{A}(\bm{U}_1^{\dagger}\bm{O}\bm{U}_1, \bm{U}_2^{\dagger}\bm{O}\bm{U}_2)=\|\bm{U}_1^{\dagger}\bm{O}\bm{U}_1-\bm{U}_2^{\dagger}\bm{O}\bm{U}_2\|_1$  to obtain the local $2\varepsilon$-packing cardinality of $\mathcal{F}$ in the $\varrho$-metric. 
	Specifically, denote $\tilde{\varepsilon}=2\sqrt{2d(d+1)}\varepsilon$, we first construct a local $\tilde{\varepsilon}$-packing $\mathcal{P}_{\tilde{\varepsilon}}^{(\gamma\tilde{\varepsilon})}(\mathcal{U}_{O}, \varrho_{A})$ of the operator group $\mathcal{U}_O$ following the manner in Lemma~\ref{lem:packing_states_exist} as $\bm{U}^{\dagger}\bm{O}\bm{U}$ is the density matrix representation of a quantum state for the projective measurement $\bm{O}=\ket{\bm{o}}\bra{\bm{o}}$. In this regard, it holds according to Lemma~\ref{lem:packing_states_exist} that for any $\bm{U}_1^{\dagger}\bm{O}\bm{U}_1, \bm{U}_2^{\dagger}\bm{O}\bm{U}_2\in \mathcal{P}_{\tilde{\varepsilon}}^{(\gamma\tilde{\varepsilon})}(\mathcal{U}_{O}, \varrho_{A})$,
	\begin{equation}\label{eq:varrho_A_bound}
		2\tilde{\varepsilon}<\varrho_{A}(\bm{U}_1^{\dagger}\bm{O}\bm{U}_1, \bm{U}_2^{\dagger}\bm{O}\bm{U}_2)<\gamma\tilde{\varepsilon}.
	\end{equation}
	On the other hand, the $\varrho$-metric between the elements $\mathrm{f}_{\bm{U}_{1}}$ and $\mathrm{f}_{\bm{U}_{2}}$ in $\mathcal{F}$ yields
	\begin{align}
		& \varrho(\mathrm{f}_{\bm{U}_1}, \mathrm{f}_{\bm{U}_2})
		\nonumber \\
		= & \sqrt{\mathbb{E}_{\bm{\rho} \sim Haar}\Tr\left(\bm{O}\left(\bm{U}_1^{\dagger}\bm{\rho} \bm{U}_1-\bm{U}_2^{\dagger}\bm{\rho} \bm{U}_2\right)\right) ^2 } 
		\nonumber \\
		= & \frac{1}{\sqrt{2d(d+1)}}\left\|\bm{U}_1^{\dagger}\bm{O}\bm{U}_1 - \bm{U}_2^{\dagger}\bm{O}\bm{U}_2 \right\|_1
		\nonumber \\
		= & \frac{1}{\sqrt{2d(d+1)}} \varrho_{A}(\bm{U}_1^{\dagger}\bm{O}\bm{U}_1, \bm{U}_2^{\dagger}\bm{O}\bm{U}_2)
		\label{eq:relation_varrho_A}
	\end{align}
	where the second equality follows Eqn.~(\ref{eq:app_risk_simplify_2}).
	In conjunction with Eqn.~(\ref{eq:varrho_A_bound}) and Eqn.~(\ref{eq:relation_varrho_A}), we have
	\begin{equation}
		4\varepsilon<\varrho(\mathrm{f}_{\bm{U}_1}, \mathrm{f}_{\bm{U}_2})= \frac{1}{\sqrt{2d(d+1)}}\varrho_{A}(\bm{U}_1^{\dagger}\bm{O}\bm{U}_1, \bm{U}_2^{\dagger}\bm{O}\bm{U}_2) < 2\gamma\varepsilon,
	\end{equation}
	where the inequality employs the notation $\tilde{\varepsilon}=2\sqrt{2d(d+1)}\varepsilon$. This implies that the set $\mathcal{X}_{2\varepsilon}^{(2\gamma\varepsilon)}=\{\mathrm{f}_{\bm{U}_1}: \bm{U}_1^{\dagger}\bm{O}\bm{U}_1\in\mathcal{P}_{\tilde{\varepsilon}}^{(\gamma\tilde{\varepsilon})}(\mathcal{U}_{O}, \varrho_{A})\}$ refers to the local $2\varepsilon$-packing of $\mathcal{F}$ where the $\varrho$-metric distance between arbitrary elements in $\mathcal{X}_{2\varepsilon}^{(2\gamma\varepsilon)}$ is larger than $4\varepsilon$ and less than $2\gamma\varepsilon$. Hence, the cardinality of the local $2\varepsilon$-packing  $\mathcal{X}_{2\varepsilon}^{(2\gamma\varepsilon)}$ is the same as that of $\mathcal{P}_{\tilde{\varepsilon}}^{(\gamma\tilde{\varepsilon})}(\mathcal{U}_{O}, \varrho_{A})$, i.e.,
	\begin{equation}
		\left|\mathcal{X}_{2\varepsilon}^{(2\gamma\varepsilon)} \right|=\left|\mathcal{P}_{\tilde{\varepsilon}}^{(\gamma\tilde{\varepsilon})}(\mathcal{U}_{O}, \varrho_{A}) \right| \ge \exp\left(\frac{d\min\{(1-2\tilde{\varepsilon})^2, (4\gamma^2\tilde{\varepsilon}^2-1)^2\}}{16}\right),
	\end{equation}
	where the inequality follows Lemma~\ref{lem:packing_states_exist} with denoting $\tilde{\varepsilon}=2\sqrt{2d(d+1)}\varepsilon$. This completes the proof.

\end{proof}

\subsection{Proof of Lemma~\ref{lem:sum_MI_upper_bound}---bounding the mutual information \texorpdfstring{$\mathrm{I}(X;\hat{X})$}{Lg}}\label{subsec:proof_sum_MI_upper_bound}

The upper bound of the mutual information $\mathrm{I}(X;\hat{X})$ is derived from two aspects. First, the mutual information between $X$ and $\hat{X}$ is upper bounded by the mutual information of the output states $\{(\bm{U}_X \otimes \mathbb{I}_{\mathcal{R}})\ket{\bm{\psi}_j}\}_{j=1}^N$ and measurement outputs $\{\bm{o}_j\}_{j=1}^N$, according to the data processing inequality. This is because increasing the number of measurements allows for more information to be extracted from the output states, thereby increasing the mutual information $\mathrm{I}(X;\hat{X})$. Second, as the mutual information $\mathrm{I}(X;\hat{X})$ cannot grow infinitely with the number of measurements, it is also upper bounded by the mutual information between the index $X$ and the training output states $\{(\bm{U}_X \otimes \mathbb{I}_{\mathcal{R}})\ket{\bm{\psi}_j}\}_{j=1}^N$, which can be interpreted as the maximum amount of information that the output states can obtain about the target unitary. These observations are formulated as the following lemmas, namely, Lemma~\ref{lem:MI_Gaussian} and Lemma~\ref{lem:mi_u_state_2}, which are used to derive the results of Lemma~\ref{lem:sum_MI_upper_bound}. 

\begin{lemma}\label{lem:MI_Gaussian}
	For a given observable $\bm{O}$ and the output state $\bm{\rho}_{X}=(\bm{U}_X \otimes \mathbb{I}_{\mathcal{R}})\bm{\rho} (\bm{U}_X^{\dagger} \otimes \mathbb{I}_{\mathcal{R}})$ with $X$ uniformly distributed in the index set of the local $2\varepsilon$-packing $\mathcal{X}_{2\varepsilon}^{(2\gamma\varepsilon)}$, assume that the $m$-measurement output $\bm{o}=\sum_{k=1}^m\bm{o}_k/m$ conditional on $X=x$ follows the normal distribution $\mathbb{N}(u_x, (\sigma_x)^2)$ where $u_x$ and $\sigma^2$ refers to the mean and variance following Assumption~\ref{assu:normal_distribution}. Then the mutual information between $X$ and $\bm{o}$ given the training state $\bm{\rho}$ yields
	\begin{equation}
		\mathrm{I}(X; \bm{o}) \le \frac{1}{|\mathcal{X}_{2\varepsilon}^{(2\gamma\varepsilon)}|^2}\sum_{x, x' \in \mathcal{X}_{2\varepsilon}^{(2\gamma\varepsilon)}} \frac{(u_x-u_{x'})^2}{2\sigma^2}.  
	\end{equation}
\end{lemma}

\begin{proof}
	[Proof of Lemma~\ref{lem:MI_Gaussian}]
	Denote $\Theta$ as the value space of $\bm{o}$. Let $p_{X}$ and $\mathbb{P}_{\bm{o}|\bm{\rho}_{x}}=\mathbb{P}(\cdot | \bm{\rho}_x):=\mathrm{p}(\cdot | X=x, \bm{\rho})$ be the distribution of $X$ and the distribution of $\bm{o}$ conditional on $X=x$ for fixed input state $\bm{\rho}$, respectively. For simplicity, in the following, we denote $\mathrm{I}(\bm{\rho}_X;\bm{o})=\mathrm{I}(X;\bm{o}|\bm{\rho})$ and $\mathrm{p}(\bm{\rho}_X, \bm{o})=\mathrm{p}(X, \bm{o}|\bm{\rho})$ as the mutual information and joint distribution between $X$ and $\bm{o}$ given $\bm{\rho}$. Then $\bm{o}$ is drawn from the mixture distribution,
	\begin{equation}\label{eq:mixture}
		\bar{\mathbb{P}}_{\bm{o}}: = \sum_{x\in \mathcal{X}_{2\varepsilon}^{(2\gamma\varepsilon)}} \mathrm{p}_{X}(x)\mathbb{P}_{\bm{o}|\bm{\rho}_x} = \frac{1}{|\mathcal{X}_{2\varepsilon}^{(2\gamma\varepsilon)}|}\sum_{x\in \mathcal{X}_{2\varepsilon}^{(2\gamma\varepsilon)}} \mathbb{P}_{\bm{o}|\bm{\rho}_x},
	\end{equation}
	where the second equality follows that $X$ is uniform in $\mathcal{X}_{2\varepsilon}^{(2\gamma\varepsilon)}$. Denoting $\mathrm{p}(\bm{\rho}_x, \bm{o})$ and $\mathrm{p}(\bm{o}|\bm{\rho}_x)$ as the joint probability distribution function of $\bm{o}$ and $X=x$, and the conditional probability distribution function of $\bm{o}$ on $X=x$, respectively. 
	With the definition of the mixture distribution and mutual information, we have 
	\begin{align}
		\mathrm{I}(\bm{\rho}_X; \bm{o}) & = \mathrm{D}_{\KL}(\mathbb{P}_{\bm{\rho}_X, \bm{o}}\|\mathbb{P}_{\bm{\rho}_X}\mathbb{P}_{\bm{o}} )
		\nonumber \\
		& = \sum_{x \in \mathcal{X}_{2\varepsilon}^{(2\gamma\varepsilon)}} \int_{\Theta} \mathrm{p}(\bm{\rho}_x, \bm{o}) \log\frac{\mathrm{p}(\bm{\rho}_x, \bm{o})}{\mathrm{p}_{X}(x) \bar{\mathrm{p}}(\bm{o})} \mathrm{d}\bm{o}
		\nonumber \\
		& = \sum_{x \in \mathcal{X}_{2\varepsilon}^{(2\gamma\varepsilon)}} \mathrm{p}_{X}(x) \int_{\Theta}  \mathrm{p}(\bm{o}|\bm{\rho}_x) \log \frac{\mathrm{p}(\bm{o}|\bm{\rho}_x)}{\bar{\mathrm{p}}(\bm{o})} \mathrm{d} \bm{o}
		\nonumber \\
		& = \frac{1}{|\mathcal{X}_{2\varepsilon}^{(2\gamma\varepsilon)}|}\sum_{x \in \mathcal{X}_{2\varepsilon}^{(2\gamma\varepsilon)}} \mathrm{D}_{\KL}(\mathbb{P}_{o|\bm{\rho}_{x}} \| \bar{\mathbb{P}}_{\bm{o}})   
		\nonumber \\
		& \le \frac{1}{|\mathcal{X}_{2\varepsilon}^{(2\gamma\varepsilon)}|^2}\sum_{x, x' \in \mathcal{X}_{2\varepsilon}^{(2\gamma\varepsilon)}}  \mathrm{D}_{\KL}(\mathbb{P}_{o|\bm{\rho}_{x}} \| \mathbb{P}_{o|\bm{\rho}_{x'}}) 
		\nonumber \\
		& = \frac{1}{|\mathcal{X}_{2\varepsilon}^{(2\gamma\varepsilon)}|^2}\sum_{x, x' \in \mathcal{X}_{2\varepsilon}^{(2\gamma\varepsilon)}}  \mathrm{D}_{\KL}(\mathbb{N}(u_x, \sigma^2) \| \mathbb{N}(u_{x'}, \sigma^2))
		\nonumber \\
		& = \frac{1}{|\mathcal{X}_{2\varepsilon}^{(2\gamma\varepsilon)}|^2}\sum_{x, x' \in \mathcal{X}_{2\varepsilon}^{(2\gamma\varepsilon)}}   \frac{(u_x-u_{x'})^2}{2\sigma^2},
	\end{align}
	where the first equality and the second equality employ the definition of mutual information that $\mathrm{I}(X;Y)=\mathrm{D}_{\KL}(\mathrm{p}_{XY}||\mathrm{p}_X\mathrm{p}_Y)$ with $\mathrm{D}_{\KL}$ referring to the KL-divergence and $\mathrm{p}_{XY},\mathrm{p}_X,\mathrm{p}_Y$ referring to the joint probability distribution of $X$ and $Y$, the probability distribution of $X$ and $Y$ respectively, the third equality exploits the property of conditional probability that $\mathrm{p}(\bm{\rho}_x, \bm{o})=\mathrm{p}( \bm{o}|\bm{\rho}_x)\mathrm{p}_{X}(x)$,   the fourth equality follows that $X$ is uniform in $ \mathcal{X}_{2\varepsilon}^{(2\gamma\varepsilon)}$ and hence $\mathrm{p}_{X}(x)=1/|\mathcal{X}_{2\varepsilon}^{(2\gamma\varepsilon)}|$ for all $x \in \mathcal{X}_{2\varepsilon}^{(2\gamma\varepsilon)}$, the first inequality employs the Jensen's inequality with the convexity of the function $-\log$, the fifth equality follows that $o_{j}=\sum_{k=1}^m o_{k}/m$ following the binomial distribution $\mathbb{B}(m,u_x)$ is approximated by the normal distribution $\mathbb{N}(u_x, \sigma^2)$ with the variance is assumed to be identical over different choice of $x$, the final equality is obtained by algebraic computation of the KL-divergence between two Gaussian distributions with the identical variance.
\end{proof}

We next derive an upper bound of the mutual information $\mathrm{I}(X;\hat{X})$ independent with the number of measurements $m$. Particularly, since a large number of measurements aims to precisely estimate the expectation value of the projective measurement $\bm{O}$ on the training output states $(\bm{U}_X\otimes \mathbb{I}_d)\ket{\bm{\psi}_j}$, i.e., $u_j^{(X)}=\Tr((\bm{U}_X^{\dagger}\bm{O}\bm{U}_X \otimes \mathbb{I}_d)\ket{\bm{\psi}_j}\bra{\bm{\psi}_j})$, we can directly focus on the mutual information of $X$ and the training output states $\{(\bm{U}_X\otimes \mathbb{I}_d)\ket{\bm{\psi}_j}\}_{j=1}^N$ by regarding the states $(\bm{U}_X\otimes \mathbb{I}_d)\ket{\bm{\psi}_j}$ as random variables. We encapsulate the analyzed results related to this case into the following lemma.

\begin{lemma}
	
	\label{lem:mi_u_state_2}
	For any Haar random unitary $\bm{U}_X$ uniformly sampled from the local $2\varepsilon$-packing $\{\bm{U}_x\}_{x\in \mathcal{X}_{2\varepsilon}^{(2\gamma\varepsilon)}}$ and the entangled
	states set $\{\ket{\bm{\psi}_j}=\sum_{k=1}^r \sqrt{c_{jk}}\ket{\bm{\xi}_{jk}} \ket{\bm{\zeta}_{jk}}\}_{j=1}^N$, let $\bm{o}_j=\sum_{k=1}^{m}\bm{o}_{jk}/m$ be the $m$ measured results  of the corresponding output states $\ket{\bm{\psi}_j^{(X)}}=(\bm{U}_X\otimes \mathbb{I}_d)\ket{\bm{\psi}_j}$ under the projective measurement $\bm{O}$.  Then the mutual information between $X$ and the measurement outputs $\bm{o}_1, \cdots, \bm{o}_N$ for the infinite number of measurements $m\rightarrow \infty$ yields 
	\begin{equation}
		\lim\limits_{m \to \infty}\mathrm{I}(X;\bm{o}_1, \cdots, \bm{o}_N) \le 
		\mathrm{I}(X;\ket{\bm{\psi}_1^{(X)}}, \cdots, \ket{\bm{\psi}_N^{(X)}}) \le Nr\log\left(d\right).
	\end{equation} 
\end{lemma}
\begin{proof}[Proof of Lemma \ref{lem:mi_u_state_2}]
	We first note that when the number of measurements $m$ is infinite, the measured result $\bm{o}_j$ is exactly the expectation value of the output states $\ket{\bm{\psi}_j^{(X)}}$ under the projective measurement $\bm{O}$, i.e., $\lim_{m \to \infty}\bm{o}_j = u_j^{(X)}:=\Tr((\bm{O}\otimes \mathbb{I}_d)\ket{\bm{\psi}_j^{(X)}}\bra{\bm{\psi}_j^{(X)}})$. In this regard, we have
	\begin{align}
		\lim\limits_{m \to \infty}\mathrm{I}(X;\bm{o}_1, \cdots, \bm{o}_N) = & \mathrm{I}\left(X;\lim\limits_{m \to \infty}\bm{o}_1,  \cdots, \lim\limits_{m \to \infty}\bm{o}_N \right)
		\nonumber \\
		= & \mathrm{I}\left(X;u_1^{(X)},  \cdots,u_N^{(X)} \right)
		\nonumber \\
		\le & \mathrm{I}\left(X;\ket{\bm{\psi}_1^{(X)}},  \cdots, \ket{\bm{\psi}_N^{(X)}} \right)
		\nonumber \\
		\le & \mathrm{I}(X;\ket{\bm{\xi}_{11}^{(X)}},  \cdots, \ket{\bm{\xi}_{1r}^{(X)}}, \cdots, \ket{\bm{\xi}_{N1}^{(X)}},  \cdots, \ket{\bm{\xi}_{Nr}^{(X)}} )
		\nonumber \\
		= & \mathrm{I}\left(X;\otimes_{j=1}^N \otimes_{k=1}^r \ket{\bm{\xi}_{jk}^{(X)}}\right)
		\nonumber \\
		\le & S\left(\frac{1}{\mathcal{X}_{2\varepsilon}^{(2\gamma\varepsilon)}}\sum_{x\in \mathcal{X}_{2\varepsilon}^{(2\gamma\varepsilon)}}\otimes_{j=1}^N \otimes_{k=1}^r \ket{\bm{\xi}_{jk}^{(x)}}\bra{\bm{\xi}_{jk}^{(x)}}\right)
		\nonumber \\
		\le & Nrn\log(2) 
	\end{align}
	where the first equality employs the dominated convergence theorem applied to the integrated function  in the mutual information, which converges point-wise  and is bounded by $\log(|\mathcal{X}_{2\varepsilon}^{(2\gamma\varepsilon)}|)$, such that the interchange of the order of limitation and integration in the calculation is allowed, the first inequality employs the data processing inequality given in the Property~\ref{property:data_pro_ineqn} for the Markov chain $X\to \ket{\bm{\psi}_1^{(X)}},  \cdots, \ket{\bm{\psi}_N^{(X)}} \to u_1^{(X)},  \cdots,u_N^{(X)}$, the second inequality follows the data processing inequality with the observation that the randomness of the states $\ket{\bm{\psi}_j^{(X)}}$ about $X$ comes from the states $\{\ket{\bm{\xi}_{jk}^{(X)}}\}_{k=1}^r$ and hence there exists an $X$-independent functional $\mathrm{g}:\{\ket{\bm{\xi}_{jk}^{(X)}}\}_{j,k=1}^{N,r} \to \{\ket{\bm{\psi}_j^{(X)}}\}_{1}^N$ leading to a Markov chain $X\to \{\ket{\bm{\xi}_{jk}^{(X)}}\}_{j,k=1}^{N,r} \to \{\ket{\bm{\psi}_j^{(X)}}\}_{1}^N$,   the second inequality employs the Holevo's theorem \cite{araki1970entropy,bengtsson2017geometry,holevo1973bounds,horodecki2009quantum}, and the last inequality uses the fact that the entropy of mixed states living in an $(N\cdot r\cdot n)$-qubit system is at most $Nrn\log 2$. This completes the proof.
\end{proof}

We note that the mutual information $\mathrm{I}(X;\ket{\bm{\psi}_j^{(X)}})$ is defined to quantify the accessible information of the target index from the limited number of training states. Lemma~\ref{lem:mi_u_state_2} implies that the mutual information between the target index $X$ and the measurement outputs of projective measurement $\bm{O}$ on the $m$ copies of the state $\ket{\bm{\psi}_j^{(X)}}$ is upper bounded by the mutual information between $X$ and \textit{the single copy} of $\ket{\bm{\psi}_j^{(X)}}$. This is different from previous literature studying quantum state (or channel) tomography where the mutual information between the target index $X$ and the measurement outputs of adaptive or non-adaptive POVM on the $m$ copies of the state $\ket{\bm{\psi}_j^{(X)}}$ is upper bounded by that between $X$ and \textit{the $m$ copies} of $\ket{\bm{\psi}_j^{(X)}}$. Particularly, under the projective measurement $\bm{O}$, the single measurement output is either $0$ or $1$ such that the infinite measurement outputs can be characterized by a sufficient statistic $u_j^{(X)}=\Tr((\bm{O}\otimes \mathbb{I}_d)\ket{\bm{\psi}_j^{(X)}}\bra{\bm{\psi}_j^{(X)}})$, which contains less information of $X$ than that in a single copy of state $\ket{\bm{\psi}_j^{(X)}}$.  Additionally, this bound is tight in the sense that when $Nr=d$, the mutual information $\mathrm{I}(X;\ket{\bm{\psi}_1}, \cdots, \ket{\bm{\psi}_N})$ reaches the maximal order of $\mathcal{O}(d)$, which is exactly equal to the order of another upper bound of the mutual information $\mathrm{I}(X;\ket{\bm{\psi}_1^{(X)}}, \cdots, \ket{\bm{\psi}_N^{(X)}})\le \mathrm{H}(X) =\mathcal{O}(d)$ according to Lemma~\ref{lem:packing_observable_lem}. 

\smallskip	
We now use the above results to prove  Lemma~\ref{lem:sum_MI_upper_bound}.
\begin{proof}[Proof of Lemma~\ref{lem:sum_MI_upper_bound}]
	For any training input state $\bm{\rho}_j$ independently sampled from the Haar distribution, we denote $\bm{\rho}_{j}^{(X)}=(\bm{U}_X\otimes \mathbb{I}_{\mathcal{R}})\bm{\rho}_j (\bm{U}_X^{\dagger}\otimes \mathbb{I}_{\mathcal{R}})$ and $\bm{\rho}_{i:j}^{(X)}=\{\bm{\rho}_{i}^{(X)}, \cdots, \bm{\rho}_{j}^{(X)}\}$ for any $j >i$. The Markov chain describing the learning process can be depicted as 
	\begin{equation}
		X \to \bm{U}_X \to \bm{\rho}_{1:N}^{(X)} \to \{\bm{o_j}\}_{j=1}^N \to \hat{X}
	\end{equation}
	Hence, with the data processing inequality, we have
	\begin{align}\label{eq:sum_MI}
		\mathrm{I}(X;\hat{X}) & \le \mathrm{I}(X; \bm{o}_{1:N} )
		\nonumber \\
		& \le \sum_{j=1}^N  \mathrm{I}(X; \bm{o}_{j}) 
		\nonumber \\
		& = \frac{1}{|\mathcal{X}_{2\varepsilon}^{(2\gamma\varepsilon)}|^2} \sum_{j=1}^N \sum_{x, x' \in \mathcal{X}_{2\varepsilon}^{(2\gamma\varepsilon)}}   \frac{(u_x-u_{x'})^2}{2\sigma^2},
	\end{align}
	where $\bm{o}_{1:j}=(\bm{o}_1, \cdots, \bm{o}_j)$ and $\bm{o}_j=\sum_{k=1}^m\bm{o}_{jk}/m$ with $\bm{o}_{jk}$ being the $k$-th 
	measurement output on the $j$-th input state, the first inequality employs the data processing inequality,  and the first equality employs Lemma~\ref{lem:MI_Gaussian} with $u_x^{(j)}=\mathbb{E}[\bm{o}_j]=\Tr((\bm{U}_x^{\dagger}O \bm{U}_x\otimes \mathbb{I}_{\mathcal{R}})\bm{\rho}_j)$ and $\sigma^2=\mathbb{E}_{\bm{\rho}_j \sim \haar}u_x^{(j)}(1-u_x^{(j)})/m$ defined in Assumption~\ref{assu:normal_distribution}. 
	
	The problem of upper bounding the mutual information between $\mathrm{I}(X;\hat{X})$ is now reduced to lower bound the variance $\sigma^2$ and to upper bound the square of the difference between $u_x$ and $u_{x'}$ for varying index $x,x'$, i.e., $(u_x-u_{x'})^2$. Particularly,   the direct calculation of Haar integration $\sigma^2=\mathbb{E}_{\bm{\rho}_j \sim \haar}u_x^{(j)}(1-u_x^{(j)})/m$ over the Haar random state $\bm{\rho}_j = \ket{\bm{\psi}_j}\bra{\bm{\psi}_j}$ yields 
	\begin{align}
		\sigma^2 & = \frac{1}{m} \cdot \mathbb{E}_{\ket{\bm{\psi}_j} \sim  \haar } 	 \left[\Tr_{\mathcal{X} \mathcal{R}}((\bm{U}^{\dagger}\bm{O}\bm{U})\otimes 	\mathbb{I}_{\mathcal{R}} \ket{\bm{\psi}_j}\bra{\bm{\psi}_j}) - \Tr_{\mathcal{X} \mathcal{R}}((\bm{U}^{\dagger}\bm{O}\bm{U})\otimes \mathbb{I}_{\mathcal{R}} \ket{\bm{\psi}_j}\bra{\bm{\psi}_j})^2 \right]
		\nonumber \\
		& = \frac{1}{m} \cdot \mathbb{E}_{\ket{\bm{\psi}_j} \sim \haar}   \left[	\Tr_{\mathcal{X}}(\bm{U}^{\dagger}\bm{O}\bm{U}\Tr_{\mathcal{R}}( 	\ket{\bm{\psi}_j}\bra{\bm{\psi}_j})) - \Tr_{\mathcal{X}}(\bm{U}^{\dagger}\bm{O}\bm{U}\Tr_{\mathcal{R}} (\ket{\bm{\psi}_j}\bra{\bm{\psi}_j}))^2\right]
		\nonumber \\
		& = \frac{1}{m} \cdot \mathbb{E}_{\ket{\bm{\xi}_j} \sim \haar} \mathbb{E}_{\ket{c_j} \sim \haar}	\left[\sum_{k=1}^r 	c_{jk}\Tr(\bm{U}^{\dagger}\bm{O}\bm{U}\ket{\bm{\xi}_{jk}}\bra{\bm{\xi}_{jk}}) - \left(\sum_{k=1}^r c_{jk}\Tr(\bm{U}^{\dagger}\bm{O}\bm{U}\ket{\bm{\xi}_{jk}}\bra{\bm{\xi}_{jk}})\right)^2\right]
		\nonumber \\		
		& = \frac{1}{m} \left( \frac{1}{d}-\frac{1}{d(d+1)} -\frac{1}{d(d+1)r} \right)
		\nonumber \\
		& = \frac{dr-1}{mrd(d+1)} 
		\nonumber \\
		& \ge \frac{1}{md}, \label{eq:variance_lower_bound}
	\end{align}
	where the first equality employs $\sigma_x^2=u_x(1-u_x)/m$, the third equality uses the definition of the entangled state in Eqn.~(\ref{eq:entangled_state}) and the coefficient of the entangled Haar random state $\ket{\bm{c}_j}=(\sqrt{c_{j1}}, \cdots, \sqrt{c_{jr}})$ forms a Haar random state in $\mathbb{S}\mathbb{U}(r)$, and the fourth equality follows Haar integration calculation by exploiting Lemma~\ref{lem:Tr(psipsi)} and Lemma~\ref{lem:Tr(psipsi)Tr(psipsi)}. 		On the other hand, as we consider the average case of the mutual information $\mathrm{I}(X;\hat{X})$ over random training data $\mathcal{S}$ following  Construction Rule \ref{construct:1}, the corresponding Haar integration of $(u_x-u_{x'})^2$ in Eqn.~(\ref{eq:sum_MI}) yields 
	\begin{align}
		& \mathbb{E}_{\mathcal{S}} (u_x-u_{x'})^2
		\nonumber \\
		= & 
		\mathbb{E}_{\mathcal{S}}\Tr( (\bm{U}_x^{\dagger}\bm{O}\bm{U}_x-\bm{U}_{x'}^{\dagger}\bm{O}\bm{U}_{x'}) \otimes \mathbb{I} \ket{\bm{\psi}_i}\bra{\bm{\psi}_i})^2 
		\nonumber \\
		= & 
		\mathbb{E}_{\mathcal{S}}\Tr(\sum_{k=1}^r c_{i,k} (\bm{U}_x^{\dagger}\bm{O}\bm{U}_x-\bm{U}_{x'}^{\dagger}\bm{O}\bm{U}_{x'}) \ket{\bm{\xi}_{i,k}}\bra{\bm{\xi}_{i,k}} )^2
		\nonumber \\
		= & 
		\mathbb{E}_{\ket{\bm{c}_{i}}\sim \haar} \Bigg( \sum_{k=1}^r c_{i,k}^2  \mathbb{E}_{\ket{\bm{\xi}_{i,k}}\sim \haar} \Tr( (\bm{U}_x^{\dagger}\bm{O}\bm{U}_x-\bm{U}_{x'}^{\dagger}\bm{O}\bm{U}_{x'})\ket{\bm{\xi}_{i,k}}\bra{\bm{\xi}_{i,k}})^2 
		\nonumber \\
		& +
		\sum_{k=1}^r \sum_{j=k+1}^r  2c_{i,k}c_{i,j}  \mathbb{E}_{\ket{\bm{\xi}_{i,k}}\sim \haar} \mathbb{E}_{\ket{\bm{\xi}_{i,j}}\sim \haar} \Tr( (\bm{U}_x^{\dagger}\bm{O}\bm{U}_x-\bm{U}_{x'}^{\dagger}\bm{O}\bm{U}_{x'})\ket{\bm{\xi}_{i,k}}\bra{\bm{\xi}_{i,k}})\Tr( (\bm{U}_x^{\dagger}\bm{O}\bm{U}_x-\bm{U}_{x'}^{\dagger}\bm{O}\bm{U}_{x'})\ket{\bm{\xi}_{i,j}}\bra{\bm{\xi}_{i,j}})  \Bigg)
		\nonumber \\
		= & 
		\mathbb{E}_{\ket{\bm{c}_{i}}\sim \haar} \left(\sum_{k=1}^r c_{i,k}^2 \varrho(\mathrm{f}_{\bm{U}_x}, \mathrm{f}_{\bm{U}_{x'}})^2  -  \sum_{k=1}^r \sum_{j=k+1}^r  \frac{2c_{i,k}c_{i,j}}{d(d^2-1)} \right) 
		\nonumber \\
		= & \mathbb{E}_{\ket{\bm{c}_i}\sim \haar} \left( \sum_{k=1}^r  \Tr(\ket{\bm{c}_i}\bra{\bm{c}_i} \ket{e_k}\bra{e_k})^2 \varrho(\mathrm{f}_{\bm{U}_x}, \mathrm{f}_{\bm{U}_{x'}})^2 
		-  \sum_{k=1}^r \sum_{j=k+1}^r  \frac{2 \Tr(\ket{\bm{c}_i}\bra{\bm{c}_i} \ket{e_k}\bra{e_k}) \Tr(\ket{\bm{c}_i}\bra{\bm{c}_i} \ket{e_j}\bra{e_j}) }{d(d^2-1)}\right)
		\nonumber \\
		= & \sum_{k=1}^r \frac{ 2\varrho(\mathrm{f}_{\bm{U}_x}, \mathrm{f}_{\bm{U}_{x'}})^2}{r(r+1)} - \sum_{k=1}^r \sum_{j=k+1}^r \frac{2}{r(r+1)d(d^2-1)} 
		\nonumber \\
		\le & \frac{8\gamma^2\varepsilon^2}{r}, \label{eq:single_MI}
	\end{align}
	where the second equality employs    the representation of the input states  $\ket{\bm{\psi}_j}=\sum_{k=1}^r \sqrt{c_{j, k}}\ket{\bm{\xi}_{j, k}}_{\mathcal{X}}\ket{\bm{\zeta}_{j, k}}_{\mathcal{R}}$ defined in Eqn.~(\ref{eq:entangled_state}) and the property of partial trace $\operatorname{Tr}_{\mathcal{R}}\left[\ket{\bm{\psi}_j}\bra{\bm{\psi}_j}\right]=\sum_{k=1}^r c_{j, k}\ket{\bm{\xi}_{j, k}}\bra{\bm{\xi}_{j, k}}_{\mathcal{X}}$, the third equality follows direct expansion of the square function and formulates the Haar integration of the random state $\ket{\bm{\psi}_i} \in \mathcal{S}$ into the Haar integration with respect to $\ket{\bm{c}_i}$ and $\ket{\bm{\xi}_{i,k}}$ which follow Haar distribution in $\mathbb{S}\mathbb{U}(r)$ and $\mathbb{S}\mathbb{U}(d)$ according to  Construction Rule \ref{construct:1}, the fourth equality employs Lemma~\ref{lem:rewrite_risk} and Lemma~\ref{lem:Tr(WA)Tr(WB)_Orth} to derive the first term and the second term respectively, the fifth equality follows that $c_{i,k} = |\braket{e_k|\bm{c}_i}|^2 = \Tr(\bra{e_k}\ket{\bm{c}_i}\bra{\bm{c}_i}\ket{e_k})$ with $\ket{\bm{c}_i}=(\sqrt{c_{i,1}}, \cdots,\sqrt{c_{i,r}} )^{\top}$ and $\ket{e_k}$ being the computational basis defined in Section~\ref{app_subsec:notation}, the sixth equality exploits Lemma~\ref{lem:Tr(psipsi)Tr(psipsi)} to calculate the Haar integration with respect to $\ket{\bm{c}_i}\in \mathbb{S}\mathbb{U}(r)$, and the last equality employs the property of the   local $2\varepsilon$-packing in Lemma~\ref{lem:packing_observable_lem} in which the distance between arbitrary two elements is less than $2\gamma\varepsilon$, i.e., $\varrho(\mathrm{f}_{\bm{U}_x}, \mathrm{f}_{\bm{U}_x'})\le 2\gamma\varepsilon$ for any $x, x' \in \mathcal{X}_{2\varepsilon}^{(2\gamma\varepsilon)}$, and additionally the positivity of the second term.

	In conjunction with Eqn.~(\ref{eq:sum_MI}), Eqn.~(\ref{eq:variance_lower_bound}) and Eqn.~(\ref{eq:single_MI}), we have
	\begin{equation}
		\mathbb{E}_{\mathcal{S}} \mathrm{I}(X;\hat{X}) \le \frac{1}{|\mathcal{X}_{2\varepsilon}^{(2\gamma\varepsilon)}|^2} \sum_{j=1}^N \sum_{x, x' \in \mathcal{X}_{2\varepsilon}^{(2\gamma\varepsilon)}}   \frac{\mathbb{E}_{\mathcal{S}} (u_x-u_{x'})^2}{2\sigma^2} \le \frac{4 Nmd\gamma^2\varepsilon^2}{r}.
	\end{equation}

	Besides, the mutual information between the target index $X$ and the measurement outcomes $\bm{o}_j$ given the state $\bm{\rho}_j^{(X)}$ is bounded by the information of $X$ carried by the states $\bm{\rho}_j^{(X)}$ itself. Particularly, employing the analyzed results in Lemma~\ref{lem:mi_u_state_2} yields
	\begin{align}\label{eq:MI_state_bound}
		\lim\limits_{m \to \infty}\mathrm{I}(X;\bm{o}_1, \cdots, \bm{o}_N) \le 
		\mathrm{I}(X;\ket{\bm{\psi}_1^{(X)}}, \cdots, \ket{\bm{\psi}_N^{(X)}}) \le Nr\log\left(d\right).
	\end{align}

	In conjunction with Eqn.~(\ref{eq:single_MI}), Eqn.~(\ref{eq:sum_MI}) and Eqn.~(\ref{eq:MI_state_bound}), we have
	\begin{align}\label{eq:sum_MI_bound}
		\mathbb{E}_{\mathcal{S}}	\mathrm{I}(X;\hat{X}) 
		\le N\cdot \min \left\{\frac{4 md\gamma^2\varepsilon^2}{r}, r\log\left(d \right) \right\}. 
	\end{align}
	This completes the proof.
\end{proof} 

\subsection{The scaling of training error and prediction error}
In this section, we explain the scaling $1/2^n$ in the training error $\varepsilon$ and the factor of the achieved lower bound $\tilde{\varepsilon}^2/4^n$. we would like to analytically show that $\varepsilon=\mathcal{O}(1/2^n)$ is a ground truth when considering the projective measurement and the framework of quantum NFL.

Particularly, we recall that quantum NFL considers the average performance of the quantum learning models over Haar unitaries. This means that considering the Haar random unitary $\bm{U}$, the expectation value of an arbitrary observable $\bm{O}$ on the evolved state $\bm{U}\ket{\bm{\psi}}$ yields
\begin{equation}\label{eq:training_err_scaling}
	\mathbb{E}_{\bm{U} \sim \haar}\Tr(\bm{O}\bm{U}\bm{\rho} \bm{U}^{\dagger})=\frac{\Tr(\bm{O})}{2^n},
\end{equation}
where the equality follows the direct calculation of Haar integration, as given by Lemma~\ref{lem:Tr(WW)}.
This indicates that the scaling of measurement output depends on the employed observable $\bm{O}$ in the target function defined in Eqn.~(\ref{eq:learning_model}). A large trace $\Tr(\bm{O})$ leads to a large training error and prediction error for given the same training data and target unitaries. Moreover, for the observable set as the projective measurement $\bm{O}=\ket{\bm{o}}\bra{\bm{o}}$, the measurement outputs scale as $\mathcal{O}(1/2^n)$ as the trace of projective measurement being $\Tr(\ket{\bm{o}}\bra{\bm{o}})=1$. In this context, the training error scales naturally as $\mathcal{O}(1/2^n)$, which is easy to reach during training. Moreover, a similar scaling occurs in the prediction error, as a small magnitude of the measurement output will lead to a small scaling of prediction error.

\subsection{The tightness of the derived lower bound in Theorem~\ref{thm:formal_finite_measurement} from the perspective of quantum state tomography}\label{app_subsec:tightness}
We now discuss the tightness of the derived lower bound in Theorem~\ref{thm:formal_finite_measurement}. According to the results of Theorem \ref{thm:formal_finite_measurement}, the query complexity of the target unitary $Nm$ grows with the Schmidt rank $r$ to achieve zero risk. To connect with quantum state tomography, in the following, we focus on the case of the Schmidt rank $r=1$, suggesting that the minimal query complexity for achieving zero risk is $Nm=\Omega(d^2/\tilde{\varepsilon}^2)$. 

Particularly, as the observable $\bm{O}=\ket{\bm{o}}\bra{\bm{o}}$ considered in this work is the projective measurement, the task of learning the unitary under such observable is equivalent to conducting tomography on the pure quantum states $\bm{U}^{\dagger}\ket{\bm{o}}\bra{\bm{o}}\bm{U}$, since both tasks aim to bound the trace distance between the target operator $\bm{U}^{\dagger}\ket{\bm{o}}\bra{\bm{o}}\bm{U}$ and the learned operator $\bm{V}_{\mathcal{S}}^{\dagger}\ket{\bm{o}}\bra{\bm{o}}\bm{V}_{\mathcal{S}}$, i.e., $\|\bm{U}^{\dagger}\ket{\bm{o}}\bra{\bm{o}}\bm{U}-\bm{V}_{\mathcal{S}}^{\dagger}\ket{\bm{o}}\bra{\bm{o}}\bm{V}_{\mathcal{S}}\|_1$. In this regard, measuring the output states $\bm{U}\ket{\bm{\psi}_j}$ with the projective measurement $\bm{O}=\ket{\bm{o}}\bra{\bm{o}}$ could be regarded as measuring the operator $\bm{U}^{\dagger}\ket{\bm{o}}\bra{\bm{o}}\bm{U}$ with the projective measurement  $\ket{\bm{\psi}_j}\bra{\bm{\psi}_j}$. Notably, while the training states can vary over $\mathbb{S}\mathbb{U}(d)$, the measurement outputs always keep $0$ or $1$. That implies that the set of training states $\{\ket{\bm{\psi}_j}\}_{j=1}^N$ acts as the measurement with only two outcomes. There is much literature studying the query complexity for quantum state tomography (note that the term of query complexity discussed in our work amounts to the sample complexity in the context of quantum state tomography). Ref.~\cite{lowe2022lower} has derived a tight lower bound of the query complex under the non-adaptive measurement with a constant number of outcomes, which is given as $\Omega(d^2/\delta^2)$ with $\delta$ obeying $\|\bm{U}^{\dagger}\ket{\bm{o}}\bra{\bm{o}}\bm{U}-\bm{V}_{\mathcal{S}}^{\dagger}\ket{\bm{o}}\bra{\bm{o}}\bm{V}_{\mathcal{S}} \|_1 \le \delta$. Notably, employing the property of $2\varepsilon$-packing leads to $\tilde{\varepsilon} = 4\varepsilon\sqrt{2d(d+1)} \le 2\|\bm{U}^{\dagger}\ket{\bm{o}}\bra{\bm{o}}\bm{U}-\bm{V}_{\mathcal{S}}^{\dagger}\ket{\bm{o}}\bra{\bm{o}}\bm{V}_{\mathcal{S}} \|_1$, which is equivalent to the effect of $\delta$. This implies that the achieved query complexity in Theorem~\ref{thm:formal_finite_measurement} $Nm=\Omega(d^2/\tilde{\varepsilon}^2)$ is tight and optimal for quantum state tomography under the non-adaptive measurement with a constant number of outcomes.

\section{Quantum NFL theorem for generic measurements (Theorem~\ref{thm:formal_finite_measurement_POVM}) }\label{app_sec:lower_bound_POVM}
In this section, we elucidate how to derive the lower bound of prediction error when the learning process uses positive operator-valued measure (POVM) to read out the classical information from output states rather than directly estimating the statistical expectation $\Tr(\bm{O}\bm{U}\bm{\rho} \bm{U}^{\dagger})$ by measuring states $\bm{U}\bm{\rho} \bm{U}^{\dagger}$ with projective measurement $\bm{O}$. We recall that the lower bound is determined by two factors, namely, the cardinality of $2\varepsilon$-packing $|\mathcal{X}_{2\varepsilon}(\bm{O})|$ and the mutual information $\mathrm{I}(X;\hat{X})$ with $X$ and $\hat{X}$ being the target index and estimated index respectively. The transition role of entangled data comes from the bound of mutual information where the Schmit rank occurs in the denominator and numerator of a minimum term as given in Lemma~\ref{lem:sum_MI_upper_bound}. On the other hand, learning the same target function with different measurement approaches does not change the $2\varepsilon$-packing cardinality $|\mathcal{X}_{2\varepsilon}(\bm{O})|$, which only depends on the observable $\bm{O}$ as well as $\varepsilon$ and is independent of the interested quantities, i.e., Schmidt rank $r$, training data size $N$, and the number of measurement $m$. In this regard, we only need to focus on the upper bound of mutual information $\mathrm{I}(X;\hat{X})$ for employing POVM to read out classical information.

This section is organized as follows. We will first introduce the general POVM with $\ell$-outcome. Subsequently, we elaborate on the upper bound of mutual information $\mathrm{I}(X;\hat{X})$ for such POVM with $\ell$-outcome. Combining the established upper bound of $\mathrm{I}(X;\hat{X})$ with Lemma~\ref{lem:Fano} and Lemma~\ref{lem:packing_observable_lem}, the lower bound of the averaged prediction error in Eqn.~(\ref{eq:minimax_risk}) can be directly established.

\subsection{Single-copy POVM with \texorpdfstring{$\ell$}{Lg}-outcome}\label{app_sec:single_copy_POVM}
We formulate the definition of single-copy positive operator-valued measure (POVM) according to conventions in Ref.~\cite{lowe2022lower, chen2022exponential}. A POVM for $d$-dimensional quantum states refers to a set of linear map $\mathcal{M}:\mathcal{H}_d\to L(\ell)$ that act on quantum states $\bm{\rho} \in \mathcal{H}_d$, where $\ell$ is the number of possible outcomes. Specifically, a POVM with $\ell$ outcomes maps the state $\bm{\rho}$ as follows:
\begin{equation}
	\mathcal{M}: \bm{\rho} \mapsto \sum_{z\in \mathcal{Z}}\Tr(M_z \bm{\rho} ) \ket{\bm{z}}\bra{\bm{z}},
\end{equation}
where $\mathcal{M}_z\in \{\Psd(d)\}$ are positive semi-definite operators that satisfy $\sum_{z\in \mathcal{Z}} M_z = \mathbb{I}$ and $\mathcal{Z}$ denotes the set of $\ell$ possible outcomes of the measurement. Without loss of generality, we can assume $\mathcal{Z}=[\ell]$. In this work, we focus on measurements with a finite number of outcomes,
letting $\textbf{Meas}(d,\ell)$ denote the set of all $\ell$-outcome measurements on $d$-dimensional states.  
The probability distribution for obtaining an outcome $z$ upon measuring state $\bm{\rho}$ is given by the probability density function $\mathrm{p}_z = \Tr(M_z\bm{\rho})$ for each $z \in \mathcal{Z}$.

In this study, we only consider the case of employing single-copy POVM acting on the output states for collecting classical responses.
In particular, there is a single $d$-dimensional register that can be prepared in the state $\bm{\rho}$ upon measurement request. This register is measured once, and this process is repeated $m$ times. This class of measurements refers to ``single-copy measurements'', where the number of access to the quantum states corresponds to the number of measurements performed.  Within this framework, we distinguish between two significant approaches: nonadaptive and adaptive measurements. Here we only consider the nonadaptive measurements strategy with adaptive measurements being beyond the considerations of this manuscript. Particularly, we consider $m$ independent copies of the prepared state $\bm{\rho}\in \mathcal{H}_d$ of which each is stored in a single $d$-dimensional register and needs to be measured individually. In the nonadaptive measurement strategy, we utilize a predetermined sequence of measurements, denoted as $\mathcal{M}_i \in \textbf{Meas}(d,\ell)$ for $i = 1, \cdots, m$, to extract classical outcomes from quantum states. Each $\mathcal{M}_i$ is selected prior to conducting any measurements. More precisely, the collective state $\bm{\rho}^{\otimes m}$ is measured using a tensor product of these measurements, represented as $\mathcal{M}_1 \otimes \mathcal{M}_2 \otimes \cdots \otimes \mathcal{M}_m$. It is important to note that treating the selection of each measurement $\mathcal{M}_i$ as an independent random variable effectively parallels the predetermined approach. This is because the inherent randomness of measurement selection can be integrated into the definition of the measurement itself. Consequently, despite the potential for variability in selecting $\mathcal{M}_i$, the applied linear maps remain consistently aligned with the initially chosen fixed measurements for $d$-dimensional states.

\subsection{Mutual information in measurement outcomes}
In this section, we analyze the upper bound of the mutual information $\mathrm{I}(X:\hat{X})$ between the target index $X$ and estimated index $\hat{X}$. The techniques of deriving this upper bound follow the convention of Ref.~\cite{lowe2022lower} which studies the query complexity for quantum state tomography.
According to the data processing inequality, the problem can be reduced to bound the mutual information between the target index $X$ and the single POVM outcome $\bm{a}_{jk}$ on the output states $\bm{\rho}_j^{(X)}$ related to system $\mathcal{X}$, where $\bm{\rho}_j=\Tr_{\mathcal{R}}((\bm{U}_X \otimes \mathbb{I} \ket{\bm{\psi}_j}\bra{\bm{\psi}_j} (\bm{U}_X^{\dagger} \otimes \mathbb{I})$ is the output mixed state of entangled states $\ket{\bm{\psi}_j}$ and $\bm{a}_{jk}$ is the $k$-th measurement outputs of the POVM on state $\bm{\rho}_j^{(X)}$.
Additionally, for general POVM, we can not employ the Gaussian distribution to approximate the distribution of outcomes under projective measurement $\bm{O}=\ket{\bm{o}}\bra{\bm{o}}$, since the outcomes are now a vector independent with the observable $\bm{O}$. In this regard, bounding the KL-divergence between the conditioned output distribution is intractable as in Lemma~\ref{lem:MI_Gaussian}. The following lemma enables us to bound mutual information in terms of the $\chi^2$-divergence, which is more amenable to analysis in this context.
\begin{lemma}[Lemma 4.1 in Ref.~\cite{lowe2022lower}]
	\label{lem:MI_chi_div}
	Let $X$ be an arbitrary random variable and $a\in \mathcal{Z}$ be a discrete random variable for some sample space $\mathcal{Z}$. Denote by $\mathrm{p}_{a|x}:\mathcal{Z}\to [0,1]$ the distribution of $a$ conditioned on the event $X = x$. For an
	arbitrary discrete distribution $\mathrm{p}_{a}: \mathcal{Z} \to [0,1]$, it holds that
	\begin{equation}\label{eq:MI_chi_D}
		\mathrm{I}(X;a) \le \frac{1}{\ln(2)} \mathbb{E}_{X\sim \mathrm{p}_X} D_{\chi^2}(\mathrm{p}_{a|X}||\mathrm{p}_a)= \frac{1}{\ln(2)}\left(\sum_{a\in \mathcal{Z}}\mathbb{E}_{X\sim \mathrm{p}_X} \frac{\mathrm{p}_{a|X}(a)^2}{\mathrm{p}_a(a)} - 1  \right).
	\end{equation}
\end{lemma}

With this lemma, we can obtain the upper bound of the mutual information $\mathrm{I}(X;a)$ and further $\mathrm{I}(X;\hat{X})$ with data processing inequality. We encapsulated the above argument in the following lemma.
\begin{lemma}\label{lem:MI_POVM}
	Let $\mathcal{S}=\{\ket{\bm{\psi}_j}\in \mathcal{H}_{\mathcal{X}}\otimes \mathcal{H}_{\mathcal{R}}\}_{j=1}^N$ be the  entangled training states prepared according to the Construction Rule~\ref{construct:1} and $\sigma_j=\Tr_{\mathcal{R}}(\ket{\bm{\psi}_j}\bra{\bm{\psi}_j})$ be the related mixed states on the system $\mathcal{X}$. Let $\bm{U}_X$ be a Haar unitary in the local $2\varepsilon$-packing $\mathcal{P}_{2\varepsilon}^{(\gamma\varepsilon)}$ and $\bm{a}_i=(\bm{a}_{i1},\cdots,\bm{a}_{im})$ with $\bm{a}_{ij} \in \mathcal{Z}$ being the outcome of the $j$-th measurement $\mathcal{M} \in \Xi(d)$ performed
	on the evolved state $\bm{U}_X\bm{\rho}_i \bm{U}_{X}^{\dagger}$ such that $\mathrm{p}_{a_{ij}|\bm{U}_X} = \diag(\mathcal{M}(\bm{U}_X\bm{\sigma}_i \bm{U}_{X}^{\dagger}))$ for every $\bm{U}_X \in \mathbb{S}\mathbb{U}(d)$. The learning procedure employs the output $\tilde{\bm{a}}=(\bm{a}_1, \cdots, \bm{a}_n)$ to obtain the estimated $\hat{X}$. It holds that
	\begin{equation}
		\mathbb{E}_{\mathcal{S}}\mathrm{I}(X;\hat{X})\le  \frac{N}{r} \min\left\{4m, \frac{6m\ell}{d-1}, r\log(d) \right\}
	\end{equation}
	where $\ell=|\mathcal{Z}|$ refers to the number of measurement outcomes. 
\end{lemma}
The proof of Lemma~\ref{lem:MI_POVM} employs the following lemma whose proof is deferred to Supplementary Note~\ref{append_subsec:MI_POVM}.

\begin{lemma}\label{lem:E_Tr_M_rho}
	Fix an $\varepsilon\in (0,1)$ and a positive integer $d\ge 1$. Let $\bm{U}_X\in \mathbb{S}\mathbb{U}(d)$ be a Haar-random unitary operator in the $2\varepsilon$-packing $\mathcal{P}_{2\varepsilon}^{(\gamma\varepsilon)}$, $M\in \Psd(d)$ be a postive semidefinite operator such that $M \preceq \mathbb{I}_d$, $\bm{U}_X \bm{\rho} \bm{U}_X^{\dagger}$ be the output state with $\bm{\sigma}$ being a mixed state with rank $r$ and $\bm{U}_X$ chosen from the $2\varepsilon$-packing and $\omega:=\Tr(M)/d$. The index $\hat{X}$ is estimated according to the outcomes of $\mathcal{M}$. It holds that 
	\begin{equation}
		\mathbb{E}_{\bm{U}_X\sim \haar}\Tr(M\bm{U}_X \bm{\rho} \bm{U}_X^{\dagger}) = \omega
	\end{equation}
	and 
	\begin{equation}
		\mathbb{E}_{\bm{\rho}}\mathbb{E}_{\bm{U}_X\sim \haar}\Tr(M\bm{U}_X \bm{\rho} \bm{U}_X^{\dagger})^2 \le \omega^2\min\left\{1+\frac{2}{r}, \frac{d^2r\omega+2d}{r\omega(d^2-1)} \right\}.
	\end{equation}
\end{lemma}

\begin{proof}[Proof of Lemma~\ref{lem:MI_POVM}]
	Similar to the proof of Lemma~\ref{lem:sum_MI_upper_bound}, employing the data processing inequality on the Markov chain $X\to \bm{U}_x \to \tilde{\bm{a}} \to \hat{X}$ yields
	\begin{align}\label{eq:povm_mi_sum_single}
		\mathrm{I}(X;\hat{X}) & \le \mathrm{I}(X; \bm{a}_{11}, \cdots, \bm{a}_{1m}, \cdots, \bm{a}_{N1}, \cdots ,\bm{a}_{Nm}) 
		\nonumber \\
		& \le \sum_{i=1}^N\sum_{j=1}^m \mathrm{I}(X;\bm{a}_{ij}),
	\end{align}
	where the second inequality follows Property~\ref{property:subadd_MI} that the subadditivity of mutual information holds when given $X$ and $\bm{\rho}_j$ the outcomes $\{\bm{a}_{ij}\}_{i,j=1}^{N,m}$ are independent. Now we only need to focus on the mutual information $\mathrm{I}(X;a)$ between target index $x$ and single measurement output $a\in \mathcal{Z}$, which is upper bounded by the $\chi^2$-divergence between the PDFs $\mathrm{p}_{a}$ and $\mathrm{p}_{a|X}$. 
	Recalling the definition of $\chi^2$-divergence in Definition~\ref{def:chi_div} and denoting $\omega(a)=\Tr(M_a)/d$, we have 
	\begin{align}\label{eq:bound_chi_single}
		\mathbb{E}_{\mathcal{S}}\mathbb{E}_{\bm{U}_x \sim \haar} \mathrm{D}_{\chi^2}(\mathrm{p}_{a|u}\|\mathrm{p}_a) & = \sum_{a\in \mathcal{Z}} \mathbb{E}_{\mathcal{S}} \mathbb{E}_{\bm{U}_x \sim \haar}\frac{\mathrm{p}_{a|u}(a)^2}{\mathrm{p}_{a}(a)}-1
		\nonumber \\
		& = \sum_{a\in \mathcal{Z}} \mathbb{E}_{\mathcal{S}} \frac{\mathbb{E}_{\bm{U}_x \sim \haar}\Tr(M_a \bm{U}_x \bm{\sigma} \bm{U}_x^{\dagger})^2}{\mathbb{E}_{\bm{U}_x\sim \haar}\Tr(M_a \bm{U}_x \bm{\sigma} \bm{U}_x^{\dagger})}-1
		\nonumber \\
		& = \sum_{a\in \mathcal{Z}}  \frac{\mathbb{E}_{\mathcal{S}}\mathbb{E}_{\bm{U}_x \sim \haar}\Tr(M_a \bm{U}_x \bm{\sigma} \bm{U}_x^{\dagger})^2}{\omega(a)}-1
		\nonumber \\
		& \le \sum_{a\in \mathcal{Z}} \omega(a)\min\left\{1+\frac{2}{r}, \frac{d^2r\omega(a)+2d}{\omega(a)\cdot(d^2-1)r} \right\} -1
		\nonumber \\
		& = \min\left\{\frac{2}{r}, \frac{d^2r+2d\ell}{(d^2-1)r}-1 \right\}
		\nonumber \\
		& = \frac{1}{r}\min\left\{2, \frac{r+2d\ell}{d^2-1} \right\}
		\nonumber \\
		& \le \frac{1}{r}\min\left\{2, \frac{3\ell}{d-1} \right\}
	\end{align}
	where the expectation over $\mathcal{S}$ means taking the expectation over the randomly sampled states $\bm{\sigma}$ according to the Construction Rule~\ref{construct:1}, the second equality follows the form of the conditional probabilities $\mathrm{p}_{a|\bm{U}_x}(a)=\Tr(M_a \bm{U}_x \bm{\sigma} \bm{U}_x^{\dagger})$ with $M_a$ being the POVM corresponding to the measurement $\mathcal{M}$, and the PDF in the denominator $\mathrm{p}_a(a)=\mathbb{E}_{\bm{U}_x\sim \haar}\Tr(M_a \bm{U}_x \bm{\sigma} \bm{U}_x^{\dagger})$. Let $\omega(a)=\Tr(M_a)/d$ for all $a\in \mathcal{Z}$, the first inequality follows Lemma~\ref{lem:E_Tr_M_rho}, and the fourth equality follows the fact that $\sum_{a\in \mathcal{Z}}\omega(a)=\sum_{a\in \mathcal{Z}}\Tr(M_a)/d=1$ as $\sum_{a\in \mathcal{Z}} M_a = \mathbb{I}_d$, the last inequality follows that $r\le d\ell$ and $d/(d^2-1)\le 1/(d-1)$.

	In this regard, in conjunction with Eqn.~(\ref{eq:povm_mi_sum_single}), Eqn.~(\ref{eq:bound_chi_single}) and Lemma~\ref{lem:MI_chi_div}, we have 
	\begin{align}
		\mathbb{E}_{\mathcal{S}}\mathrm{I}(X;\hat{X}) & \le \mathbb{E}_{\mathcal{S}}\sum_{i=1}^N\sum_{j=1}^m \mathrm{I}(X;a_{ij})
		\nonumber \\
		& \le \sum_{i=1}^N\sum_{j=1}^m \frac{1}{(\ln{2})} \mathbb{E}_{\mathcal{S}}\mathbb{E}_{\bm{U}_X \sim \haar} \mathrm{D}_{\chi^2}(\mathrm{p}_{a_{ij}|\bm{U}_x}\|\mathrm{p}_{a_{ij}})
		\nonumber \\
		& \le \frac{Nm}{r(\ln{2})} \min\left\{2, \frac{3\ell}{d-1} \right\}
		\nonumber \\
		& \le \frac{Nm}{r} \min\left\{4, \frac{6\ell}{d-1} \right\}.
	\end{align}
	On the other hand, from the observation that the expectation value $\Tr((\bm{U}_X^{\dagger}O \bm{U}_X\otimes \mathbb{I})\ket{\bm{\psi}_j}\bra{\bm{\psi}_j})$ is the sufficient statistics of the output states $\bm{U}_X\otimes \mathbb{I}\ket{\bm{\psi}_j}$ in terms of $X$, leading to the Markov chain $X\to \bm{a}_j \to \Tr((\bm{U}_X^{\dagger}\bm{O} \bm{U}_X\otimes \mathbb{I})\ket{\bm{\psi}_j}\bra{\bm{\psi}_j}) \to \hat{X}$. In this regard, following Lemma~\ref{lem:mi_u_state_2}, we have $\mathrm{I}(X;\hat{X})\le Nr\log(d)$.
	This completes the proof.
\end{proof}

In conjunction with Lemma~\ref{lem:Fano}, Lemma~\ref{lem:packing_observable_lem}, and Lemma~\ref{lem:MI_POVM}, we finally come to the lower bound of prediction error for employing an $\ell$-outcome POVM, given in the following theorem.
\begin{theorem}[Formal statement of Theorem~2 in the maintext]
	\label{thm:formal_finite_measurement_POVM}
	Let $\{ \mathrm{f}_{\bm{U}_{x}}\}_{x \in \mathcal{X}_{2\varepsilon}^{(2\gamma\varepsilon)}}$ be a local $2\varepsilon$-packing with the maximal distance being $\gamma\varepsilon$ of the function class $\mathcal{F}$ in the $\varrho$-metric. Assume that the index $X$ corresponding to the target function $\mathrm{f}_{\bm{U}_{X}}$ is uniformly sampled from the set $\mathcal{X}_{2\varepsilon}^{(2\gamma\varepsilon)}$. Conditional on $X=x$, we obtain the training set $\mathcal{S}=\{\ket{\bm{\psi}_j}, \bm{a}_j \}_{j=1}^N$ where $\ket{\bm{\psi}_j}$ is the random entangled state of Schmidt rank $r$ sampled from the distribution described in Construction Rule \ref{construct:1}, and $\bm{a}_j=(\bm{a}_{j1}, \cdots, \bm{a}_{jm})$ is the outputs after measuring the mixed states $\bm{U}_X\Tr_{\mathcal{R}}(\ket{\bm{\psi}_j}\bra{\bm{\psi}_j} \bm{U}_X^{\dagger}$ for $m$ times with $\ell$-outcome POVM $\mathcal{M}$. Then the averaged risk function defined in Eqn.~(\ref{eq:minimax_risk}) is lower bounded by
	\begin{equation}\label{eq:lower_bound_POVM}
		\mathbb{E}_{\bm{U}}\mathbb{E}_{\mathcal{S}} \mathrm{R}_{\bm{U}}(\bm{V}_{\mathcal{S}}) \ge \frac{\tilde{\varepsilon}^2}{8d(d+1)} \left(1- \frac{N\cdot \min\{4m/r, 6m\ell/r(d-1), r\log(d))\}+16\log 2 }{\log(|\mathcal{X}_{2\varepsilon}(\bm{O})|)}  \right),
	\end{equation}
	where $\tilde{\varepsilon}=2\sqrt{2d(d+1)}{\varepsilon}$ and the expectation is taken over all target unitaries $\bm{U}$.
\end{theorem}
\begin{proof}
	[Proof of Theorem~\ref{thm:formal_finite_measurement_POVM}]. The lower bound in Eqn.~(\ref{eq:lower_bound_POVM}) can be obtained from the conjunction with Lemma~\ref{lem:est_test}, Lemma~\ref{lem:Fano} and Lemma~\ref{lem:MI_POVM}, as done in the proof of Theorem~\ref{thm:formal_finite_measurement}. 
\end{proof}
The term $\log(|\mathcal{X}_{2\varepsilon}(\bm{O}))|)$ measures the model complexity, which only depends on the employed observable in the target function and $\varepsilon$. While the $\varepsilon$-packing for the case of projective measurement $\bm{O}=\ket{\bm{o}}\bra{\bm{o}}$ is easy to construct as $\bm{U}^{\dagger}\bm{O}\bm{U}$ refers to pure state, it is hard to establish the lower bound of $\varepsilon$-packing for the case of general observable, which could remain as independent interest for future studies.

\subsection{Proof of Lemma~\ref{lem:E_Tr_M_rho}}\label{append_subsec:MI_POVM}
Here we present the proof of Lemma~\ref{lem:E_Tr_M_rho}.
\begin{proof}[Proof of Lemma~\ref{lem:E_Tr_M_rho}]
	Employing Lemma~\ref{lem:Tr(psipsi)} to calculate the Haar expectation over $\bm{U}_x$ yields
	\begin{equation}
		\mathbb{E}_{\bm{U}_x\sim \haar} \Tr(M\bm{U}_x \bm{\rho} \bm{U}_x^{\dagger}) = \frac{\Tr(M)}{d}.
	\end{equation}
	For the second expectation in the lemma, by employing Lemma~\ref{lem:Tr(psipsi)Tr(psipsi)} we have 
	\begin{align}
		\mathbb{E}_{\bm{U}_x\sim \haar} \Tr(M\bm{U}_x \bm{\rho} \bm{U}_x^{\dagger})^2 & = \frac{\Tr(M)^2(d-\Tr(\bm{\rho}^2)) + \Tr(M^2)(d\Tr(\bm{\rho}^2)-1)}{d(d^2-1)} 
		\nonumber \\
		& \le \frac{d^2\omega^2(d-\Tr(\bm{\rho}^2)) + \min\{d^2\omega^2, d\omega\}(d\Tr(\bm{\rho}^2)-1)}{d(d^2-1)} 
		\nonumber \\
		& = \frac{\omega^2}{d+1}\min\left\{d(1+\Tr(\bm{\rho}^2)), \frac{d^2\omega+d\Tr(\bm{\rho}^2)-1-d\omega\Tr(\bm{\rho}^2)}{\omega(d-1)} \right\}
		\nonumber \\
		& \le \frac{d\omega^2}{d+1}\min\left\{1+\Tr(\bm{\rho}^2), \frac{d\omega+\Tr(\bm{\rho}^2)}{\omega(d-1)} \right\},
	\end{align}
	where the first inequality follows the property $0 \preceq M \preceq \mathbb{I}_d$ such that $\Tr(M^2)\le \Tr(M)^2=d^2\omega^2$ and $\Tr(M^2)\le \Tr(M)=d\omega$, the second inequality follows the nonnegativity of $d\omega\Tr(\bm{\rho}^2)+1$. We further calculate the integration in terms of the training state according to the Construction Rule yields
	\begin{align}
		\mathbb{E}_{\bm{\rho}} \mathbb{E}_{\bm{U}_x\sim \haar} \Tr(M\bm{U}_x \bm{\rho} \bm{U}_x^{\dagger})^2 & \le \mathbb{E}_{\bm{\rho}} \frac{d\omega^2}{d+1}\min\left\{1+\Tr(\bm{\rho}^2), \frac{d\omega+\Tr(\bm{\rho}^2)}{\omega(d-1)} \right\}
		\nonumber \\
		& = \frac{d\omega^2}{d+1}\min\left\{1+\frac{2}{r}, \frac{dr\omega+2}{\omega(d-1)r} \right\}
		\nonumber \\
		& = \omega^2 \min\left\{1+\frac{2}{r}, \frac{d^2r\omega+2d}{r\omega(d^2-1)} \right\}
	\end{align}
	where the second equality employs $\Tr(\bm{\rho}^2)=2/(r+1)<2/r$ which follows the derivation in Eqn.~(\ref{eq:single_MI}) and the last inequality employs $d/(d+1)<1$. This completes the proof.
\end{proof}

\section{Numerical Simulations for the quantum NFL theorem}\label{app_sec:Numerical_NFL}
In this section, we provide the simulation details omitted in the main text and more numerical results.

\subsection{Details of numerical simulations}

\noindent\textbf{Constructions of target unitaries.} For the numerical simulations, we consider the set of all possible target unitaries $\mathcal{U}=\{\bm{U} \in \mathbb{S}\mathbb{U}(d)| \bm{U}_{1j}=e^{i\theta_j}, \theta_j \in \mathbb{R}, j \in[d]\}$, where the first column vector $\bm{U}_{1}$ is constrained to the form of $e^{i\theta_j}\ket{\bm{e}_j}$ with $\theta_j\in \mathbb{R}$ being any real number and $\ket{\bm{e}_j}$ referring to the computational basis, and the remaining column vectors could be arbitrary with the only requirement of being orthogonal with the first column vector. The fixed observable is $\bm{O}=\ket{\bm{0}}\bra{\bm{0}}$. In this regard, the substantial set of learning target operators refers to $\mathcal{U}_{O}=\{\bm{U}^{\dagger}\bm{O}\bm{U}=\ket{\bm{e}_j}\bra{\bm{e}_j}: j \in [d]\}$, where $\ket{\bm{e}_j}$ is the computation basis with the $j$-th entry being $1$ and the other being $0$. 
Remarkably, with the aim of constructing such a set of target operators $\mathcal{U}_{O}$, the observable could be any projective measurement $\ket{\bm{o}}\bra{\bm{o}}$ with $\ket{\bm{o}} \in \mathcal{H}_d$ being arbitrary state vectors when restricting the first column vector of the target unitary $\bm{U}$ on the form of $\bm{U}_1=\ket{\bm{o}}$. The target unitary $\bm{U}^*$ is uniformly sampled from the set $\mathcal{U}$. Denote the corresponding target operator as $\bm{U}^{*\dagger}\bm{O}\bm{U}^*=\ket{\bm{e}_{k^*}}\bra{\bm{e}_{k^*}}$, learning the target unitary under the fixed observable $\bm{U}^{*\dagger}\bm{O}\bm{U}^*$ is equivalent to identifying the unknown index $k^*\in[d]$.

\noindent \textit{Remark.} The construction of orthogonal operators $\bm{U}^{\dagger}\bm{O}\bm{U}$ leads to the best distinguishability. Additionally, any discretized function class $\mathcal{F}$ with $\varrho(\bm{U}_i, \bm{U}_j)\ge \varepsilon$ for any $\bm{U}_i, \bm{U}_j \in \mathcal{F}$ could be chosen as the set of target unitaries. A small $\varepsilon$ leads to worse distinguishability, and hence a large number of measurements is required for correctly estimating the target unitary.

\medskip

\noindent	\textbf{Constructions of training data.} We first recall that the entangled training states $\{\ket{\bm{\psi}_j}\}_{j=1}^N$ defined in Eqn.~(\ref{eq:entangled_state}) have the form of $\ket{\bm{\psi}}=\sum_{k=1}^r \sqrt{c_{jk}}\ket{\bm{\xi}_{jk}}_{\mathcal{X}}\ket{\bm{\zeta}_{jk}}_{\mathcal{R}}$ where $\sum_{k=1}^r c_{jk}=1$, and $\mathcal{X}$, $\mathcal{R}$ refer to the quantum system the target unitary acts on and the reference system, respectively. In this regard, we only need to consider the quantum states $\ket{\bm{\xi}_{jk}}$ related to the quantum system $\mathcal{X}$ as the target operators only acts on the partial trace of the entangled state, i.e., $\bm{\sigma}_j:=\operatorname{Tr}_{\mathcal{R}}\left[\ket{\bm{\psi}_j}\bra{\bm{\psi}_j}\right]=\sum_{k=1}^r c_{j, k}\ket{\bm{\xi}_{j, k}}\bra{\bm{\xi}_{j, k}}_{\mathcal{X}}$. Particularly, we consider the set of mixed states 
\begin{equation}
	\widetilde{\mathcal{S}}=\left\{\bm{\sigma} =\sum_{k=1}^r c_{k} \ket{\bm{e}_{\pi(k)}}\bra{\bm{e}_{\pi(k)}}: \pi\in S_d, \ket{c}=(\sqrt{c_{1}}, \cdots, \sqrt{c_{r}})^{\top}\in \mathbb{S}\mathbb{U}(r), \ket{\bm{e}_{\pi(k)}} \in \mathcal{H}_{\mathcal{X}} \right\},
\end{equation}
where $\ket{\bm{c}}\in \mathbb{S}\mathbb{U}(r)$ follows the Haar distribution and is independent with $\ket{\bm{e}_{\pi(k)}}$,  $\pi \in S_d$ refers to the permutation operator rearranging the order of the set $\{1, \cdots, d\}$ as $\{\pi(1), \cdots, \pi(d)\}$, $\mathcal{H}_{\mathcal{X}}$ is the Hilbert space $\mathcal{H}_d$ related to the quantum system $\mathcal{X}$. The $N$ training mixed states $\{\bm{\sigma}_j\}_{j=1}^N$ are uniformly and independently sampled from the set $\widetilde{\mathcal{S}}$. Remarkably, for the case of $r=1$, the ensemble consisting of the mixed states set $\widetilde{\mathcal{S}}$ equipped with the uniform distribution forms $2$-design for which the second moment is the same as that of the random states following the Haar distribution, according with Construction Rule \ref{construct:1}. Furthermore, the measurement outputs $(\bm{o}_{j1}, \cdots, \bm{o}_{jm})$ are collected by measuring the output states $\bm{U}^{*}\bm{\sigma}_j(\bm{U}^{*})^{\dagger}$ for $m$ times with the observable $\bm{O}$. We use the statistical mean value of the measurement outputs $\bm{o}_j = \sum_{k=1}^m \bm{o}_{jk}/m$ as the response of the training states $\ket{\bm{\psi}_j}$ (or equivalently $\bm{\sigma}_j$).

\medskip

\noindent	\textbf{Quantum dynamics learning process.} As discussed above, learning the target unitary $\bm{U}^*$ under the projective measurement $\bm{O}$ is equivalent to identifying the target index $k^*$ corresponding to the target unitary. In this end, for given input states $\{\ket{\bm{\psi}_j}\}_{j=1}^N$, we collect the measurement outputs $\{(\bm{o}_1^{(k)}, \cdots, \bm{o}_{N}^{(k)})\}_{k=1}^d$ over all possible index $k\in[d]$, where $\bm{o}_j^{(k)}$ is the statistical estimation of $\Tr(\bm{\sigma}_j\ket{\bm{e}_k}\bra{\bm{e}_k})$. The estimated index $\hat{k}$ is determined by solving the following minimization problem
\begin{equation}\label{eq:training_process}
	\hat{k} = \arg\min_{k\in [d]} \sum_{j=1}^N \left(\bm{o}_j^{(k)}-\bm{o}_j\right)^2.
\end{equation}
This minimization of Eqn.~(\ref{eq:training_process}) refers to the training process, which can be accomplished through direct calculation. Finally, the corresponding target operator $\bm{V}_\mathcal{S}^{\dagger}\bm{O}\bm{V}_{\mathcal{S}}=\ket{\bm{e}_{\hat{k}}}\bra{\bm{e}_{\hat{k}}}$ is exploited to predict the output of any unseen input states $\ket{\bm{\phi}}\in \mathcal{H}_d$. Notably, we only consider the relation between the prediction error and the resource used for collecting the training data, including the training data size $N$, the Schmidt rank $r$ of the entangled state, and the number of measurements $m$. We do not consider the measurement cost during the training process.

\begin{figure}  
	\centering
	\includegraphics[width=176mm]{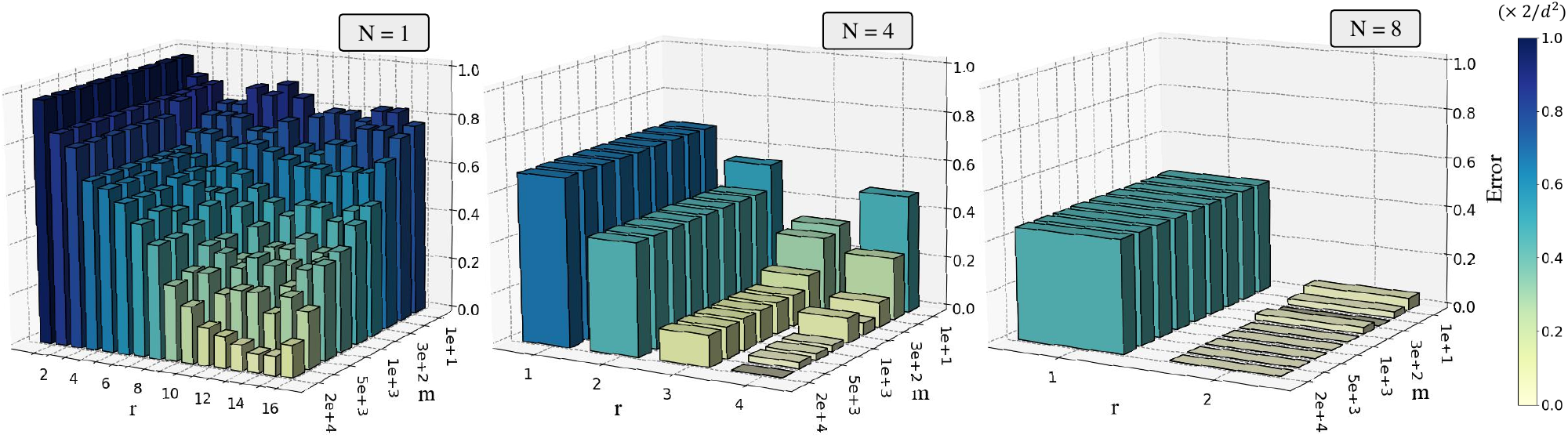}
	\caption{{\textbf{Simulation results of quantum NFL theorem with orthogonal training states.}  The averaged prediction error with a varied number of measurements $m$ and Schmidt rank $r$ when $N=1$, $N=4$,  and $N=8$. The z-axis refers to the averaged prediction error defined in Eqn.~(\ref{eq:learning_model}). The label `($\times 2/d^2$)' refers that the plotted prediction error is normalized by a multiplier factor $2/d^2$.}   }
	\label{appendix_fig:numerical_orth_m_r}
\end{figure}

\begin{figure} 
	\centering
	\includegraphics[width=\textwidth]{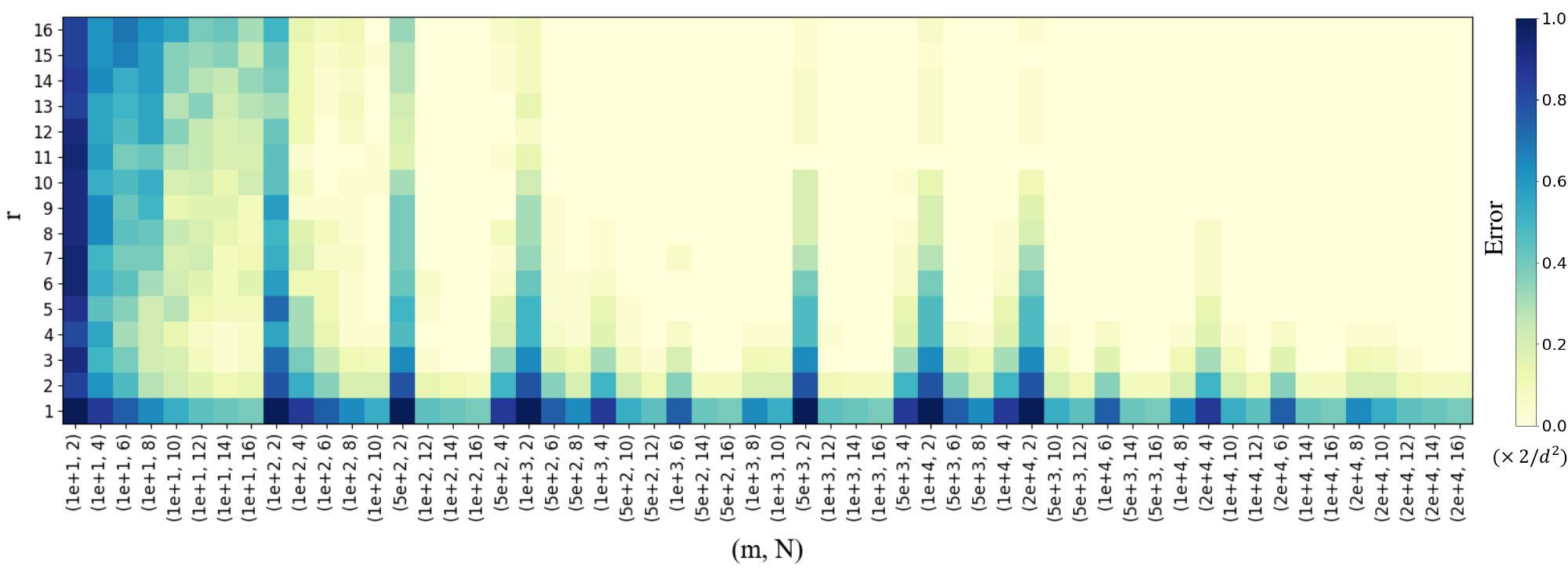}
	\caption{{\textbf{Simulation results of quantum NFL theorem with independent training states.}  The averaged prediction error with a varied number of measurements $m$, size of training data $N$, and Schmidt rank $r$. The x-axis is arranged from left to right by the magnitude of the product of $m$ and $N$.  The label `($\times 2/d^2$)' refers that the plotted prediction error is normalized by a multiplier factor $2/d^2$. } }
	\label{appendix_fig:Numerical_Varying_mN_r}
\end{figure}

\medskip

\noindent	\textbf{Hyper-parameters setting.} The hyper-parameters setting related to the task of quantum dynamics learning is as follows. The number of qubits related to the quantum system $\mathcal{X}$ is set as $4$, and hence the corresponding dimension of Hilbert space $\mathcal{H}_{\mathcal{X}}$ is $d=16$. The Schmidt rank $r$ and the size of training data $N$ are set as $\{1,2, \cdots, 16\}$. We set the number of measurements as 
\begin{equation}
	m \in \{10, 100, 300, 500, 1000, 2000, 5000, 10000, 20000\}.
\end{equation}
The averaged risk is recorded after learning $4$ random unitaries for $10$ random training data.

\subsection{More numerical results}

\noindent	\textbf{Simulation results for orthogonal training states.} Supplementary Figure~\ref{appendix_fig:numerical_orth_m_r} plots the prediction error when the sampled training states are orthogonal, i.e., $\braket{\bm{\psi}_i|\bm{\psi}_j}=0$ for any $i\ne j$. In this case, the product of the Schmidt rank $r$ and the size of training data $N$ obeys $r\times N \le 2^n$ as the number of mutually orthogonal states in the Hilbert space $\mathcal{H}_{\mathcal{X}}$ is less than the dimension of this space $d=2^n$. The simulation results show that the prediction error vanishes when the product of the Schmidt rank and the size of the training data $rN$ equals the dimension $d$. Additionally, for the case of a small number of measurements $m$, increasing the Schmidt rank $r$ may increase the prediction error, indicating the transition role of entangled data. These phenomenons accord with Theorem~\ref{thm:formal_finite_measurement}.

\medskip

\noindent	\textbf{Simulation results for independently sampled training states.} Supplementary Figure~\ref{appendix_fig:Numerical_Varying_mN_r} plots the prediction error with a varied number of measurements $m$, size of training data $N$, and Schmidt rank $r$. The simulation result shows that for the case of the identical magnitude of the product of the measurement times and the training data size (e.g., $(m,N) \in {(100, 10), (500, 2)}$), the tuple of a larger amount of training data and a smaller number of measurements in each training state $(m, N)=(100, 10)$ can achieve a smaller prediction error. This indicates that when the access times to the unknown target unitary is limited, the measurement outcomes should be collected from more different training states rather than from a few quantum states with a large number of measurements. These phenomenons accord with Theorem~\ref{thm:formal_finite_measurement}.

Moreover, in order to show various performance in the regime of $r<\sqrt{m/c_1n}$ and $r\ge \sqrt{m/c_1n}$, we record the prediction error of Schmidt rank taking $r$ with a varying number of measurements $m$ in Fig~\ref{fig:fig1}. In particular, when the number of measurements is small with $m=10$, the prediction error for the case of $r=2$ achieves the best performance over other highly entangled data for both cases of training data size being $N=2$ and $N=8$. On the other hand, when the number of measurements increases to $m=100$, the entangled data with $r=2$ reaches the best prediction performance in itself and keeps invariant with increasing $m$ but has the worst performance than other highly entangled data. This means $m=100$ with $r=2$ lies in the region of $r<\sqrt{m/c_1n}$ where the number of measurements is enough for extracting all information from the entangled states. In this case, increasing the entanglement degree $r$ helps to obtain more information about the target unitary and hence decreases the prediction error. Notably, these numerical results echo with our theoretical results.

\medskip
\begin{figure}
	\centering
	\includegraphics[width=176mm]{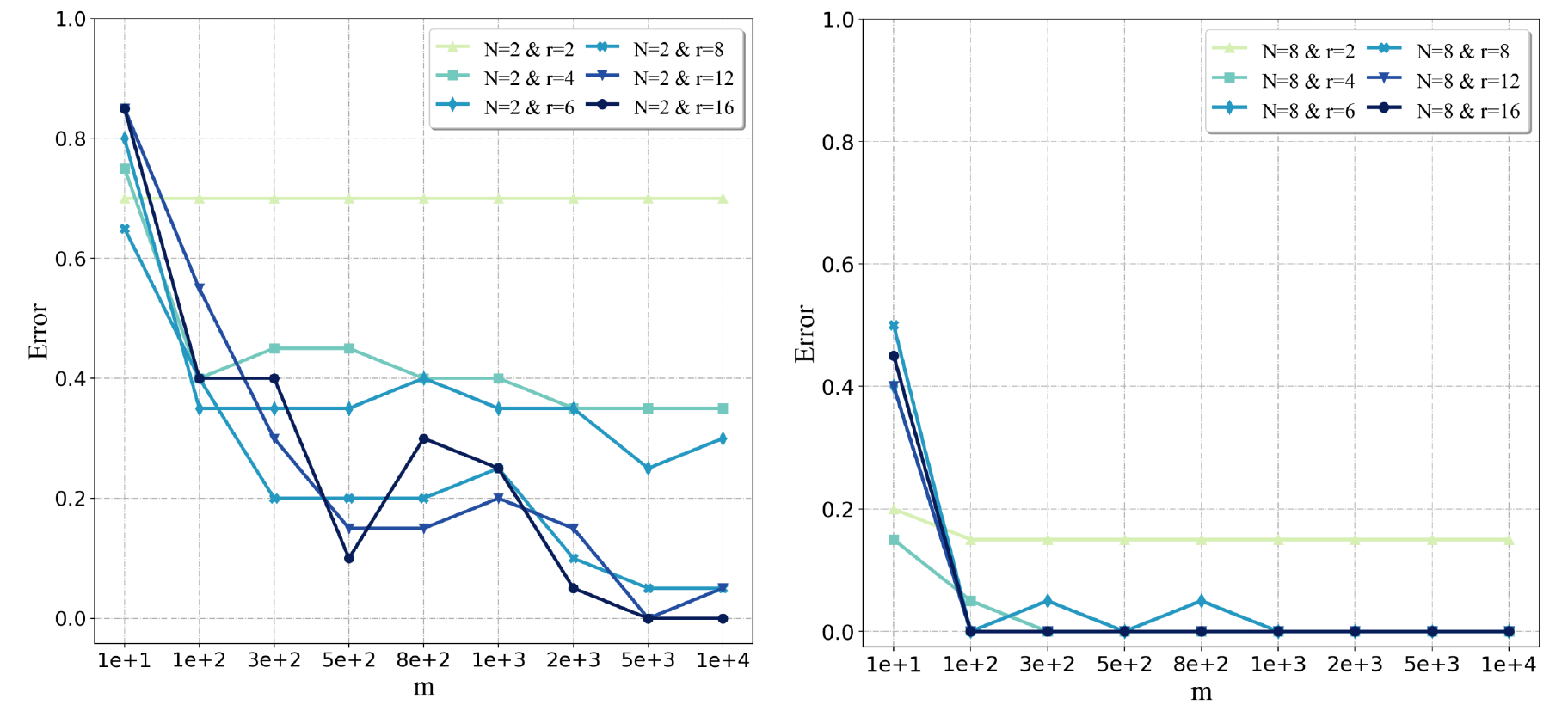}
	\caption{\small{\textbf{Prediction error with various Schmidt rank and number of measurements.} The number of Schmidt rank and the number of measurements are set as $r\in \{2,4,6,8,12,16 \}$ and $m \in \{10, 100, 300, 500, 800, 1000, 2000,$ 
			$5000, 10000\}$, respectively.  }}
	\label{fig:fig1}
\end{figure}

\begin{figure} 
	\centering
	\includegraphics[width=176mm]{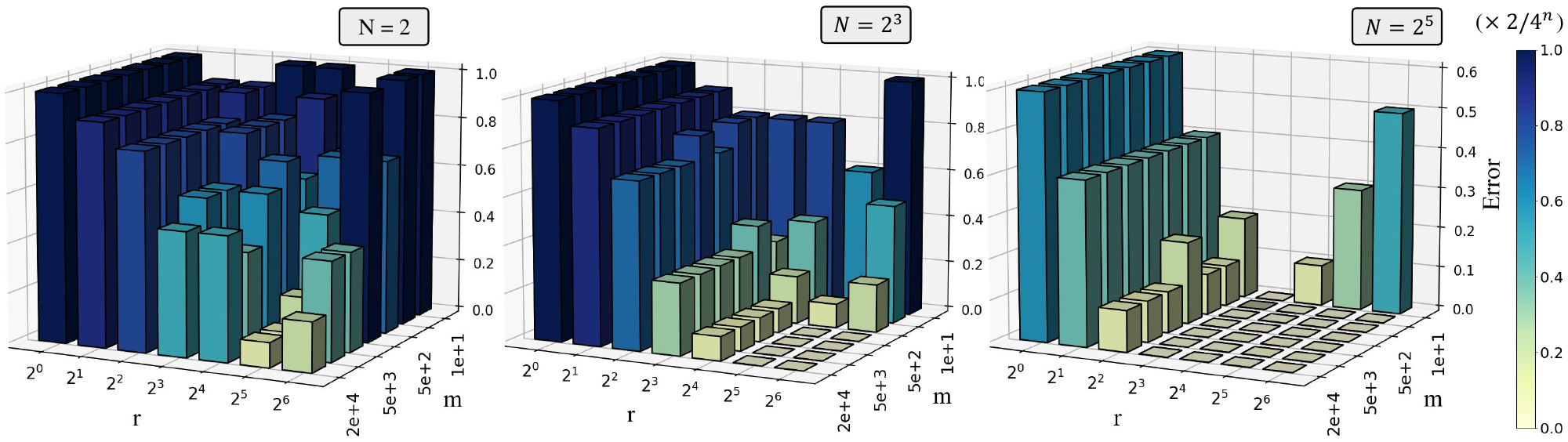}
	\caption{\small{\textbf{Simulation results of quantum NFL theorem for $6$-qubits.} The averaged prediction error with a varied number of measurements $m$ and Schmidt rank $r$ when $N=2$, $N=2^3$, and $N=2^5$. The z-axis refers to the averaged prediction error defined in Eqn.~(\ref{eq:learning_model}). The label `($\times 2/4^n$)' refers that the plotted prediction error is normalized by a multiplier factor $2/4^n$. }}
	\label{fig:6-qubits}
\end{figure}

\begin{figure} 
	\centering
	\includegraphics[width=176mm]{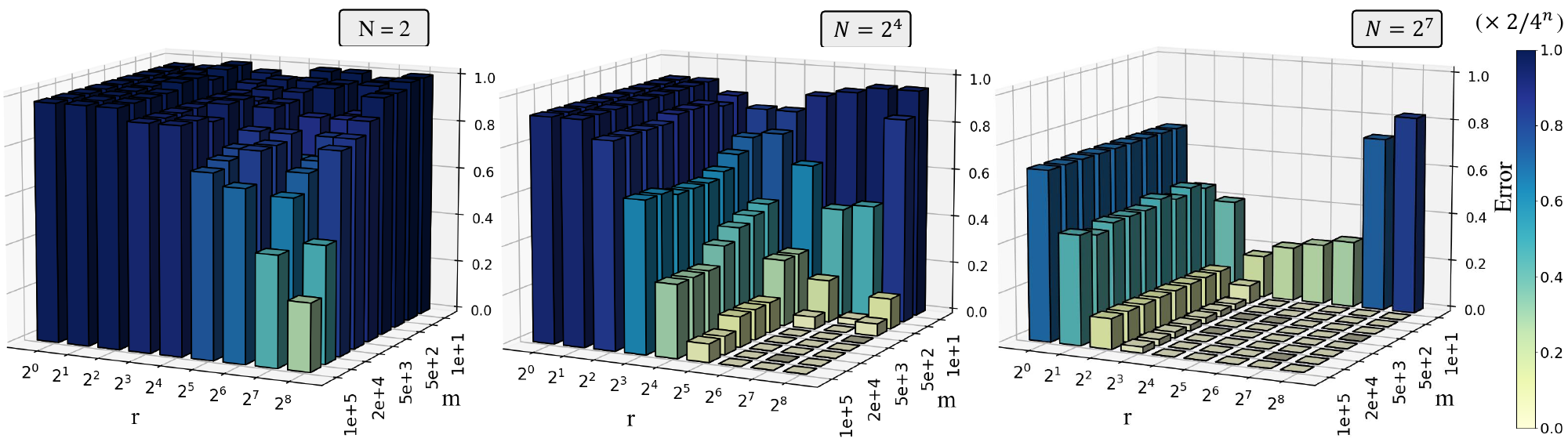}
	\caption{\small{\textbf{Simulation results of quantum NFL theorem for $8$-qubits.} The notations are identical to those in Supplementary Figure~\ref{fig:6-qubits}. }}
	\label{fig:8-qubits}
\end{figure}

\noindent  \textbf{Large scale simulation results.} To further verify the achieved results in Theorem~\ref{thm:formal_finite_measurement}, we conduct numerical simulations for the number of qubits $n=6$ and $n=8$. The setting of other hyperparameters is as follows. The Schmit rank $r$ and the training data size $N$ are set as $\{2^0, 2^1, \cdots, 2^6\}$ and $\{2^0, 2^1, \cdots, 2^8\}$ for the case of $n=6$ and $n=8$ respectively. The number of measurements for $n=6$ and $n=8$ is set as 
\begin{align}
	& m\in \{10, 100, 500, 1000, 5000, 10000, 20000\} \nonumber \\
	\quad \mbox{and} \quad & m\in \{10, 100, 500, 1000, 5000, 10000, 20000, 50000, 100000\}.
\end{align}
The simulation results for $n=6$ and $n=8$ are depicted in Supplementary Figure~\ref{fig:6-qubits} and Supplementary Figure~\ref{fig:8-qubits}, respectively.

The simulation results show that the transition role of entangled data still occurs for large quantum systems. Particularly, for a small number of measurements, increasing the Schmidt rank $r$ of entangled data could first decrease the prediction error, and then increase the prediction error once $r$ surpasses some critical point. When the number of measurements $m$ is sufficiently large, increasing the Schmidt rank $r$ constantly decreases the prediction error. For instance, we can see from the results of $8$-qubit and training data size $N=2^7$ that for a small number of measurements $m=10$, increasing the Schmidt rank $r$ will first decrease the prediction error in the regime of $r\le \sqrt{m/c_1 n}$ and then increase the prediction error once the Schmidt rank $r$ surpass a critical point such that $r\ge \sqrt{m/c_1 n}$.
These phenomenons accord with the theoretical results of Theorem~1 in the original manuscript.

\medskip

\noindent	\textbf{Simulation results for the training error.} To verify that the training error scales as $\mathcal{O}(1/2^n)$, we record the training error, given by $\sum_{j=1}^N (\bm{o}_j^{(\hat{k})}-\bm{o}_j)^2/N$ with $\hat{k}$ being the estimated index according to Eqn.~(\ref{eq:training_process}), for the varying number of qubits $n\in\{4,5,6,7,8 \}$ and varying Schmidt rank $\{2,4,8,16\}$. The training data size is set as $N=16$ and the number of measurements is set as $m\in\{10,100,1000\}$. The numerical results are presented in Supplementary Figure~\ref{fig:loss}, where the training error decreases with scaling $2^n$ as the system size $n$ increases. This indicates that the training error of scaling $\mathcal{O}(1/2^n)$ is easy to reach during the training process.

\begin{figure}[h!]
	\centering
	\includegraphics[width=176mm]{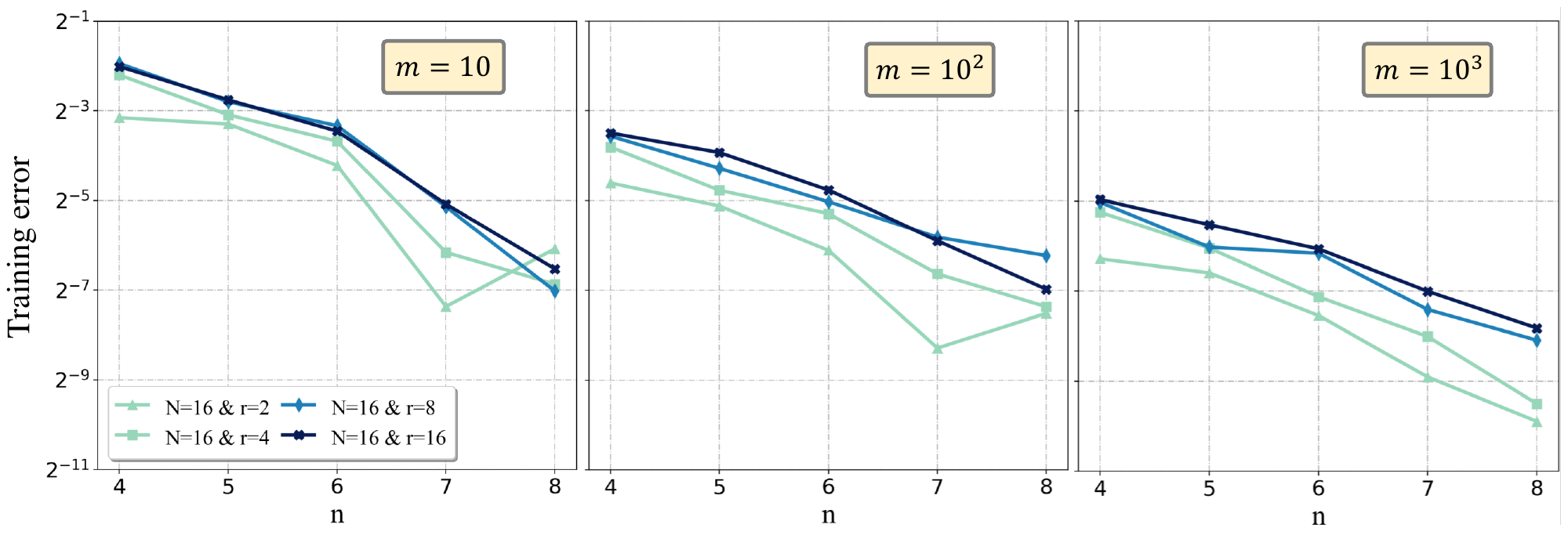}
	\caption{\small{\textbf{Training error with varying system size.} The left panel, middle panel, and right panel present the training error for the number of measurements $m=10, 100, 1000$, respectively.  }}
	\label{fig:loss}
\end{figure}

\end{document}